\documentclass[12pt,oneside,a4paper]{memoir}
\usepackage{amsmath,amsfonts,amssymb,bbm,amsthm}
\usepackage{graphicx,psfrag,color}%,subfigure
\usepackage{cite}
\usepackage[breaklinks]{hyperref}

\numberwithin{equation}{section}
\nouppercaseheads
\setsecnumdepth{subsection}

\DeclareFontFamily{OT1}{pzc}{}
\DeclareFontShape{OT1}{pzc}{m}{it}{<-> s * [1.10] pzcmi7t}{}
\DeclareMathAlphabet{\mathpzc}{OT1}{pzc}{m}{it}

\newcommand{\Or}{\mathcal{O}}
\newcommand{\Ai}{\mathrm{Ai}}

\newcommand{\const}{\mathrm{const}}
\newcommand{\Pb}{\mathbbm{P}}
\newcommand{\Id}{\mathbbm{1}}
\newcommand{\Tr}{\mathrm{Tr}}
\newcommand{\e}{\varepsilon}
\newcommand{\D}{\mathrm{d}}
\newcommand{\C}{\mathbb{C}}
\newcommand{\R}{\mathbb{R}}

\newcommand{\Z}{\mathbb{Z}}
\renewcommand{\Re}{\mathrm{Re}}
\renewcommand{\Im}{\mathrm{Im}}
\renewcommand{\cal}{\mathcal}
\DeclareMathOperator{\sgn}{sgn}
\DeclareMathOperator\supp{supp}

\newtheorem{prop}{Proposition}[chapter]
\newtheorem{thm}[prop]{Theorem}
\newtheorem{lem}[prop]{Lemma}
\newtheorem{defin}[prop]{Definition}
\newtheorem{cor}[prop]{Corollary}

\newtheorem{rem}[prop]{Remark}
\newenvironment{remark}{\begin{rem}\normalfont}{\end{rem}}
\renewcommand{\textrm}{\textnormal}
\newcommand{\hf}{\textnormal{half-flat}}
\newcommand{\hst}{\textnormal{half-stat}}
\newcommand{\stf}{\textnormal{stat-flat}}
\newcommand{\shf}{\textnormal{\,half-flat}}
\newcommand{\shst}{\textnormal{\,half-stat}}

\DeclareRobustCommand{\gobblefive}[5]{}
\newcommand*{\SkipTocEntry}{\addtocontents{toc}{\gobblefive}}

\newlength{\drop}% for my convenience

\newcommand*{\titleGM}{\thispagestyle{empty}
\begingroup% Gentle Madness
\drop = 0.1\textheight
\vspace*{-2\baselineskip}
\hbox{%
\hspace*{0\textwidth}%
\rule{1pt}{\textheight}
\hspace*{0.05\textwidth}%
\parbox[b]{1\textwidth}{
\vbox{%
{\noindent\HUGE\bfseries Reflected Brownian motions\\[0.5\baselineskip]
in the KPZ universality class}\\[2\baselineskip]
\\[4\baselineskip]
{\Large Thomas Weiss\\[0.3\baselineskip]
Patrik Ferrari\\[0.3\baselineskip]
Herbert Spohn}\par
\vspace{0.45\textheight}
{\noindent September 2016}\\[\baselineskip]
}% end of vbox
}% end of parbox
}% end of hbox
\null
\endgroup}

\begin{document}
\frontmatter

\titleGM

\SkipTocEntry\chapter{Abstract}

This book presents a detailed study of a system of interacting Brownian motions in one dimension. The interaction is point-like such that the $n$-th Brownian motion is reflected from the Brownian motion with label $n-1$. This model belongs to the Kardar-Parisi-Zhang (KPZ) universality class. In fact, because of the singular interaction, many universal properties can be established with rigor. They depend on the choice of initial conditions. Discussion addresses packed and periodic initial conditions (Chapter~\ref{secPP}), stationary initial conditions (Chapter~\ref{secPoi}), and mixtures thereof (Chapter~\ref{secMixed}). The suitably scaled spatial process will be proven to converge to an Airy process in the long time limit. A chapter on determinantal random fields and another one on Airy processes are added to have the notes self-contained. This book serves as an introduction to the KPZ universality class, illustrating the main concepts by means of a single model only. It will be of interest to readers from interacting diffusion processes and non-equilibrium statistical mechanics.

\clearpage
\SkipTocEntry\tableofcontents

\mainmatter

\chapter{Introduction}

Back in 1931 Hans Bethe diagonalized the hamiltonian of the one-dimensional Heisenberg  spin chain through what is now called the
``Bethe ansatz''~\cite{Bet31}. At that time physicists were busy with other important developments  and hardly realized the monumental step:
for the first time a strongly interacting many-body system had been ``solved exactly''.
In the 1960ies Lieb~\cite{LL63}, Yang~\cite{Yang67}, and many more ~\cite{Sut04} discovered  that other quantum systems can be handled via Bethe ansatz, which triggered a research area known as
quantum integrability. More details on the history of the Bethe ansatz can be found in~\cite{Batch}. Even with the  Bethe ansatz at one's disposal,  it is a highly non-trivial task to arrive at predictions of physical interest. This is why  efforts in quantum integrability continue even today, reenforced by the experimental realization of such chains  through an array of cold atoms~\cite{Sim11}.

On a mathematical level, quantum hamiltonians and generators of Markov processes have a comparable structure. Thus one could
imagine that the Bethe ansatz is equally useful for interacting stochastic systems with many particles. The first indication came somewhat indirectly from the Kardar-Parisi-Zhang equation~\cite{KPZ86}, for short KPZ, in one dimension. We refer to books~\cite{BaSt95,Me98}, lecture notes~\cite{Jo05,Sp06,QRev,BoGo12,BoPe14,Sp16a}, and survey
articles~\cite{HH95, Kr97, SaSp11a,FS11,CorwinRev,Ta15,QuSp15,HaKa15}.
The KPZ equation is a stochastic PDE governing the time-evolution
of a height function $h(x,t)$ at spatial point $x$ and time $t$, $h  \in \mathbb{R}$, $x \in \mathbb{R}$, $t \geq 0$. The equation reads
\begin{equation}\label{1}
\partial_t h =  \tfrac{1}{2}(\partial_x h)^2 + \tfrac{1}{2}\partial_x^2 h + W
\end{equation}
with $W(x,t)$ normalized space-time white noise. We use here units of height, space, and time such that all coupling parameters take a definite value. For a solution of \eqref{1}, the function $x\mapsto h(x,t)$ is locally like a Brownian motion, which is too singular for the nonlinearity to be well-defined as written. This difficulty was resolved through the regularity structures of Hairer~\cite{Hai11}, see also~\cite{GP15}
for a somewhat different approach. Kardar~\cite{Kar87} noted one link to quantum integrability.
He considered the moments of $\mathrm{e}^h$ and established that they are related to the $\delta$-Bose gas with attractive
interactions, which is an integrable quantum many-body system~\cite{LL63}. More precisely, one defines
\begin{equation}\label{2}
\mathbb{E}\Big( \prod_{\alpha=1}^n \mathrm{e}^{h(y_\alpha,t)}\Big) = f_t(\vec{y})
\end{equation}
with $\vec{y} = (y_1,...,y_n)$. Then
\begin{equation}\label{3}
\partial_t f_t = -H_n f_t,
\end{equation}
where $H_n$ is the $n$-particle Lieb-Liniger hamiltonian
\begin{equation}\label{4}
H_n =  -\tfrac{1}{2} \sum_{\alpha=1}^n \partial_{y_\alpha}^2  - \tfrac{1}{2} \sum_{\alpha \neq \alpha' = 1}^n\delta(y_\alpha - y_{\alpha'}).
\end{equation}

Almost thirty years later the generator of the asymmetric simple exclusion process (ASEP) was diagonalized through the Bethe ansatz. In case of $N$ sites, the ASEP configuration  space is $\{0,1\}^N$ signalling a similarity with quantum spin chains. In fact, the ASEP generator can be viewed as the
Heisenberg chain with an imaginary XY-coupling. For the totally asymmetric limit (TASEP) and half filled lattice, Gwa and Spohn~\cite{GS92} established that the spectral gap of the generator is of order $N^{-3/2}$. The same order is argued for the KPZ equation. This led to the strong belief that, despite their very different set-up, both models have the same statistical properties on large space-time scales. In the usual jargon of statistical mechanics, both
models are expected to  belong to the same universality class, baptized  KPZ universality class according to its most prominent representative.

The KPZ equation is solved with particular initial conditions. Of interest are (i) \textit{sharp wedge}, $h(x,0)= -c_0|x|$ in the limit $c_0 \to \infty$,
(ii) \textit{flat}, $h(x,0) = 0$, and (iii) \textit{stationary}, $h(x,0) = B(x)$ with $B(x)$ two sided Brownian motion.
The quantity of prime interest is the distribution of $h(0,t)$ for large $t$.
More ambitiously, but still feasible in some models, is the large time limit of  the joint distribution of $\{h(x_\alpha,t), \alpha = 1,...,n\}$.

In our notes we consider an integrable system of interacting diffusions, which is governed by
the coupled stochastic differential equations
\begin{equation}\label{5}
\mathrm{d}x_j(t) = \beta\mathrm{e}^{-\beta (x_j(t) - x_{j-1}(t))}\mathrm{d}t + \mathrm{d}B_j(t),
\end{equation}
where $\{B_j(t)\}$ a collection of independent standard Brownian motions. For the parameter $\beta >0$ we will eventually consider only the limit $\beta \to \infty$. But for the purpose of our discussion we keep $\beta$ finite for a while.
In fact  there is no choice, no other system of this structure is known to be integrable.
The index set depends on the problem, mostly we choose $j \in \mathbb{Z}$.  Note that $x_j$ interacts only with its left index neighbor $x_{j-1}$.  The drift depends  on the slope, as it should be for a proper height function. But the exponential dependence on $x_j - x_{j-1}$ is very special, however
familiar from other integrable systems. The famous Toda chain~\cite{Tod67} is a classical integrable system with exponential nearest neighbour interaction. Its quantized version is also integrable~\cite{Sut78}.

Interaction with only the left neighbor corresponds to the total asymmetric version. Partial asymmetry would read
\begin{equation}\label{5z}
\mathrm{d}x_j(t) = \big(p\beta\mathrm{e}^{-\beta (x_j - x_{j-1})} - (1-p)\beta\mathrm{e}^{-\beta (x_{j+1} - x_{j})}\big)\mathrm{d}t
 + \mathrm{d}B_j(t)
\end{equation}
with $0 \leq p \leq 1$. These are non-reversible diffusion processes. Only in the symmetric case, $p =\tfrac{1}{2}$, the drift is the gradient of a potential and the diffusion process is reversible. Then the model is no longer in the KPZ universality class and has very distinct large scale properties, see ~\cite{GPV88,CY92}, for example.

Eqs. \eqref{5} and \eqref{5z} should be viewed as a discretization of \eqref{1}. The independent Brownian motions  are the natural spatial discretization of the white noise $W(x,t)$. For the drift one might have expected the form $(x_j - x_{j-1})^2 + x_{j+1} - 2 x_j +x_{j-1}$,  but integrability  forces another dependence on the slope. For $p=1$ the similarity with the KPZ equation is even stronger when considering the exponential moments
\begin{equation}\label{5a}
\mathbb{E}\Big( \prod_{\alpha=1}^n \mathrm{e}^{\beta x_{m_\alpha}(t)}\Big) = f_t(\vec{m}),
\end{equation}
$\vec{m} \in \mathbb{Z}^n$. Differentiating in $t$ one obtains
\begin{equation}\label{5b}
\beta^{-2} \frac{\mathrm{d}}{\mathrm{d}t} f_t(\vec{m})= \sum_{\alpha = 1}^n \partial_\alpha f_t(\vec{m}) +  \tfrac{1}{2}\sum_{\alpha, \alpha' =1}^n
\delta (m_{\alpha} - m_{\alpha'})f_t(\vec{m}) - n f_t(\vec{m}),
\end{equation}
where $\delta$ is the Kronecker delta and $\partial_\alpha f(\vec{m}) = f(...,m_\alpha,...) -  f(...,m_\alpha - 1,...)$. When comparing with
\eqref{4}, instead of $\partial_\alpha$ one could have guessed the discrete Laplacian $\Delta_\alpha = -\partial_\alpha^{\mathrm{T}}\partial_\alpha$ with ${}^\mathrm{T}$ denoting the transpose. To obtain such a result, the drift in \eqref{5} would have to be replaced by
$\beta\mathrm{e}^{-\beta (x_j - x_{j-1})} + \beta\mathrm{e}^{-\beta (x_j - x_{j+1})}$. But the linear equations for the exponential moments are no longer Bethe integrable.
Note however that the semigroups $\exp[\partial_\alpha t]$ and $\exp[\Delta_\alpha t]$ differ on a large scale only by a uniform translation proportional to $t$. With such a close correspondence one would expect
system \eqref{5} to be in the KPZ universality class for any $\beta >0$, as has been verified to some extent~\cite{ACQ10,SS10b,BCFV14}.
Also partial asymmetry, $\tfrac{1}{2} <p <1$, should be in the KPZ universality class. In fact, the true claim is
by many orders more sweeping: the exponential in \eqref{5z}, $p \neq \tfrac{1}{2}$, can be replaced by ``any'' function of $x_j - x_{j-1}$, except for the linear one, and the system is still in the KPZ universality class. There does not seem to be a  promising idea around of how to prove such a property. Current techniques  heavily rely on integrability.

This is a good opportunity to reflect another difficulty. Bethe ansatz is like a first indication. But an interesting asymptotic analysis is yet another
huge step. This is well illustrated by the KPZ equation. Solving the  $n$-particle  equations \eqref{2} yields the exponential moment $\mathbb{E}\big(\mathrm{e}^{nh(0,t)}\big)$. However these moments grow rapidly as
$\exp(n^3)$, much too fast to determine the distribution of $h(0,t)$. Because of the underlying lattice, for system \eqref{5}  the exponential moments grow only as $\exp(n^2)$, still too fast.
In replica computations one nevertheless continues formally, often with correct results~\cite{CDR10, IS13, Dot10}. A proof must exploit integrability, but cannot use exponential moments directly~\cite{BC11}.

The system \eqref{5} simplifies substantially in the limit $\beta \to \infty$. Then one arrives at interacting Brownian motions, where
the Brownian motions maintain their ordering and
Brownian motion with label $j$ is reflected from its left neighboring Brownian motion with label  $j-1$, see~\cite{SS15}, Appendix B. These are the reflected Brownian motions of the title. The proper definition of their dynamics requires martingales involving local time, as will be discussed in Chapter
~\ref{secRBMaR}. Our note discusses exclusively this limit case. Thereby we arrive at a wealth of results on universal statistical properties. Only for the TASEP a comparably  detailed
analysis has been carried out ~\cite{FS11}, which does not come as a surprise, since in the limit of low density, under diffusive rescaling of space-time and switching to a moving frame of reference, the TASEP converges to system  \eqref{5}~\cite{KPS12}. We will not exploit this limit. Our philosophy
is to work in a framework which uses only interacting diffusions.

The limit $\beta \to \infty$ is meaningful also for $p \neq 0,1$. Then the order of Brownian particles is still preserved, but the reflection between neighbors is oblique. The symmetric version, $p = \tfrac{1}{2}$, corresponds to independent Brownian motions, maintaining their order, a case which has been studied quite some time ago~\cite{Har65}.
The partially asymmetric version of the model is still Bethe integrable, but less tractable. Only for the half-Poisson initial condition, an expression  sufficiently compact for asymptotic analysis has been obtained~\cite{SS15}.

The notion of integrability was left on purpose somewhat vague. In the $\beta = \infty$ limit for \eqref{5}, integrability can be more concretely illustrated. For this case, let us set $j = 2,...,n$ with $x_1(t)$ a standard Brownian motion. Then the transition probability from $\vec{x}$ to $\vec{y}$ at time $t$ is given by
\begin{equation}\label{6}
\mathbb{P}\big( \vec{x}(t) \in \mathrm{d}\vec{y} \big| \vec{x}(0) = \vec{x}\,\big) = \det \big{\{} \Phi_t^{(i-j)}(y_j - x_i)\big{\}}_{1 \leq i,j \leq n}\,\mathrm{d}\vec{y},
\end{equation}
where
\begin{equation}\label{7}
 \Phi_t^{(j)}(\xi,{t})=\frac{1}{2\pi\mathrm{i}}\int_{\mathrm{i}\mathbb{R}+\delta}\mathrm{d} w\,e^{{t} w^2/2+\xi w}w^{-j}
\end{equation}
with $\delta >0$ as first established by Sasamoto and Wadati~\cite{SW98}. There is a similar formula for the TASEP~\cite{Sch97}. Such formuli nourish the hope to uncover interesting features of the model.

The three initial conditions of particular interest, wedge, flat, and stationary, are easily transcribed to system \eqref{5} and become (i) \textit{packed}, half-infinite system with $x_j(0) = 0$ for
$j = 1,2,...$, (ii) \textit{periodic}, $x_j(0) = j$ for $j\in\mathbb{Z}$, (iii) \textit{Poisson}, $\{x_j(0), j\in \mathbb{Z}\}$ is a Poisson process with constant density.
According to our discussion, in the latter case one might think that the quantity of prime interest is $x_1(t)$. But the reflection induces a propagation of statistical fluctuations, as can also be
seen from
\eqref{5a} together with \eqref{5b}. Their propagation speed is $1$ and the correct quantity is $x_{\lfloor t \rfloor}(t)$ with $\lfloor t \rfloor$ denoting integer part. Along other
space-time observation rays a central limit type behavior would be observed.

As for other models in the KPZ universality class, our asymptotic analysis is limited to a single time and arbitrary number of spatial, resp. index  points. Only recently Johansson~\cite{Jo15}   posted a result on the joint distribution
of $(x_{\lfloor t \rfloor}(t), x_{\lfloor \alpha t \rfloor}(\alpha t))$, $\alpha >0$, and identified its universal limit. Possibly such progress will lead eventually to a complete understanding of the Airy sheet and the KPZ fixed point~\cite{CQR15}.

Let us explain of how our material is organized as a whole. The following three chapters provide background material. In Chapter~\ref{secRBMaR} we properly define the infinite system of reflected Brownian motions as the solution of a martingale problem and provide a variational formulation of this solution. Also the  uniform Poisson process is identified as stationary measure. In Chapter~\ref{secDPP} we introduce the theory of determinantal point processes and some related material on Fredholm determinants. At first sight this looks unconnected.
But to study the quantities of prime interest one first identifies a ``hidden'' signed determinantal process which leads to an analytically more tractable representation.
In the long time limit we will arrive at a stochastic process which describes the limiting spatial statistics. Such a process has been baptised Airy process, in analogy to the Airy kernel and Airy operator
which are one of the defining elements. In fact there are several Airy processes depending on the initial conditions and on the window of observation. The literature on Airy processes is somewhat dispersed.
Chapter~\ref{secAiryP} provides a streamlined account.

In Chapter~\ref{secPP} we investigate the two deterministic initial data, packed and periodic,  while in Chapter~\ref{secPoi} we study random initial data as defined through a Poisson process.
These results will be used to discuss more general  initial data, which should be viewed as an open ended enterprise. One natural choice is to have in the left half lattice either packed, periodic, or Poisson join up with either one of them in the right half lattice.  An example would be to have periodic to the left and Poisson to the right. This then leads to distinct cross over processes. The mixed cases are studied in Chapter~\ref{secMixed}. Slow decorrelation, referring to space-like paths more general than fixed time,
 is a further topic. Each core chapter builds on a suitable asymptotic analysis. In the early days the required techniques were developed ad hoc
 for the particular model. Over the years a common strategy based on contour integrations has been established, which will  be also used here. Thus on the basis of a specific example one can learn a technique applicable also to other models.

\chapter{One-sided reflected Brownian motions and related models}\label{secRBMaR}

For  reflected Brownian motions with index set $\mathbb{Z}_+$ we properly define the dynamics using an iterated Skorokhod construction.
The full index set $\mathbb{Z}$ requires a proof of well-posedness. In addition, we establish that the uniform Poisson process is stationary under the dynamics and discuss previous results for the model.

%In this chapter we give a precise definition of the model we are studying. For systems indexed by $\Z$, well-definedness has to be proven explicitly. %Additionally we show some properties of the system that do not require the techniques from Chapter~\ref{secDPP} and give a short overview of %previous results and connections to other models.

\section{Skorokhod construction}
Our definition of reflected Brownian  motions is through the so\nobreakdash-called Skorokhod representation~\cite{Sko61,AO76}, which is a deterministic function of the driving Brownian motions. This representation is the following: the process $x(t)$, driven by the Brownian motion $B(t)$, starting from $x(0)\in\R$ and being reflected (in the positive direction) at some continuous function $f(t)$ with $f(0)\leq x(0)$ is defined as:
\begin{equation}\begin{aligned}
x(t)&=x(0)+B(t)-\min\big\{0,\inf_{0\leq s \leq t}(x(0)+B(s)-f(s))\big\} \\
&= \max\big\{x(0)+B(t),\sup_{0\leq s \leq t}(f(s)+B(t)-B(s))\big\}.
\end{aligned}\end{equation}

Let $B_n$, $n\in\Z$, be independent standard Brownian motions starting at $0$. Throughout this work, $B_n$ always denotes these Brownian motions, allowing for coupling arguments. The non-positive indices will be used from Section~\ref{secInfPS} on.

\begin{defin}
The half-infinite system of one-sided reflected Brownian motions $\{x_n(t),n\geq1\}$ with initial condition $\vec{x}(0)=\vec{\zeta}$, $\zeta_n\leq\zeta_{n+1}$ is defined recursively by \mbox{$x_1(t)=\zeta_1+B_1(t)$} and, for $n\geq2$,
\begin{equation}
 x_n(t)=\max\big\{\zeta_n+B_n(t),\sup_{0\leq s \leq t}(x_{n-1}(s)+B_n(t)-B_n(s))\big\}.
\end{equation}
\end{defin}

Introducing the random variables
\begin{equation}\label{eq2.2}
Y_{k,n}(t)=\sup_{0\leq s_{k}\leq \ldots\leq s_{n-1}\leq t}\sum_{i=k}^n (B_i(s_{i})-B_i(s_{i-1}))
\end{equation}
for $k\leq n$, with the convention $s_{k-1}=0$ and $s_n=t$, allows for an equivalent explicit expression:
\begin{equation}\label{eq2.4}
x_n(t)= \max_{k\in [1,n]}\{Y_{k,n}(t)+\zeta_k\}.
\end{equation}
Although we are constructing an infinite system of particles, well-definedness is clear in this case, as each process $x_n(t)$ is a deterministic function of only finitely many Brownian motions $B_k$, $1\leq k\leq n$.

Adopting a stochastic analysis point of view, the system \mbox{$\{x_n(t),n\geq1\}$} satisfies
\begin{equation}
 x_n(t)=\zeta_n+B_n(t)+L^n(t),\quad\text{for } n\geq0,
 \end{equation}
Here, $L^1(t)=0$, while $L^n$, $n\geq2$, are continuous non-decreasing processes increasing only when $x_n(t)=x_{n-1}(t)$. In fact, $L^n$ is twice the semimartingale local time at zero of $x_n-x_{n-1}$.

\section{Packed initial conditions}
A canonical and in fact the most studied initial condition for the system \mbox{$\{x_n(t),n\geq1\}$} is the one where all particles start at zero. We call this \emph{packed initial condition}:
\begin{equation}
 \vec{x}(0)=\vec{\zeta}^\textrm{packed}=0.
\end{equation}
Using the monotonicity,
\begin{equation}\begin{aligned}
Y_{k-1,n}(t)&=\sup_{0\leq s_{k-1}\leq s_{k}\leq \ldots\leq s_{n-1}\leq t}\sum_{i=k-1}^n (B_i(s_i)-B_i(s_{i-1}))\\
&\geq\sup_{0= s_{k-1}\leq s_{k}\leq \ldots\leq s_{n-1}\leq t}\sum_{i=k}^n (B_i(s_i)-B_i(s_{i-1}))=Y_{k,n}(t),
\end{aligned}\end{equation}
and inserting this initial condition into \eqref{eq2.4}, leads to:
\begin{equation}
 x_n(t)=Y_{1,n}(t)=\sup_{0\leq s_{1}\leq \ldots\leq s_{n-1}\leq t}\sum_{i=1}^n (B_i(s_i)-B_i(s_{i-1})),
\end{equation}
again with the convention $s_0=0$ and $s_{n}=t$.

\subsection{Queues, last passage percolation and directed polymers}\label{secStepOther}
There are other interpretations for the quantity $x_n(t)$ than the system of reflected Brownian motions focused on in this work, and each interpretation has inspired different results over the last decades. One of these is seeing it as a sequence of \emph{Brownian queues in series}. A famous theorem of queueing theory, \emph{Burke's theorem}, which states that the output of a stable, stationary M/M/1 queue is Poisson, can be adapted to the Brownian setting~\cite{OCY01}. This will be employed in Section~\ref{secStationarity}.

Furthermore, $x_n(t)$ can be viewed as a model of directed last-passage percolation through a random medium, or equivalently a zero-temperature directed polymer in a random environment. This model is constructed as follows:
\begin{figure}
\begin{center}
 \psfrag{0}[lb]{$1$}
 \psfrag{1}[lb]{$2$}
 \psfrag{2}[lb]{$3$}
 \psfrag{3}[lb]{$4$}
 \psfrag{4}[lb]{$5$}
 \psfrag{s0}[lb]{$s_1$}
 \psfrag{s1}[lb]{$s_2$}
 \psfrag{s2}[lb]{$s_3$}
 \psfrag{s3}[lb]{$s_4$}
 \psfrag{t}[lb]{$t$}
\includegraphics[height=5cm]{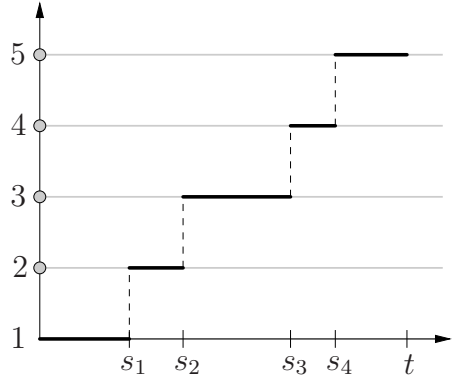}
 \caption{A path $\pi\in\Pi(0,1;t,5)$ (thick black) and the random background noise (gray).}
 \label{FigDP}
\end{center}
\end{figure}

Consider the space $\R_+\times\Z$ and assign white noise $\mathrm{d} B_n$ as random background weight on each line $\R_+\times \{n\}$ for $n\geq1$.
An up-right path is characterized by its jumping points $s_i$ and consists of line segments $[s_{n-1},s_n]\times \{n\}$, see Figure~\ref{FigDP}. The set of up-right paths going from $(t_1,n_1)$ to $(t_2,n_2)$ can then be parameterized by
\begin{equation}
 \Pi(t_1,n_1;t_2,n_2) = \left\{\vec{s}\in\R^{n_2-n_1+2}|t_1=s_{n_1-1}\leq s_{n_1}\leq \dots\leq s_{n_2}=t_2\right\}.
\end{equation}
The \emph{percolation time} or \emph{weight} of a path $\vec\pi\in\Pi$ is the integral over the background weights along the path. Explicitly, we have:
\begin{equation}
 w(\vec{\pi})=\sum_{i=n_1}^{n_2} \left(B_i(s_{i})-B_i(s_{i-1})\right).
\end{equation}

The \emph{last passage percolation time} is given by the supremum over all such paths:
\begin{equation}
 L_{(t_1,n_1)\to(t_2,n_2)}:=\sup_{\vec{\pi}\in\Pi(t_1,n_1;t_2,n_2)}w(\vec{\pi}).
\end{equation}
The supremum is almost surely attained by a unique path $\vec\pi^*$, called the maximizer. It exists because the supremum can be rewritten as a composition of a finite maximum and a supremum of a continuous function over a compact set. Uniqueness follows from elementary properties of the Brownian measure. Most importantly, from the definition, we have
\begin{equation}
 x_n(t)=L_{(0,1)\to(t,n)}.
\end{equation}

This representation will be used repeatedly throughout this work as it nicely visualizes coupling arguments, however, it also gives some connections to different works. Our model can be seen as the semi-continuous limit of a more widely studied discrete last passage percolation model (see for example~\cite{Jo00b,Jo03b}).

This last passage percolation model is also the zero temperature limit of a directed polymer model, which has been studied thoroughly in the recent past~\cite{SV10,BCFV14}. In the directed polymer setting we have a parameter $\beta$ representing the inverse temperature, consider $w(\vec{\pi})$ as an \emph{energy} and assign a Gibbs measure on the set of paths according to the density $e^{\beta w(\vec{\pi})}$, i.e.\ paths with higher energy have a higher probability. The partition function of the polymer is given by
\begin{equation}
 Z_{(t_1,n_1)\to(t_2,n_2)}(\beta)=\int_{\Pi(t_1,n_1;t_2,n_2)}\D\vec{\pi}\, e^{\beta w(\vec{\pi})},
\end{equation}
and satisfies the limit
\begin{equation}
 \lim_{\beta\to\infty}\frac{1}{\beta}\log Z_{(t_1,n_1)\to(t_2,n_2)}(\beta)=L_{(t_1,n_1)\to(t_2,n_2)}.
\end{equation}
Although apparently more difficult to handle, attention has been turning to these positive temperature models recently, because, among other things, they allow for a scaling limit to the KPZ equation by tuning the parameter $\beta$ in the right way.
\subsection{Previous results}
The behaviour of $x_n(t)$ under the packed initial condition is quite well understood by now. Notice that by Brownian scaling we have the distributional identity $\{x_n(t),n\geq1\}\stackrel{d}{=}\{\sqrt{t}x_n(1),n\geq1\}$. The first result has been a law of large numbers, i.e.\ $x_n(1)/\sqrt{n}$ converges to a constant~\cite{GlWh91}, therein it was already conjectured it equals $2$, which has been proven subsequently~\cite{Sep97}. Much of the results that followed exploited connections to a random matrix model, so let us introduce it in full generality right away:

Let $b_{i,i}(t)$ for $1\leq i\leq N$ and $b_{i,j}(t)$, $b'_{i,j}(t)$ for $1\leq i<j\leq N$, be independent Brownian motions. Define a stochastic process $H(t)$, $t\geq0$, on the space of $N\times N$ Hermitian matrices by
\begin{equation}\begin{aligned}
 H_{i,i}(t)&=b_{i,i}(t) & &\text{for } 1\leq i \leq N \\
 H_{i,j}(t)&=\sqrt{2}\left(b_{i,j}(t)+\mathrm{i} b'_{i,j}(t)\right) & &\text{for } 1\leq i<j \leq N \\
 H_{i,j}(t)&=\sqrt{2}\left(b_{i,j}(t)-\mathrm{i} b'_{i,j}(t)\right) & &\text{for } 1\leq j<i \leq N
\end{aligned}\end{equation}
Denote by $\lambda_1^{n}(t)\leq \lambda_2^n(t)\leq\dots\leq\lambda_n^n(t)$ the ordered eigenvalues of the $n\times n$ principal minor of the matrix $H(t)$. The process $\{\lambda_k^n(t),1\leq k\leq n\leq N\}$ is called the \emph{Dyson Brownian minor process}.
It is a classical result, that the eigenvalues of consecutive minors are interlaced, \mbox{$\lambda_k^{n+1}(t)\leq\lambda_k^n(t)\leq\lambda_{k+1}^{n+1}(t)$}, so this process lives in the Gelfand-Tsetlin cone:
\begin{equation}
 \textrm{GT}_N =\left\{x_{k}^n\in\R, 1\leq k \leq n \leq N, x_k^{n+1}\leq x_k^n\leq x_{k+1}^{n+1}\right\}.
\end{equation}

Restricted to one layer $n$, the process $\{\lambda_k^n(t),1\leq k\leq n\}$ is called \emph{Dyson's Brownian motion}~\cite{Dys62}. It is a Markov process and satisfies the coupled stochastic differential equation
\begin{equation}\label{eqDyson}
 \D\lambda_k^n=\mathrm{d} B_k+\sum_{i\neq k} \frac{\mathrm{d} t}{\lambda_k^n-\lambda_i^n}.
\end{equation}
Its fixed time distribution is the eigenvalue distribution of the Gaussian unitary ensemble (GUE):
\begin{equation}
 \Pb\left(\lambda_k^n(t)\in\mathrm{d} y_k,1\leq k\leq n\right)=\frac{1}{Z_n(t)}\prod_{k=1}^Ne^{-\frac{y_k^2}{2t}}\prod_{k<i}(y_i-y_k)^2\D\vec{y}.
\end{equation}
Recognizing the second product as the square of the Vandermonde determinant, one is able to describe this eigenvalue distribution as a determinantal point process governed by the Hermite kernel. Asymptotic analysis of the distribution of the largest eigenvalue in the appropriate \emph{edge scaling} gives the GUE Tracy-Widom distribution~\cite{TW94}.

The first connections between our system of reflecting Brownian motions and the matrix diffusion were found by~\cite{GTW00}, proving that for every $n\geq1$, $x_n(1)\stackrel{d}{=}\lambda_n^n(1)$, and~\cite{Bar01} generalizing this to $\{x_n(1),n\geq1\}\stackrel{d}{=}\{\lambda_n^n(1),n\geq1\}$ by a combinatorial procedure originating from group representation theory, called the Robinson-Schensted-Knuth (RSK) correspondence.~\cite{Bar01} also showed the remarkable fact that conditioned on the top layer \mbox{$\{\lambda_k^N(1),1\leq k\leq N\}$} the distribution of $\{\lambda_k^n(1),1\leq k\leq n<N\}$ is uniform on the compact set given by the Gelfand-Tsetlin interlacing inequality. Restricting the Dyson Brownian minor process to a fixed time gives the \emph{GUE minor process} $\{\lambda_k^n(1),1\leq k\leq n\leq N\}$, whose full distribution has been found in~\cite{JN06}, again in the form of a determinantal point process.

There is a natural extension of the system of reflecting Brownian motions $\{x_n(t),n\geq1\}$ to a process in the Gelfand-Tsetlin cone that is constructed in the following way: Let $B_1^1(t)$ be a Brownian motion. Let $B_1^2(t)$ and $B_2^2(t)$ be Brownian motions, which are reflected downwards resp. upwards from $B_1^1(t)$. Iteratively construct $B^n_k(t)$ as a Brownian motion reflected downwards from $B^{n-1}_k(t)$ and upwards from $B^{n-1}_{k-1}(t)$, with the peripheral processes $B^{n}_k(t)$ for $k=1$ or $k=n$ being reflected from one process only. By construction, we have $x_n(t)=B^n_n(t)$. The process $\{B^n_k(t),1\leq k\leq n\leq N\}$ is called \emph{Warren's process}, and has been introduced and studied in~\cite{War07}. Restricted to one layer $n$, it is distributed as a Dyson's Brownian motion. Warren's process shares the fixed time distribution with the GUE minor process,
\begin{equation}
 \{\lambda_k^n(1),1\leq k\leq n\leq N\}\stackrel{d}{=}\{B_k^n(1),1\leq k\leq n\leq N\}.
\end{equation}
There are also formulas for the transition density of the system along the edge, $\{B^n_n(t),1\leq n\leq N\}$ as well as for the system of two consecutive layers $\{B^k_n(t),1\leq k\leq n, N-1\leq n\leq N\}$.

The connection between Warren's process and the Dyson Brownian minor process does not, however, hold in full generality. The common dynamics of any amount of consecutive layers of Warrens process, i.\ e.\ the process $\{B^k_n(t),1\leq k\leq n, N_1\leq n\leq N_2\}$, is simply given by Dyson's SDE \eqref{eqDyson} for the layer $n=N_1$  and the reflection SDE's for the higher order layers. In the Dyson Brownian minor process, on the other hand, the common dynamics of two consecutive layers is given by a more complicated SDE (see (2.30) in~\cite{ANvM10b}) and the evolution of three or more consecutive layers is not even a Markov process anymore. Interestingly, both processes still show the same distribution along so-called space-like paths, i.e.\ sequences of points $(n_i,t_i)$ satisfying $t_i\leq t_{i+1}$ and $n_i\geq n_{i+1}$, in which case also determinantal formulas exist~\cite{ANvM10,FF10}.

The determinantal formulas coming from the random matrix model are suitable for asymptotic analysis to show multi-point scaling limits, where the Airy$_2$ process arises. It appears both for correlations of $x_n(t)$ along the $n$ direction and the $t$ direction, as well as along general space-like paths.~\cite{Jo03b} gives a sketch of the proof for both directions,~\cite{AvM03} prove the scaling limit with correlations along $t$ rigorously. A complete proof for correlations along $n$ is given in Section~\ref{secStep}.

\section{Infinite particle systems}\label{secInfPS}
The main focus of this work is showing determinantal formulas and scaling limits for other initial conditions. Unfortunately the connection to random matrices breaks down in this case. In fact, neither the Airy$_1$ process nor the Airy$_\textrm{stat}$ process have ever been found in a scaling limit of a random matrix model. This has been rather surprising, as the one-point distribution of the Airy$_1$ process, the GOE Tracy-Widom distribution, \emph{does} arise in such a model, namely as the limiting distribution of the largest eigenvalue of a Gaussian real symmetric matrix, the GOE ensemble~\cite{BFP08}.

\subsection{Definition}
At first sight it might seem trivial to extend the definition of the half-infinite system of one-sided reflected Brownian motions to an infinite number of particles \mbox{$\{x_n(t),n\in\Z\}$}, by letting the index $k$ run over $(-\infty,n]$ in \eqref{eq2.4}. However, it has to be shown that this maximum is finite, which is only the case for initial conditions which are not too closely spaced together. Roughly said, the growth rate of $-\zeta_{-k}$ has to be faster than $\sqrt{k}$ for large $k$. We call those initial conditions \emph{admissible}.

Knowing the law of large numbers under packed initial conditions, it is reasonable to expect such a behaviour. It implies that $Y_{k,n}(t)$ grows roughly as $2\sqrt{(n-k)t}$, so in order for the maximum in \eqref{eq2.3.7} being attained by a finite $k$ for we need $-\zeta_{-k}$ to grow faster than $\sqrt{k}$.

\begin{defin}
 A random vector $\vec{\zeta}\in\R^\Z$ with $\zeta_n\leq\zeta_{n+1}$ for all $n\in\Z$ is called an \emph{admissible initial condition}, if there exists a $\chi>\frac{1}{2}$ such that for any $n\in\Z$ the sum
 \begin{equation}
  \sum_{M\geq0}\Pb\left(\zeta_n-\zeta_{-M}\leq M^\chi\right)
 \end{equation}
is finite.
\end{defin}

\begin{defin}%\label{}
Let $\vec{\zeta}\in\R^\Z$ be an admissible initial condition. The infinite system of one-sided reflected Brownian motions $\{x_n(t),n\in\Z\}$ with initial condition $\vec{x}(0)=\vec{\zeta}$ is defined by
\begin{equation}\label{eq2.3.7}
x_n(t)=\max_{k\leq n}\{Y_{k,n}(t)+\zeta_k\}.
\end{equation}
\end{defin}
By Proposition~\ref{propMconv} below, which is proven in Section~\ref{secWellD}, this maximum exists and is finite. More specifically, we will show that for $\vec{\zeta}$ being any admissible initial condition, the infinite system \mbox{$\{x_n(t),n\in\Z\}$} is the limit of certain half-infinite systems $\{x_n^{(M)}(t),n\geq -M\}$ as $M\to\infty$, where
\begin{equation}\label{eq3a}
x_n^{(M)}(t)=\max_{k\in [-M,n]}\{Y_{k,n}(t)+\zeta_k\},\quad n\geq-M.
\end{equation}
Notice that these processes indeed satisfy the Skorokhod equation,
\begin{equation}\label{eqMsko}
x_n^{(M)}(t)= \max\big\{\zeta_n+B_n(t),\sup_{0\leq s \leq t}(x_{n-1}^{(M)}(s)+B_n(t)-B_n(s))\big\},
\end{equation}
for $n>-M$, while the leftmost process is simply
\begin{equation}
 x_{-M}^{(M)}(t)=\zeta_{-M}+B_{-M}(t).
\end{equation}
Thus as desired $x_n^{(M)}(t)$ is a Brownian motion starting from $\zeta_n$ and reflected off by $x_{n-1}^{(M)}$ for $n>-M$. %The representation (\ref{eq3a}) can be sees as a zero-temperature case of the O'Connell-Yor semi-discrete directed polymer~\cite{OCY01} with appropriate boundary conditions (see discussion at the end of this section).

%Next we show strong converge of the system $\{x_n^{(M)}(t),n\geq -M\}$ to the processes we are studying.
\begin{prop}\label{propMconv}
For any $t>0$, $n\in\Z$ there exists almost surely a $k\leq n$ maximizing $Y_{k,n}(t)+\zeta_k$, i.e.\ the maximum in \eqref{eq2.3.7} exists. Furthermore, for any $T>0$,
\begin{equation}\label{eqUniBound}
	\sup_{t\in[0,T]}|x_n(t)|<\infty,\quad \textrm{a.s.},
\end{equation}
as well as
\begin{equation}
	\lim_{M\to\infty}\sup_{t\in[0,T]}|x_n^{(M)}(t)-x_n(t)|=0,\quad \textrm{a.s.}.
\end{equation}
\end{prop}
The convergence result allows for taking the limit in \eqref{eqMsko}, implying that the system $\{x_n(t),n\in\Z\}$ satisfies the Skorokhod equation, too.

The initial conditions we are actually interested in are corresponding to the two remaining fundamental geometries in the KPZ universality class. The first one is the flat surface, which translates into periodic initial conditions $\vec{x}(0)=\vec{\zeta}^\textrm{flat}$, defined by
\begin{equation}
 \zeta_n^\textrm{flat}=n, \text{ for } n\in\Z,
\end{equation}
which is obviously admissible.

Finally we will study the case where the model starts in its random stationary distribution, which is in our case a Poisson point process on the real line. However, as already familiar from other models in the KPZ universality class~\cite{BFP09,BCFV14,FS05a,SI04}, Theorem~\ref{thmAsymp0} will be proven via a sequence of approximating initial conditions.

Let therefore be $\{\textrm{Exp}_n,n\in\Z\}$ be i.i.d. random variables with exponential distribution with parameter $1$. For parameters $\lambda>0$ and $\rho>0$ define the initial condition $\vec{x}(0)=\vec{\zeta}^\textrm{stat}(\lambda,\rho)$ by
\begin{equation}\label{statModel}
 \begin{aligned}
  \zeta^\textrm{stat}_0&=0,\\
  \zeta^\textrm{stat}_n-\zeta^\textrm{stat}_{n-1}&=
  \begin{cases}\lambda^{-1}\,\textrm{Exp}_n, \quad &\text{for } n>0,\\\rho^{-1}\,\textrm{Exp}_n, &\text{for } n\leq0.\end{cases}
 \end{aligned}
\end{equation}
Admissibility of this initial condition is also not hard to prove.

To recover the uniform Poisson process on the whole real line, we will set $\lambda=1$ and carefully take the limit $\rho\to1$ in the determinantal formulas that hold in the case $\rho<\lambda$. Finally, setting $\zeta_0=0$ will induce a difference of order one as compared to the true Poisson process case. By Proposition~\ref{propBoundMod} this difference will stay bounded at all times, and consequently be irrelevant in the scaling limit. Thus it is enough to prove Theorem~\ref{thmAsymp0} for the initial conditions $\vec{x}(0)=\vec{\zeta}^\textrm{stat}(1,1)$.

%\textit{In the sequel $x_n(t)$  always refers to the initial conditions \eqref{statModel}, in such a way that the choice of the parameters $\lambda,\rho$ can be inferred from the context}. We obtain the fixed time multi-point distributions of the system $\{x_n(t),n\in\Z_{\geq0}\}$ in terms of a Fredholm determinant in the case $\lambda >\rho$.
%The restriction to non-negative integers comes from Burke's theorem by which the  particles with $n <0 $ can be replaced by choosing $x_0(t)$ as a Brownian motion with drift $\rho$.

\subsection{Well-definedness}\label{secWellD}
For the proof of Proposition~\ref{propMconv} we first need the following concentration inequality:
\begin{prop}\label{concIn}
For each $T>0$ there exists a constant $C>0$ such that for all \mbox{$k<m$}, $\delta>0$,
\begin{equation}
	\Pb\bigg(\frac{Y_{k,m}(T)}{\sqrt{(m-k+1)T}}\geq 2+\delta\bigg)\leq \const \cdot e^{-(m-k+1)^{2/3}\delta}.
\end{equation}
\end{prop}
This proposition is proven in Section~\ref{secStepAsy}.
Another necessary lemma, that will be proven an intuitive way in Section~\ref{secInfLPP}, is:
\begin{lem}\label{lemMaxPath}
Consider $0\leq t_1\leq t_2$ and $m$, $M_{t_1}$, $M_{t_2}$ such that
\begin{equation}\label{eq10a}
	x_m(t_i)=x_m^{(M_{t_i})}(t_i)=\widetilde{x}_m^{(M_{t_i})}(t_i),\quad\text{for }i=1,2.
\end{equation}
Then
\begin{equation}\label{eq11}
	x_m(t_1)=x_m^{(M_{t_2})}(t_1)=\widetilde{x}_m^{(M_{t_2})}(t_1).
\end{equation}
\end{lem}

\begin{proof}[Proof of Proposition~\ref{propMconv}]
Let us define an auxiliary system of processes, which we will use later in proving Proposition~\ref{propBM}, by
\begin{equation}
 \widetilde{x}_{-M}^{(M)}(t)=\zeta_{-M}+B_{-M}(t)+\rho t,
\end{equation}
and
\begin{equation}
\widetilde{x}_n^{(M)}(t)= \max\big\{\zeta_n+B_n(t),\sup_{0\leq s \leq t}(\widetilde{x}_{n-1}^{(M)}(s)+B_n(t)-B_n(s))\big\}
\end{equation}
for $n>-M$. This system differs from $x_n^{(M)}(t)$ just in the drift of the leftmost particle, which of course influences all other particles as well (the choice of the extra drift is because the system with infinite many particles in $\R_-$ generates a drift $\rho$). This system of particles satisfies
\begin{equation}\label{eq3}
\widetilde{x}_n^{(M)}(t)=\max\big\{\widetilde{Y}_{-M,n}(t)+\zeta_{-M},\max_{k\in [-M+1,n]}\{Y_{k,n}(t)+\zeta_k\}\big\},
\end{equation}
with
\begin{equation}
\widetilde{Y}_{k,n}(t)=\sup_{0\leq s_{k+1}\leq \ldots\leq s_m\leq t}\Big(\rho s_{k+1}+\sum_{i=k}^n (B_i(s_{i+1})-B_i(s_i))\Big).
\end{equation}
Also, we have the inequalities
\begin{equation}\label{eq2.16}
 Y_{k,n}(t)\leq\widetilde{Y}_{k,n}(t)\leq Y_{k,n}(t)+\rho t.
\end{equation}

Consider the event
\begin{equation}\begin{aligned}\label{eq3.15}
	A_M&:=\{Y_{-M,n}(T)\geq3\sqrt{(M+n+1)T}\}\cup\{\zeta_n-\zeta_{-M}\leq M^\chi\}\\&\qquad\cup\{Y_{n,n}(T)\leq \rho T+3\sqrt{(M+n+1)T}-M^\chi\}.
\end{aligned}\end{equation}
It is now straightforward to show $\sum_{M=0}^\infty\Pb(A_M)<\infty$. In fact, summability of the probabilities of the first set in \eqref{eq3.15} is a consequence of Proposition~\ref{concIn}, applied with $\delta=1$, while the second set is covered by the definition of an admissible initial condition. For the third set, notice that the left hand side is a Gaussian distribution independent of $M$, while the right hand side is dominated by the $M^\chi$ term for large $M$. Finiteness of the sum allows applying Borel-Cantelli, i.e.\ $A_M$ occurs only finitely many times almost surely. This means, that a.s. there exists a $M_T$, such that for all $M\geq M_T$ the following three inequalities hold:
\begin{equation}\begin{aligned}
Y_{-M,n}(T)&<3\sqrt{(M+n+1)T}\\
M^\chi&<\zeta_n-\zeta_{-M} \\
-Y_{n,n}(T)&< -\rho T-3\sqrt{(M+n+1)T}+M^\chi
\end{aligned}\end{equation}
Adding up these, $Y_{-M,n}(T)+\zeta_{-M}+\rho T<Y_{n,n}(T)+\zeta_n$ for all $M\geq M_T$ and dropping the term $\rho T$ shows us that the maximizing element in \eqref{eq2.3.7} cannot be a $k\leq-M_T$, or
\begin{equation}
	x_n(T)=x_n^{(M_T)}(T).
\end{equation}
Moreover, applying \eqref{eq2.16}, gives
\begin{equation}
 \widetilde{Y}_{-M_T,n}(t)+\zeta_{-M_T}\leq Y_{-M_T,n}(t)+\zeta_{-M_T}+\rho T<Y_{n,n}(T)+\zeta_n,
\end{equation}
resulting in
\begin{equation}
	\widetilde{x}_n^{(M_T)}(T)=x_n^{(M_T)}(T).
\end{equation}

Repeating the same argument, we see that for every $t\in[0,T]$ there exists $M_t$ such that $x_n(t)=x_n^{(M_t)}(t)=\widetilde{x}_n^{(M_t)}(t)$. Applying Lemma~\ref{lemMaxPath} then gives $x_n(t)=x_n^{(M_T)}(t)=\widetilde{x}_n^{(M_T)}(t)$ for every $t\in[0,T]$. This settles the convergence and the existence of a finite maximizing $k$ in \eqref{eq2.3.7}.

To see \eqref{eqUniBound}, which is equivalent to $\sup_{t\in[0,T]}|x_m^{(M_T)}(t)|<\infty$, we apply the bound
\begin{equation}
	|Y_{k,n}(t)|\leq\sum_{i=k}^n\Big(\sup_{0\leq s\leq t}B_i(s)-\inf_{0\leq s\leq t}B_i(s)\Big)<\infty.
\end{equation}
\end{proof}
\subsection{Last passage percolation}\label{secInfLPP}
It is also possible to extend the last passage percolation interpretation to nontrivial initial conditions. In order to do this, add non-negative Dirac background weights $\zeta_k-\zeta_{k-1}$ on $(0,k)$, $k\in\Z$. The weight of a path is explicitly given by
\begin{equation}
 w(\vec{\pi})=\sum_{i=n_1}^{n_2} \left( B_i(s_{i})-B_i(s_{i-1})+(\zeta_i-\zeta_{i-1})\Id_{s_{i-1}=0}\right),
\end{equation}
and the percolation time $L_{(0,n_1)\to(t,n_2)}$ again as the supremum over the weight of all paths. As $t\to0$ it is clear that any contribution from the Brownian background weight will converge to $0$, so the path tries to accumulate as much of the Dirac weights as possible, i.e.\ we have the initial condition
\begin{equation}
 \lim_{t\to0} L_{(0,n_1)\to(t,n_2)} = \sum_{i=n_1}^{n_2}(\zeta_i-\zeta_{i-1}) = \zeta_{n_2}-\zeta_{n_1-1}.
\end{equation}

By defining a normalized percolation time,
\begin{equation}
 \widehat{L}_{(0,n_1)\to(t,n_2)}=L_{(0,n_1)\to(t,n_2)}+\zeta_{n_1-1},
\end{equation}
we recover the system
\begin{equation}
 x^{(M)}_n(t)=\widehat{L}_{(0,-M)\to(t,n)}.
\end{equation}

For any $M\leq n$, $M=-\infty$ included, we can define an \emph{exit point} of a path $\pi=(\dots,s_{n_2-1},s_{n_2})\in\Pi(0,-M;t,n)$ by
\begin{equation}
 \inf\{k\in[n_1,n_2],s_k>0\},
\end{equation}
which is of course the maximizing index $k$ in \eqref{eq3a}.

We also can reproduce the system $\widetilde{x}^{(M)}_n(t)$ by adding a Lebesgue measure of density $\rho$ on the line $\{-M\}\times\R_+$.

%Optional{Picture of intersecting paths}

\begin{proof}[Proof of Lemma~\ref{lemMaxPath}]
 For $t_1=t_2$ there is nothing to prove, so let \mbox{$t_1<t_2$}. For each $i$, the equation $x_m(t_i)=x_m^{(M_{t_i})}(t_i)$ implies that the maximizing paths of the LHS and the RHS are equal on the restriction to $(s_i,i\geq M_i)$, i.e.\ they have the same exit point $e_i$ that satisfies $e_i\geq M_{t_i}$. Now if $e_1\geq e_2$, then also $e_1\geq M_{t_2}$, which means that the path maximizing $x_m(t_1)$ is contained in the set $\Pi(0,-M_{t_2};t_1,m)$, resulting in $x_m(t_1)=x_m^{(M_{t_2})}(t_1)$.

 If, however, $e_1< e_2$, then the maximizing path segments $(0,e_1)\to(t_1,m)$ and $(0,e_2)\to(t_2,m)$ would need to have an intersection point $(t^*,m^*)$. We can then construct a new maximizing path for $x_m(t_1)$ by stringing together the segments $(0,e_1)\to(0,e_2)\to(t^*,m^*)\to(t_1,m)$, where the middle segment is part of the $x_m(t_2)$-maximizing path and the last segment is part of the original $x_m(t_1)$-maximizing path. This contradicts the uniqueness of the maximizing path.

 The equality $x_m(t_1)=\widetilde{x}_m^{(M_{t_2})}(t_1)$ is shown in the same way.
\end{proof}

\subsection{Stationarity}\label{secStationarity}
We establish a useful property which will allow us to study our system of interacting Brownian motions through a system with a left-most Brownian particle.
\begin{prop}\label{propBM}
Under the initial condition $\vec{x}(0)=\vec{\zeta}=\vec{\zeta}^\mathrm{stat}(\lambda,\rho)$ defined in \eqref{statModel}, for each $n\leq0$ the process
\begin{equation}\label{eq2.9}
 x_n(t)-\zeta_n-\rho t
\end{equation}
is a standard Brownian motion.
\end{prop}

\begin{remark}\label{RemarkBurke}
Proposition~\ref{propBM} allows us to restrict our attention to the half-infinite system. In fact, conditioned on the path of $x_0$, the systems of particles $\{x_n(t),n<0\}$ and $\{x_n(t),n>0\}$ are independent, as it is clear by the definition of the system. Then (\ref{eq2.9}) implies that the law of \mbox{$\{x_n(t), n>0\}$} is the same as the one obtained replacing the infinitely many particles $\{x_m(t), m\leq 0\}$ with a single Brownian motion $x_0(t)$ which has a drift $\rho$. This property will be used to derive our starting result, Proposition~\ref{propKernel}.
\end{remark}

\begin{remark}From a stochastic analysis point of view, we find that the system \mbox{$\{x_n(t),n\geq0\}$} satisfies
\begin{equation}\begin{aligned}
 x_n(t)&=\zeta_n+B_n(t)+L^n(t),\quad\text{for } n\geq1,\\
 x_0(t)&=\widetilde{B}_0(t)+\rho t.
 \end{aligned}\end{equation}
Here $L^n$, $n\geq2$, are continuous non-decreasing processes increasing only when $x_n(t)=x_{n-1}(t)$. In fact, $L^n$ is twice the semimartingale local time at zero of $x_n-x_{n-1}$.
Notice that $\widetilde{B}_0(t)$ is a standard Brownian motion independent of $\{\zeta_n,B_n(t),n\geq1\}$, but not equal to $B_0(t)$.\end{remark}

\begin{proof}[Proof of Proposition~\ref{propBM}]
First notice that for any $M$,
\begin{equation}
 \widetilde{x}_{-M}^{(M)}(t)-\zeta_{-M}-\rho t,
\end{equation}
is a Brownian motion. Now assume $\widetilde{x}_{n-1}^{(M)}(t)-\zeta_{n-1}-\rho t$ is a Brownian motion. By definition,
\begin{equation}\begin{aligned}
\widetilde{x}_n^{(M)}&(t)-\zeta_{n-1} \\&=\max\big\{\zeta_n-\zeta_{n-1}+B_n(t),\sup_{0\leq s \leq t}(\widetilde{x}_{n-1}^{(M)}(s)-\zeta_{n-1}+B_n(t)-B_n(s))\big\},
\end{aligned}\end{equation}
which allows us to apply Proposition~\ref{propNeil}, i.e., we have that
\begin{equation}
 \widetilde{x}_n^{(M)}(t)-\zeta_{n-1}-(\zeta_n-\zeta_{n-1})-\rho t=\widetilde{x}_n^{(M)}(t)-\zeta_n-\rho t
\end{equation}
is a Brownian motion. Since $\widetilde{x}_n^{(M_T)}(t)=x_n(t)$ the proof is completed.
\end{proof}

It is clear, that in the case $\lambda=\rho$ the process \eqref{eq2.9} is a Brownian motion for $n>0$, too, i.e.,  the system is stationary in $n$. We also have stationarity in $t$, in the sense that for each $t\geq0$ the random variables \mbox{$\{x_n(t)-x_{n-1}(t),n\in\Z\}$} are independent and distributed exponentially with parameter $\rho$. The following result is a small modification of Theorem 2 in~\cite{OCY01}.

\begin{prop}[Burke's theorem for Brownian motions]\label{propNeil}
Fix $\rho>0$ and let $B(t)$, $C(t)$ be standard Brownian motions, as well as $\zeta\sim\exp(\rho)$, independent. Define the process
\begin{equation}
D(t)=\max\big\{\zeta+C(t),\sup_{0\leq s \leq t}(B(s)+\rho s+C(t)-C(s))\big\}.
\end{equation}
Then
\begin{equation}\label{eq2.29}
 D(t)-\zeta-\rho t
\end{equation}
is distributed as a standard Brownian motion.
\end{prop}
\begin{proof}
Extend the processes $B(t)$, $C(t)$ to two-sided Brownian motions indexed by $\R$. Defining
\begin{equation}
 q(t)=\sup_{-\infty<s\leq t}\{B(t)-B(s)+C(t)-C(s)-\rho(t-s)\}
\end{equation}
and
\begin{equation}
 d(t)=B(t)+q(0)-q(t),
\end{equation}
we can apply Theorem 2~\cite{OCY01}, i.e., $d(t)$ is a Brownian motion. Now,
\begin{equation}
 q(0)=\sup_{s\leq0}\{-B(s)-C(s)+\rho s\}\stackrel{d}{=}\sup_{s\geq0}\{\sqrt{2}B(s)-\rho s\}\stackrel{d}{=}\sup_{s\geq0}\big\{B(s)-\frac{\rho}{2} s\big\},
\end{equation}
so by Lemma~\ref{lemBMExp} it has exponential distribution with parameter $\rho$. As it is independent of the processes $\{B(t),C(t),t\geq0\}$ we can write $q(0)=\zeta$. Dividing the supremum into $s<0$ and $s\geq0$ we arrive at:
\begin{equation}\begin{aligned}
 -d(t)&=q(t)-B(t)-q(0)\\
 &=\max\Big\{C(t)-\rho t,\sup_{0\leq s\leq t}\{-B(s)+C(t)-C(s)-\rho(t-s)\}-\zeta\Big\},
\end{aligned}\end{equation}
which is \eqref{eq2.29} up to a sign flip of $B(s)$.
\end{proof}

\begin{lem}\label{lemBMExp}
Fix $\rho>0$ and let $B(t)$ be a standard Brownian motion. Then
 \begin{equation}
  \sup_{s\geq0}(B(s)-\rho s)\sim\exp(2\rho).
 \end{equation}
\end{lem}
\begin{proof}
This formula can be found in~\cite{BS96}, Part II, Section 2, Eq.\ (1.1.4).  We provide here  two different proofs of the claim.\\ (a) The random variable
 \begin{equation}
  \sup_{0\leq s\leq t}(B(s)-\rho s)
 \end{equation}
is distributed as $\tilde B(t)$, where $\tilde B$ is a Brownian motion starting at $0$, reflected (upwards) at zero and drift $-\rho$. This follows from the Skorokhod construction. Indeed,
\begin{equation}
\begin{aligned}
\tilde B(t)&\stackrel{d}{=}\sup_{0\leq s\leq t} (B(t)-B(s) - \rho(t-s)) \\
&= \sup_{0\leq s\leq t} (B(t)-B(t-s) - \rho s) \stackrel{d}{=}\sup_{0\leq s\leq t} (B(s) - \rho s),
\end{aligned}
\end{equation}
where we changed $s$ into $t-s$ and used the fact that $B(t)-B(t-s)$ has the same distribution of $B(s)$. As $t\to\infty$, this converges to the stationary distribution of this process, which is the exponential distribution with parameter $2\rho$.\\
(b) For $x\geq 0$, define the stopping time $\tau_x=\inf\{s\geq 0\,|\, B(s)-s\rho/2\geq x\}$. Then,
\begin{equation}
\Pb\left(\sup_{s\geq 0}(B(s)-\rho s)\geq x\right)=\Pb(\tau_x<\infty).
\end{equation}
The process $s\mapsto M_s=e^{\rho B(s)-\tfrac12 \rho^2 s}$ is a martingale (the geometric Brownian motion), with $M_0=1$. Thus by applying the optional sampling theorem (apply it first to $\tau_x\wedge T$ and then take $T\to\infty$) one obtains
\begin{equation}
1=\mathbbm{E}(M_{\tau_x}) = e^{\rho x}\Pb(\tau_x<\infty)+0\Pb(\tau_x=\infty),
\end{equation}
from which $\Pb(\tau_x<\infty)=e^{-\rho x}$ as claimed (just replace $\rho$ by $2\rho$).
\end{proof}

\chapter{Determinantal point processes}\label{secDPP}
Determinantal point processes are the main tool for the study of reflected Brownian motions. Thereby marginal distributions can be expressed in terms of Fredholm determinants, a form which is well suited for an asymptotic analysis. However, only partial aspects of the underlying theory
of determinantal point processes is needed for our purposes and we merely introduce the main definitions including the crucial Lemma~\ref{lemDetMeasure}.
Up to minor modifications, we follow~\cite{Jo05} as a very accessible introduction to the topic.

%The main tool for studying the systems of reflected Brownian motions is the theory of determinantal point processes and Fredholm determinants. It %leads to formulas for the marginal distributions of measures given by a product of determinants in a form that is suitable for asymptotic analysis. %However, not much of the underlying theory is needed for our purposes, so we just introduce the main definitions and the crucial
%Lemma~\ref{lemDetMeasure}. Up to minor modifications, we follow~\cite{Jo05}, which is a well-written introduction to the topic.
\section{Definition}

\begin{defin}
 Let $\Lambda$ be a complete separable metric space and let $\mathcal{N}(\Lambda)$ denote the space of all counting measures $\mu$ on $\Lambda$ which are boundedly finite, i.e.\ $\mu(X)<\infty$ for all bounded $X\subseteq\Lambda$. Define a $\sigma$-algebra $\mathcal{F}$ on $\mathcal{N}(\Lambda)$ by taking the smallest $\sigma$-algebra for which $X\mapsto\mu(X)$ is measurable for all Borel sets $X$ in $\Lambda$.

 A \emph{point process} is a probability measure on $\mathcal{N}(\Lambda)$. A \emph{signed point process} is a signed measure on $\mathcal{N}(\Lambda)$ that is normalized to $1$.
\end{defin}
Typically, $\Lambda$ will be either $\R$ or $S\times\R$ for some discrete set $S$. For any realization $\mu$ of a point process $\Xi$, $\mu(X)$ is interpreted as the number of points in the set $X$. For each bounded set $X$, the restriction of the point process to this set is given by a finite sum of Dirac measures:
\begin{equation}
 \mu|_X=\sum_{i=1}^{\mu(X)}\delta(x_i).
\end{equation}

For a given point process $\Xi$ we can construct a measure $M_n$ over $\Lambda^n$, called the factorial moment measure, by
\begin{equation}
 M_n=\mathbbm{E}\left[\sum_{x_{i_1}\neq\dots\neq x_{i_n}}\delta(x_{i_1},\dots,x_{i_n})\right],
\end{equation}
where we abused the notation for the expectation in the case of a signed point process. The measure $M_n$ is an intensity measure for $n$-tuples of distinct points in the original process.

\begin{defin}
 If $M_n$ is absolutely continuous with respect to the Lebesgue measure, i.e.
 \begin{equation}
  M_n(X_1,\dots,X_n)=\int_{X_1\times\dots\times X_n}\rho_n(x_1,\dots,x_n)\,\mathrm{d} x_1\dots\mathrm{d} x_n
 \end{equation}
for all Borel sets $X_i$ in $\Lambda$, we call $\rho_n(x_1,\dots,x_n)$ the $n$-th \emph{correlation function} or \emph{joint intensity}.
\end{defin}

Often times we will construct a point process from a symmetric measure on $\R^N$ with a density $p_N(x_1,\dots,x_N)$ by employing the canonical map \mbox{$(x_1,\dots,x_N)\mapsto\sum_{i=1}^N\delta(x_i)$}, and speak interchangeably of the measure and its associated point process. The correlation functions are then given by
\begin{equation}
 \rho_n(x_1,\dots,x_n)=\frac{N!}{(N-n)!}\int_{\R^{N-n}} p_N(x_1,\dots,x_N)\,\mathrm{d} x_{n+1}\dots\mathrm{d} x_N.
\end{equation}
If the density $p_N$ is not symmetric, but instead normalized over the Weyl chamber $W^N=\{\vec{x}\in\R^N|x_1\leq\dots\leq x_N\}$, it can be transformed into a symmetric one with density
\begin{equation}
\frac{1}{N!}\sum_{\sigma\in\mathcal{S}_N}p_N(x_{\sigma(1)},\dots,x_{\sigma(N)}).
\end{equation}

We will study point processes whose correlation function are given by determinants:
\begin{defin}
 Consider a (signed) point process on $\Lambda$, all of whose correlation functions exist. If there is a function $K\colon\Lambda\times\Lambda\to\C$ such that
 \begin{equation}
  \rho_n(x_1,\dots,x_n)=\det_{1\leq i,j\leq n}[K(x_i,x_j)],
 \end{equation}
for all $x_1,\dots,x_n\in\Lambda$, $n\geq1$, then we say that it is a \emph{(signed) determinantal process}, and call $K$ its \emph{correlation kernel}.
\end{defin}

Correlation functions allow for a convenient calculation of \emph{hole probabilities}, i.e.\ the probability of finding no particle in some set $X$:
\begin{equation}\label{eqhole}
 \Pb\left(\mu(X)=0\right)=\sum_{n\geq0}\frac{(-1)^n}{n!}\int_{X^n}\rho_n(x_1,\dots,x_n)\mathrm{d} x_1\dots\mathrm{d} x_n.
\end{equation}
If $\Lambda=\R$, then choosing $X=(s,\infty)$ gives the distribution function of the rightmost particle $\Pb(x_\textrm{max}\leq s)$, provided the series converges absolutely.

\section{Fredholm determinants}
Let $\mathcal{H}$ be a separable Hilbert space with scalar product
denoted by $\langle\cdot,\cdot\rangle$ and $A$ be a bounded linear operator acting on $\mathcal{H}$. Let $|A|=\sqrt{A^*A}$ be the unique square root of the operator $A^*A$, and $\{e_i,i\in I\}$ be an orthonormal basis of $\mathcal{H}$. The \emph{trace norm} of $A$ is given by $||A||_1=\sum_i\langle e_i,|A|e_i\rangle$, and $A$ is called \emph{trace class}, if $||A||_1<\infty$. Similarly, the \emph{Hilbert-Schmidt norm} of $A$ is given by $||A||_2=\left(\,\sum_i||Ae_i||^2\right)^{1/2}$, and $A$ is called \emph{Hilbert-Schmidt}, if $||A||_2<\infty$. Both norms are independent of the choice of the basis, and, with $||\cdot||_\textrm{op}$ denoting the usual operator norm, satisfy:
\begin{equation}
 ||A||_\textrm{op}\leq||A||_2\leq||A||_1.
\end{equation}
With $B$ being another bounded linear operator acting on $\mathcal{H}$ we also have the inequalities
\begin{equation}
||AB||_1\leq||A||_2||B||_2,
\end{equation}
as well as
\begin{equation}
 ||AB||_1\leq||A||_1||B||_\textrm{op}.
\end{equation}

We call $A$ an \emph{integral operator} on the space $L^2(\Lambda)$ if there is a function \mbox{$A\colon\Lambda\times\Lambda\to\R$}, called its \emph{integral kernel}, so that $(Af)(x)=\int_\Lambda A(x,y)f(y)\mathrm{d} y$. We abuse notation by denoting the operator and its integral kernel by the same letter. The Hilbert-Schmidt norm of an integral operator is given by
\begin{equation}
 ||A||_2^2=\int_{\Lambda^2}|A(x,y)|^2\mathrm{d} x\mathrm{d} y.
\end{equation}

\begin{defin}\label{defFredhD}
 Let $A\colon L^2(\Lambda)\to L^2(\Lambda)$ be a trace class operator with integral kernel $A(x,y)$. Then the \emph{Fredholm determinant} of $A$ is given by
 \begin{equation}
  \det(\Id+A)_{L^2(\Lambda)}=\sum_{n=0}^\infty\frac{1}{n!} \int_{\Lambda^n}\mathrm{d} x_1\dots\mathrm{d} x_N\det_{1\leq i,j\leq n}[A(x_i,x_j)].
 \end{equation}
\end{defin}
Instead of using the series expansion, one can define the Fredholm determinant in a more abstract way for operators on a general separable Hilbert space $\mathcal{H}$, see~\cite{Sim00}. It can be seen as a natural generalization of the ordinary determinant as it satisfies $\det(\Id+A)=\prod_n(1+\lambda_n)$, with $\lambda_n$ being the eigenvalues of $A$, and has also the following properties:
\begin{itemize}
 \item Continuity, specifically: \begin{equation}|\det(\Id+A)-\det(\Id+B)|\leq||A-B||\exp(||A||_1+||B||_1+1),\end{equation}
 \item Multiplicativity: \begin{equation}\det(\Id+A+B+AB)=\det(\Id+A)\det(\Id+B),\end{equation}
 \item Sylvester's determinant theorem: \begin{equation}\det(\Id+AB)=\det(\Id+BA).\end{equation}
\end{itemize}
By the last identity, the Fredholm determinant is invariant under conjugations $A\mapsto U^{-1}AU$. Definition~\ref{defFredhD} can thus be extended to operators that are not necessarily trace class themselves, but have a conjugate that is trace class.

For a determinantal point process, the formula \eqref{eqhole} for the hole probability can be written as a Fredholm determinant:
\begin{equation}\begin{aligned}
 \Pb\left(\mu(X)=0\right)&=\sum_{n\geq0}\frac{(-1)^n}{n!}\int_{X^n}\det_{1\leq i,j\leq n}[K(x_i,x_j)]\mathrm{d} x_1\dots\mathrm{d} x_n\\
 &=\det(\Id-\Id_XK\Id_X)_{L^2(\Lambda)},
\end{aligned}\end{equation}
where $\Id_X$ denotes the projection operator on the set $X$. In the case $\Lambda=\R$ we recover the distribution function of the rightmost particle by choosing $X=(s,\infty)$.

In the case $\Lambda=S\times\R$, for a discrete set $S$, we have a point process with particles in each layer $\{r_k\}\times\R$, $r_k\in S$. The choice $X=\bigcup_k \{r_k\}\times(s_k,\infty)$ then gives the joint distribution of the rightmost particles in the layers $r_k$. For this choice of $X$ we use the shorthand $\Id_X=\chi_s(r_k,x)$. The integrals over the discrete measure can be written out explicitly as sums, resulting in the formula

\begin{equation}%\label{}
\begin{aligned}
	\Pb&\bigg(\bigcap_{k=1}^m\big\{x_\textrm{max}(r_k)\leq s_k\big\}\bigg)=\det(\Id-\chi_sK\chi_s)_{L^2(S\times\R)}\\
	&\quad=\sum_{n\geq0}\frac{(-1)^n}{n!}\sum_{i_1,\dots,i_n=1}^m\int_{\R^n} \det_{1\leq k,l\leq n}\left[K(r_{i_k},x_k;r_{i_l},x_l)\right] \prod_{k=1}^n \Id_{x_k>s_{i_k}}\mathrm{d} x_k.
\end{aligned}\end{equation}

\section{Correlation kernel}
Whenever a point process comes from a measure that is given by the product of two or more determinants in a certain way, the correlation functions are determinantal and there is an explicit formula for the correlation kernel. This has been discovered in~\cite{EM97} and generalized to various similar settings later on~\cite{FNH99,FS03,Jo03b}. This explicit formula, however, involves the inverse of a quite complicated matrix. Instead of finding this inverse, which is usually not feasible, one relies on the \emph{orthogonal polynomial method}, i.e.\ chooses functions in the right way, such that the resulting matrix is diagonal. A variant of this method is the basis of our analysis:
\begin{lem}[Corollary of Theorem~4.2~\cite{BF07}]\label{lemDetMeasure}
Assume we have a signed measure on $\{x_i^n,n=1,\dotsc,N,i=1,\dotsc,n\}$ given in the form,
\begin{equation}\label{Sasweight}
 \frac{1}{Z_N}\prod_{n=1}^{N} \det[\phi_n(x_i^{n-1},x_j^n)]_{1\leq i,j\leq n} \det[\Psi_{N-i}^{N}(x_{j}^N)]_{1\leq i,j \leq N},
\end{equation}
where $x_{n+1}^n$ are some ``virtual'' variables and $Z_N$ is a normalization constant. If $Z_N\neq 0$, then the correlation functions are determinantal.

To write down the kernel we need to introduce some notations. Define
\begin{equation}\label{Sasdef phi12}
\phi^{(n_1,n_2)}(x,y)=
\begin{cases} (\phi_{n_1+1} \ast \dotsb \ast \phi_{n_2})(x,y),& n_1<n_2,\\
0,& n_1\geq n_2,
\end{cases}.
\end{equation}
where $(a* b)(x,y)=\int_\R \mathrm{d} z\, a(x,z) b(z,y)$, and, for $1\leq n<N$,
\begin{equation}\label{Sasdef_psi}
\Psi_{n-j}^{n}(x) := (\phi^{(n,N)} * \Psi_{N-j}^{N})(y), \quad j=1,\dotsc,N.
\end{equation}
Then the functions
\begin{equation}
\{ \phi^{(0,n)}(x_1^0,x), \dots,\phi^{(n-2,n)}(x_{n-1}^{n-2},x), \phi_{n}(x_{n}^{n-1},x)\}
\end{equation}
are linearly independent and generate the $n$-dimensional space $V_n$. Define a set of functions $\{\Phi_{n-j}^{n}(x), j=1,\dotsc,n\}$ spanning $V_n$ defined by the orthogonality relations
\begin{equation}\label{Sasortho}
\int_\R \mathrm{d} x\, \Phi_{n-i}^n(x) \Psi_{n-j}^n(x) = \delta_{i,j}
\end{equation}
for $1\leq i,j\leq n$.

Further, if $\phi_n(x_n^{n-1},x)=c_n \Phi_0^{n}(x)$, for some $c_n\neq 0$, \mbox{$n=1,\dotsc,N$}, then the kernel takes the simple form
\begin{equation}\label{SasK}
K(n_1,x_1;n_2,x_2)= -\phi^{(n_1,n_2)}(x_1,x_2)+ \sum_{k=1}^{n_2} \Psi_{n_1-k}^{n_1}(x_1) \Phi_{n_2-k}^{n_2}(x_2).
\end{equation}
\end{lem}

\section{Strategy  for future proofs}
The basic concepts of the asymptotic behaviour of our system of one-sided reflected Brownian motions in the three fundamental initial conditions as well as the three mixed initial ones have quite some overlap.

The starting point is always the formula for the transition density provided by Proposition~\ref{PropWarren}. After that one can insert resp. integrate over the initial condition, subject to some simplifications using stationarity. This is trivial for packed and periodic initial conditions, but involves some subtle algebraic handling in the Poisson case. After obtaining the fixed time measure in this way, one has to introduce virtual particles to obtain a version of the measure that allows applying Lemma~\ref{lemDetMeasure}.

The resulting Fredholm determinant expression for the joint distribution of particles at a fixed time can then be analyzed asymptotically using steep descent. In order for the Fredholm determinants to converge, one has to show not only pointwise convergence of the kernel but also some uniform exponential bounds.

The situation is more complex for the Poisson case, as a determinantal structure exists only for a positive difference $\lambda-\rho$ of the Poisson densities on the positive and negative half-axis. The rigorous limit $\rho\to\lambda$ is the topic of Section~\ref{SectAnCont}.

\chapter{Airy processes}\label{secAiryP}

Airy processes arise through a scaling limit of our system of reflected Brownian motions.
Their detailed structure depends on the initial conditions. For the various Airy processes appearing in the text we list here the definitions and, in particular, discuss their interrelation.

\section{Elementary Airy processes}
\subsection{Airy\texorpdfstring{$_2$}{\_2} process}
The first appearance of an Airy-type process has been the Airy$_2$ process as the limit of the top layer in the polynuclear growth model~\cite{PS02}. In this initial paper, it has already been noted that this process is stationary, has continuous sample paths, its one-point distribution
is the GUE Tracy-Widom distribution, as well as that the correlation $\mathbbm{E}\left(\mathcal{A}_2(r)\mathcal{A}_2(0)\right)-\mathbbm{E}(\mathcal{A}_2(r))\mathbbm{E}(\mathcal{A}_2(0))$ is positive and decays as $r^{-2}$. Later it was proven that its sample paths are locally absolutely continuous to Brownian motion~\cite{CH11}, which implies they are H\"{o}lder continuous with exponent $\frac{1}{2}-$.

The Airy$_2$ process is defined by its finite-dimensional distribution function:
\begin{defin}[Airy$_2$ process]
The \emph{Airy$_2$ process, $\mathcal{A}_2$}, is the process with
$m$-point joint distributions at $r_1<r_2<\cdots<r_m$ given by
the Fredholm determinant

\begin{equation}
\Pb \Biggl(\bigcap_{k=1}^m \bigl\{
\mathcal{A}_2(r_k)\leq s_k +r_k^2 \bigr\} \Biggr)=
\det (\Id-\chi_s K_{\mathcal{A}_2}\chi_s
)_{L^2(\{
r_1,\ldots,r_m\}\times\R)},
\end{equation}
where $\chi_s(r_k,x)=\Id_{x>s_k}$. The correlation kernel $K_{\mathcal{A}_2}$ is given by
\begin{equation}\label{eqAiry2Def}
 K_{\mathcal{A}_2}(s_1,r_1;s_2,r_2)=-V_{r_1,r_2}(s_1,s_2)\Id_{r_1<r_2}+K_{r_1,r_2}(s_1,s_2),
\end{equation}
with
\begin{equation}\label{eqKernel2Def}
\begin{aligned}
  V_{r_1,r_2}(s_1,s_2)&=\frac{e^{-\frac{(s_2-s_1)^2}{4(r_2-r_1)}}}{\sqrt{4\pi (r_2-r_1)}}\\
  K_{r_1,r_2}(s_1,s_2)&=\frac{e^{\frac{2}{3}r_2^3+r_2s_2}}{e^{\frac{2}{3}r_1^3+r_1s_1}}
    \int_0^\infty\mathrm{d} x\, e^{x(r_2-r_1)}\Ai(r_1^2+s_1+x)\Ai(r_2^2+s_2+x).
\end{aligned}
\end{equation}
\end{defin}

\begin{remark}
The parts of the kernel, \eqref{eqKernel2Def}, have alternative representations:
\begin{equation}\label{contIntK}\begin{aligned}
	K_{r_1,r_2}(s_1,s_2)&=\frac{-1}{(2\pi\mathrm{i})^2}\int_{e^{-2\pi\mathrm{i}/3}\infty}^{e^{2\pi\mathrm{i}/3}\infty}\hspace{-2pt}\mathrm{d} W\int_{e^{\pi\mathrm{i}/3}\infty}^{e^{-\pi\mathrm{i}/3}\infty}\hspace{-2pt}\mathrm{d} Z\frac{e^{Z^3/3+r_2Z^2-s_2Z}}{e^{W^3/3+r_1W^2-s_1W}}\frac{1}{Z-W}\\
 V_{r_1,r_2}(s_1,s_2)&=\frac{e^{\frac{2}{3}r_2^3+r_2s_2}}{e^{\frac{2}{3}r_1^3+r_1s_1}}\int_\R\mathrm{d} x\, e^{-x(r_1-r_2)}\Ai(r_1^2+s_1+x)\Ai(r_2^2+s_2+x).
\end{aligned}\end{equation}
In the integral defining $K$, the path for $W$ and $Z$ do not have to intersect.
\end{remark}

\begin{remark}
Using the integral representation \eqref{contIntK} for $V$, inserting the parameter shift $s_i\to s_i-r_i^2$ and finally applying a conjugation, leads to another way to describe the Airy$_2$ process. With
 \begin{equation}\label{eqKAi2'}
 K'_{\mathcal{A}_2}(s_1,r_1;s_2,r_2)=\begin{cases}\int_0^\infty\mathrm{d} x\, e^{-x(r_1-r_2)}\Ai(s_1+x)\Ai(s_2+x)&\text{for }r_1\geq r_2\\-\int^0_{-\infty}\mathrm{d} x\, e^{-x(r_1-r_2)}\Ai(s_1+x)\Ai(s_2+x)&\text{for }r_1<r_2.\end{cases}
\end{equation}
its multi-dimensional distributions are given by
\begin{equation}
\Pb \Biggl(\bigcap_{k=1}^m \bigl\{
\mathcal{A}_2(r_k)\leq s_k \bigr\} \Biggr)=
\det (\Id-\chi_s K'_{\mathcal{A}_2}\chi_s
)_{L^2(\{
r_1,\ldots,r_m\}\times\R)}.
\end{equation}
\end{remark}

\subsection{Airy\texorpdfstring{$_1$}{\_1} process}
The Airy$_1$ process has been discovered three years later as the limit of the TASEP with flat initial condition~\cite{Sas05}. It is also stationary with the one\nobreakdash-point distribution now being the GOE Tracy-Widom distribution, i.e.\ \mbox{$\Pb({\mathcal A}_1(r)\leq s)=F_\textrm{GOE}(2s)$}, which has already been conjectured in the initial paper, and soon been confirmed~\cite{FS05b}. More recently, it was proven that the sample paths have locally Brownian fluctuations~\cite{QR12}, and are consequently H\"{o}lder continuous with exponent $\frac{1}{2}-$, too.

The Airy$_1$ process is also defined in terms of its finite-dimensional distributions:
\begin{defin}[Airy$_1$ process]
Let
 \begin{equation}\begin{aligned}
 K_{\mathcal{A}_1}(s_1,r_1;s_2,r_2) =&-V_{r_1,r_2}(s_1,s_2)\Id_{r_1<r_2}\\
&+\Ai \bigl(s_1+s_2+(r_2-r_1)^2
\bigr) e^{(r_2-r_1) (s_1+s_2)+
\frac{2}{3}(r_2-r_1)^3}.
\end{aligned}\end{equation}
The \emph{Airy$_1$ process, denoted by $\mathcal{A}_1$}, is the process with
$m$-point joint distributions at $r_1<r_2<\cdots<r_m$ given by
the Fredholm determinant
%
%e2.15 #&#
\begin{equation}
\Pb \Biggl(\bigcap_{k=1}^m \bigl\{
\mathcal{A}_1(r_k)\leq s_k \bigr\} \Biggr)=
\det (\Id-\chi_s K_{\mathcal{A}_1}\chi_s
)_{L^2(\{
r_1,\ldots,r_m\}\times\R)},
\end{equation}
where $\chi_s(r_k,x)=\Id_{x>s_k}$.
\end{defin}

\begin{remark}
 The kernel can also be stated in a different way.
Let therefore $B_0(x,y)=\Ai(x+y)$ and
$\Delta
$ be the one-dimensional Laplacian. Then the kernel $K_{\mathcal{A}_1}$ is given by
%
%e2.14 #&#
\begin{equation}\begin{aligned}
\label{defairykernel}
K_{\mathcal{A}_1}(s_1,r_1;s_2,r_2)=-\bigl(e^{(r_2-r_1)\Delta} \bigr) (s_1,s_2)\Id_{r_2>r_1}+ \bigl(e^{-r_1\Delta}B_0e^{r_2\Delta}
\bigr) (s_1,s_2).
\end{aligned}\end{equation}
The equivalence of these formulas is proven in Appendix A~\cite{BFPS06}.
\end{remark}

\subsection{Airy\texorpdfstring{$_\textrm{stat}$}{\_stat} process}
In spite of what the name might suggest, the Airy$_\textrm{stat}$ process is \emph{not} stationary. The name is derived from the fact that it arises as limit process of models started in their respective stationary initial condition. The one\nobreakdash-point distribution ${\mathcal A}_\textrm{stat}(0)$ has been identified in~\cite{BR00} as the limit of the PNG model, it has mean zero and is often called Baik-Rains distribution. The multi-point distribution has been discovered some time later in the TASEP~\cite{BFP09}.

The increments of the Airy$_\textrm{stat}$ process are identical to those of a Brownian motion, this is straightforward to see by indirect arguments:

Since $\{x_{n}(t), n \in \mathbb{Z}\}$ is a Poisson point process for every $t\geq0$,
the process $X_t^\textrm{stat}(r) - X_t^\textrm{stat}(0)$ is a scaled Poisson jump process up to a linear part and
\begin{equation}
 \lim_{t \to \infty}\big(X_t^\textrm{stat}(r)-X_t^\textrm{stat}(0)\big)\stackrel{d}{=}B(2r)\,.
\end{equation}
By Theorem~\ref{thmAsymp0} the limit process $\mathcal{A}_\mathrm{stat}(r) -\mathcal{A}_\mathrm{stat}(0)$ must also have the statistics of two-sided Brownian motion.
This property is not so easily inferred from the formulas in Definition~\ref{DefAstat}. We will provide a direct proof of this fact in Section~\ref{secGaussFl}.
The structure of the Airy$_\textrm{stat}$ is nevertheless quite rich, as these Brownian increments are non-trivially correlated with the random height shift $\mathcal{A}_\mathrm{stat}(0)$.

\begin{defin}[Airy$_\textrm{stat}$ process]\label{DefAstat}
Let $P_s$ be the projection operator on $[s,\infty)$ and $\bar{P}_s=\Id-P_s$ the one on $(-\infty,s)$. With $V_{r_1,r_2}(s_1,s_2)$ as in \eqref{eqKernel2Def}, define
\begin{equation}
	\mathcal{P}=\Id-\bar{P}_{s_1}V_{r_1,r_2}\bar{P}_{s_2}\cdots V_{r_{m-1},r_m}\bar{P}_{s_m}V_{r_m,r_1},
\end{equation}
as well as an operator $K$ with integral kernel
\begin{equation}
	K(s_1,s_2)=K_{r_1,r_1}(s_1,s_2)=e^{r_1(s_2-s_1)}\int_{r_1^2}^\infty\mathrm{d} x\, \Ai(s_1+x)\Ai(s_2+x).
\end{equation}
Further, define the functions
\begin{equation}\begin{aligned}
	\mathcal{R}&=s_1+e^{\frac{2}{3}r_1^3}\int_{s_1}^\infty \mathrm{d} x\int_x^\infty \mathrm{d} y\, \Ai(r_1^2+y)e^{r_1y},\\
	f^*(s)&=-e^{-\frac{2}{3}r_1^3}\int_s^\infty \mathrm{d} x\, \Ai(r_1^2+x)e^{-r_1x},\\
	g(s)&=1-e^{\frac{2}{3}r_1^3}\int_s^\infty \mathrm{d} x\, \Ai(r_1^2+x)e^{r_1x}.
\end{aligned}\end{equation}
With these definitions, set
\begin{equation}\label{eqGm} G_m(\vec{r},\vec{s})=\mathcal{R}-\left\langle(\Id-\mathcal{P}K)^{-1}\left(\mathcal{P}f^*+\mathcal{P}KP_{s_1}\mathbf{1}+(\mathcal{P}-P_{s_1})\mathbf{1}\right),g\right\rangle,
\end{equation}
where $\langle\cdot,\cdot\rangle$ denotes the inner product on $L^2(\R)$.
Then, the \emph{Airy$_\textrm{{stat}}$ process}, $\mathcal{A}_\mathrm{stat}$, is the process with $m$-point joint distributions at $r_1<r_2<\dots<r_m$ given by
\begin{equation}\label{eqAiryDef}
	\Pb\bigg(\bigcap_{k=1}^m\{\mathcal{A}_\mathrm{stat}(r_k)\leq s_k\}\bigg) = \sum_{i=1}^m\frac{\D}{\mathrm{d} s_i}\left(G_m(\vec{r},\vec{s})\det\left(\Id-\mathcal{P}K\right)_{L^2(\R)}\right).
\end{equation}
\end{defin}
\begin{remark}
In this definition there appears actually yet another representation for the joint distributions of the Airy$_2$ process:
\begin{equation}
 \Pb\bigg(\bigcap_{k=1}^m\{\mathcal{A}_2(r_k)\leq s_k+r_k^2\}\bigg)=\det(\Id-\mathcal{P}K)_{L^2(\R)}.
\end{equation}
This is actually the version first obtained in~\cite{PS02} for $m=2$. Equivalence of this formula to our definition is proven in~\cite{BCR13}.
\end{remark}
For well-definedness of the formula \eqref{eqGm} we need the following lemma:
\begin{lem}\label{lemInv}
 The operator $\Id-\mathcal{P}K$ is invertible.
\end{lem}
\begin{proof}
 We employ the same strategy as in Appendix~B~\cite{BFP09}. For that purpose we use the following equivalence
 \begin{equation}
  \det(\Id+A)\neq 0\ \Longleftrightarrow\ \Id+A \text{ is invertible}.
 \end{equation}
Let $s_\textrm{min}=\min_ks_k$.
 \begin{equation}\begin{aligned}
  \det&(\Id-\mathcal{P}K)=\Pb\bigg(\bigcap_{k=1}^m\{\mathcal{A}_2(r_k)-r_k^2\leq s_k\}\bigg)\\
  &\geq \Pb\bigg(\bigcap_{k=1}^m\{\mathcal{A}_2(r_k)-r_k^2\leq s_\textrm{min}\}\bigg)\geq \Pb\Big(\max_{r\in\R}(\mathcal{A}_2(r)-r^2)\leq s_\textrm{min}\Big)\\
  &=F_\textrm{GOE}(2^{2/3}s_\textrm{min})>0
  \end{aligned}\end{equation}
for any $s_\textrm{min}>-\infty$, where $F_\textrm{GOE}$ is the GOE Tracy-Widom distribution function. For the last equality see Section~\ref{secVarId}. The tails of the GOE Tracy-Widom distribution have been studied in great detail in various publications, see for instance~\cite{BBdF08}.
\end{proof}

Our Definition~\ref{DefAstat} is actually an alternative formula for the joint distributions of the Airy$_\textrm{stat}$ process compared to the one given in~\cite{BFP09} (Definition 1.1 and Theorem 1.2 therein). The main difference is that in~\cite{BFP09} the joint distributions are given in terms of a Fredholm determinant on $L^2(\{1,\ldots,m\}\times\R)$, while here we have a Fredholm determinant on $L^2(\R)$. A similar twist was already visible in~\cite{PS02} and has been generalized in~\cite{BCR13}.

\section{Crossover Processes}
In this section we introduce the three mixed Airy processes that are transition processes between two of the elementary Airy processes, as well as a new process with a parameter that interpolates between an elementary and a crossover Airy process.
\subsection{Airy\texorpdfstring{$_{2\to1}$}{\_2to1} process}
The Airy$_{2\to1}$ process has been discovered as the limit of the TASEP under half-periodic initial conditions~\cite{BFS07}. It is again defined in terms of its finite-dimensional distributions:
\begin{defin}[Airy$_{2\to1}$ process]\label{defAiry2to1}
The  Airy$_{2\to1}$ process, $\mathcal{A}_{2\to1}$, is the process with m-point joint distributions at $r_1<r_2<\dots<r_m$ given by
\begin{equation}%\label{eq3.6}
	\Pb\bigg(\bigcap_{k=1}^m\{\mathcal{A}_{2\to1}(r_k)\leq s_k+r_k^2\,\Id_{r_k\leq0}\}\bigg) = \det\left(\Id-\chi_s K_{\mathcal{A}_{2\to1}}\chi_s\right)_{L^2(\{r_1,\dots,r_m\}\times\R)},
\end{equation}
where $\chi_s(r_k,x)=\Id_{x>s_k}$ and the kernel $K_{\mathcal{A}_{2\to1}}$ is defined by
\begin{equation}\label{eqKAi2to1}
\begin{aligned}
	K_{\mathcal{A}_{2\to1}}&(r_1,s_1;r_2,s_2)=K_{\mathcal{A}_2}(s_1,r_1;s_2,r_2)\\
	&+\frac{e^{\frac{2}{3}r_2^3+r_2s_2}}{e^{\frac{2}{3}r_1^3+r_1s_1}}\int_0^\infty\mathrm{d} x\, e^{x(r_1+r_2)}\Ai(r_1^2+s_1-x)\Ai(r_2^2+s_2+x),
\end{aligned}
\end{equation}
with $K_{\mathcal{A}_2}$ as in \eqref{eqAiry2Def}. %, and
%\begin{equation}
%f_{r_1}(s_1)&=1-e^{-\frac{2}{3}r_1^3-r_1s_1}\int_0^\infty\mathrm{d} x\, \Ai(r_1^2+s_1+x)e^{-r_1x}
%\end{equation}
\end{defin}
The Airy$_{2\to1}$ process satisfies the limits:
\begin{equation}\begin{aligned}
 \lim_{w\to\infty}\mathcal{A}_{2\to1}(r+w)&=2^{-1/3}\mathcal{A}_1(2^{2/3}r),\\
 \lim_{w\to\infty}\mathcal{A}_{2\to1}(r-w)&=\mathcal{A}_2(r).
\end{aligned}\end{equation}

\begin{remark}
 There exists a contour integral representation, which will be used in proving Theorem~\ref{thmAsympHf}:
\begin{equation}\label{eqKAi2to1Int}\begin{aligned}
K_{\mathcal{A}_{2\to1}}(r_1,s_1;r_2,s_2)&=-V_{r_1,r_2}(s_1,s_2)\Id_{r_1<r_2}\\&
\quad+\frac{1}{(2\pi\mathrm{i})^2}\int_{\gamma_W} \mathrm{d} W\int_{\gamma_Z} \mathrm{d} Z\frac{e^{Z^3/3+r_2Z^2-s_2Z}}{e^{W^3/3+r_1W^2-s_1W}}\frac{2Z}{W^2-Z^2}.
\end{aligned}\end{equation}
The contours $\gamma_W\colon e^{-2\pi\mathrm{i}/3}\infty \to e^{2\pi\mathrm{i}/3}\infty$ and $\gamma_Z\colon e^{\pi\mathrm{i}/3}\infty \to e^{-\pi\mathrm{i}/3}\infty$ are chosen in such a way that both $\gamma_W$ and $-\gamma_W$ pass left of $\gamma_Z$. \eqref{eqKAi2to1} can be derived from \eqref{eqKAi2to1Int}  by noticing the identity
\begin{equation}
 \frac{2Z}{W^2-Z^2}=\frac{1}{W-Z}-\frac{1}{W+Z}=-\int_0^\infty\mathrm{d} x\, e^{x(W-Z)}-\int_0^\infty\mathrm{d} x\,e^{-x(W+Z)},
\end{equation}
and employing the definition of the Airy function.
\end{remark}

\subsection{Airy\texorpdfstring{$_{2\to\textrm{BM}}$}{\_2toBM} process}
The Airy$_{2\to\textrm{BM}}$ process has been discovered in~\cite{SI04}. Therein it was already shown that $\mathcal{A}_{2\to\textrm{BM}}(r)$ converges to the GUE Tracy Widom distribution for $r\to-\infty$ and to a Gaussian distribution as $r\to\infty$. They also identified the distribution at $r=0$, which is given by $\Pb(\mathcal{A}_{2\to\textrm{BM}}(0)\leq s)=\left(F_\textrm{GOE}(s)\right)^2$, i.e.\ the distribution of the maximum of two independent GOE-distributed random variables.

\begin{defin}[Airy$_{2\to\textrm{BM}}$ process]%\label{DefAiryStat}
The  Airy$_{2\to\textrm{BM}}$ process, $\mathcal{A}_{2\to\textrm{BM}}$, is the process with m-point joint distributions at $r_1<r_2<\dots<r_m$ given by
\begin{equation}%\label{eq3.6}
	\Pb\bigg(\bigcap_{k=1}^m\{\mathcal{A}_{2\to\textrm{BM}}(r_k)\leq s_k\}\bigg) = \det\left(\Id-\chi_s K_{\mathcal{A}_{2\to\textrm{BM}}}\chi_s\right)_{L^2(\{r_1,\dots,r_m\}\times\R)},
\end{equation}
where $\chi_s(r_k,x)=\Id_{x>s_k}$ and the kernel $K_{\mathcal{A}_{2\to\textrm{BM}}}$ is defined by
\begin{equation}\begin{aligned}
K_{\mathcal{A}_{2\to\textrm{BM}}}(r_1,s_1;r_2,s_2)&=K'_{\mathcal{A}_2}(r_1,s_1;r_2,s_2)\\&
\quad+\left(e^{-\frac{1}{3}r_1^3+r_1s_1}-\int_0^\infty\mathrm{d} x\, \Ai(s_1+x)e^{-r_1x}\right)\Ai(s_2).
\end{aligned}\end{equation}
Here, $K'_{\mathcal{A}_2}$ is defined as in (\ref{eqKAi2'}).%, and
%\begin{equation}
%f_{r_1}(s_1)&=1-e^{-\frac{2}{3}r_1^3-r_1s_1}\int_0^\infty\mathrm{d} x\, \Ai(r_1^2+s_1+x)e^{-r_1x}
%\end{equation}
\end{defin}

\subsection{Airy\texorpdfstring{$_{\textrm{BM}\to1}$}{\_BMto1} process}
The last one of the mixed Airy processes is the Airy$_{\textrm{BM}\to1}$ process, which first appeared in~\cite{BFS09}. The limit process that appears therein is actually more general, denoted by $\mathcal{A}_{2\to1,M,\kappa}$ with parameters $M\in\Z_{\geq0}$ and $\kappa\in\R$. In our case this more general process would appear through modifying the model asymptotically analyzed in Section~\ref{secStFl}: Instead of applying the drift $\rho$ to the particle $x_1(t)$ only, the first $M$ particles have drift $\rho$. Furthermore, $\rho$ is not equal to $1$, but scaled critically as $\rho=1-\kappa t^{-1/3}$.

We will not give this general definition, but only the version that appears in this work, i.e.\ $\mathcal{A}_{\textrm{BM}\to1}(r)=\mathcal{A}_{2\to1,1,0}(r)$.
\begin{defin}[Airy$_{\textrm{BM}\to1}$ process]\label{defAiBto1}
The  Airy$_{\textrm{BM}\to1}$ process, $\mathcal{A}_{\mathrm{BM}\to1}$, is the process with m-point joint distributions at $r_1<r_2<\dots<r_m$ given by
\begin{equation}%\label{eq3.6}
	\Pb\bigg(\bigcap_{k=1}^m\{\mathcal{A}_{\mathrm{BM}\to1}(r_k)\leq s_k\}\bigg) = \det\left(\Id-\chi_s K_{\mathcal{A}_{\mathrm{BM}\to1}}\chi_s\right)_{L^2(\{r_1,\dots,r_m\}\times\R)},
\end{equation}
where $\chi_s(r_k,x)=\Id_{x>s_k}$ and the kernel $K_{\mathcal{A}_{\mathrm{BM}\to1}}$ is defined by
\begin{equation}\begin{aligned}
&K_{\mathcal{A}_{\mathrm{BM}\to1}}(r_1,s_1;r_2,s_2)=K_{\mathcal{A}_{2\to1}}(r_1,s_1;r_2,s_2)\\&
\qquad+e^{-\frac{2}{3}r_1^3-r_1s_1}\Ai(s_1+r_1^2)\left(1-2e^{\frac{2}{3}r_2^3+r_2s_2}\int_0^\infty\mathrm{d} x\, \Ai(r_2^2+s_2+x)e^{r_2x}\right).
\end{aligned}\end{equation}
Here, $K_{\mathcal{A}_{2\to1}}$ is defined as in (\ref{eqKAi2to1}).%, and
%\begin{equation}
%f_{r_1}(s_1)&=1-e^{-\frac{2}{3}r_1^3-r_1s_1}\int_0^\infty\mathrm{d} x\, \Ai(r_1^2+s_1+x)e^{-r_1x}
%\end{equation}
\end{defin}
\begin{remark}
 In the original definition of the Airy$_{2\to1,M,\kappa}$ process, Definition~18~\cite{BFS09}, there is an issue with incompletely specified contours, that can be misleading. For the triple contour integral (5.8), with integration variables $u\in\Gamma_\kappa$, $w_1\in\gamma_1$ and $w_2\in\gamma_2$, it is only required that $\gamma_1$, $\gamma_2$ pass on the left of $\Gamma_\kappa$. In fact, for the formula to be correct, both $\gamma_2$ and $-\gamma_2$ have to pass to the left of $u$. As $\Gamma_\kappa$ is a loop around $\kappa$, this requires some intricate choices especially in the case $\kappa=0$. A good way to avoid this problem is by calculating one of the residues, so that $\gamma_2$ can be chosen in a more usual form, resulting in Definition~\ref{defAiBto1} after specifying $M=1$ and $\kappa=0$.
\end{remark}

\subsection{Finite-step Airy\texorpdfstring{$_\textrm{stat}$}{\_stat} process}
The finite-step Airy$_\textrm{stat}$ process is a new presumably universal limit process that appears in the course of proving Theorem~\ref{thmAsymp0}. It has a parameter $\delta$, which one can interpret as a step size, that gives the transition from the standard Airy$_\textrm{{stat}}$ process at $\delta=0$ towards the Airy$_{2\to\textrm{BM}}$ process as $t\to\infty$. It is defined again by its finite-dimensional distribution:
\begin{defin}[Finite-step Airy$_\textrm{stat}$ process]\label{DefAiryStat}
The finite-step Airy$_\textrm{{stat}}$ process with parameter $\delta>0$, $\mathcal{A}_\mathrm{stat}^{(\delta)}$, is the process with m-point joint distributions at $r_1<r_2<\dots<r_m$ given by
\begin{equation}\begin{aligned}\label{eq3.6}
	\Pb\bigg(\bigcap_{k=1}^m&\{\mathcal{A}_\mathrm{stat}^{(\delta)}(r_k)\leq s_k\}\bigg) \\
	&= \bigg(1+\frac{1}{\delta}\sum_{i=1}^m\frac{\D}{\mathrm{d} s_i}\bigg)\det\left(\Id-\chi_s K^\delta\chi_s\right)_{L^2(\{r_1,\dots,r_m\}\times\R)},
\end{aligned}\end{equation}
where $\chi_s(r_k,x)=\Id_{x>s_k}$ and the kernel $K^\delta$ is defined by
\begin{equation}\label{eqKdelta}
K^\delta(r_1,s_1;r_2,s_2)=-V_{r_1,r_2}(s_1,s_2)\Id_{r_1<r_2}+K_{r_1,r_2}(s_1,s_2)+\delta f_{r_1}(s_1)g_{r_2}(s_2).
\end{equation}
Here, $V_{r_1,r_2}$ and $K_{r_1,r_2}(s_1,s_2)$ are defined as in (\ref{eqKernel2Def}), and
\begin{equation}\label{eqKernelDef}
\begin{aligned}
		f_{r_1}(s_1)&=1-e^{-\frac{2}{3}r_1^3-r_1s_1}\int_0^\infty\mathrm{d} x\, \Ai(r_1^2+s_1+x)e^{-r_1x}\\
		g_{r_2}(s_2)&=e^{ \delta^3/3+r_2\delta^2-s_2\delta}-e^{\frac{2}{3}r_2^3+r_2s_2}\int_0^\infty\mathrm{d} x\, \Ai(r_2^2+s_2+x)e^{(\delta+r_2)x}.
\end{aligned}
\end{equation}
\end{defin}

\begin{remark}
Instead of integrals over Airy functions, the functions \eqref{eqKernelDef} can also be written as contour integrals:
\begin{equation}\label{contIntfg}\begin{aligned}
		f_{r_1}(s_1)&=\frac{1}{2\pi\mathrm{i}}\int_{e^{-2\pi\mathrm{i}/3}\infty,\text{ right of }0}^{e^{2\pi\mathrm{i}/3}\infty}   \mathrm{d} W\frac{e^{-(W^3/3+r_1W^2-s_1W)}}{W}\\
		g_{r_2}(s_2)&=\frac{1}{2\pi\mathrm{i}}\int_{e^{\pi\mathrm{i}/3}\infty,\text{ left of }\delta}^{e^{-\pi\mathrm{i}/3}\infty} \mathrm{d} Z\frac{e^{Z^3/3+r_2Z^2-s_2Z}}{Z-\delta}.
\end{aligned}\end{equation}
\end{remark}
\begin{remark}
The identity $\mathcal{A}^{(0)}_\textrm{stat}(r)=\mathcal{A}_\mathrm{stat}(r)$ is a consequence of Proposition~\ref{propAnalyt}. The limit
\begin{equation}\label{eq2.20bb}
\lim_{\delta\to \infty}\mathcal{A}^{(\delta)}_\textrm{stat}(r) = \mathcal{A}_{2\to\textrm{BM}}(r)-r^2
\end{equation}
can be seen as follows. Using the identity (D.3) from~\cite{FS05a} on $g_{r_2}(s_2)$ we obtain that
\begin{equation}\begin{aligned}
\lim_{\delta\to\infty}\delta\cdot g_{r_2}(s_2)&=e^{\frac{2}{3}r_2^3+r_2s_2}\lim_{\delta\to\infty}\int^0_{-\infty}\mathrm{d} x\, \Ai(r_2^2+s_2+x)\delta e^{(\delta+r_2)x}\\
&=e^{\frac{2}{3}r_2^3+r_2s_2}\lim_{\delta\to\infty}\int^0_{-\infty}\mathrm{d} y\, \Ai(r_2^2+s_2+y/\delta) e^{y(1+r_2/\delta)}\\
&=e^{\frac{2}{3}r_2^3+r_2s_2}\Ai(r_2^2+s_2).
\end{aligned}\end{equation}

This means precisely,
\begin{equation}
\lim_{\delta\to\infty} K^\delta(r_1,s_1;r_2,s_2)=\frac{e^{\frac{2}{3}r_2^3+r_2s_2}}{e^{\frac{2}{3}r_1^3+r_1s_1}} K_{{\cal A}_{2\to\textrm{BM}}}(r_1,s_1+r_1^2,;r_2,s_2+r_2^2),
\end{equation}
Finally, taking $\delta\to\infty$ in (\ref{eq3.6}) implies that all the terms with the derivatives vanish, giving (\ref{eq2.20bb}).
\end{remark}
\section{Variational identities}\label{secVarId}
In the work~\cite{Jo03b}, an interesting variational identity was given that connects the Airy$_2$ and the Airy$_1$ process:
\begin{equation}
 \Pb\left(\sup_{x\in\R}\left\{\mathcal{A}_2(x)-x^2\right\}\leq s\right)=F_\textrm{GOE}(2^{2/3}s).
\end{equation}
Its proof was however very indirect, and it took some time until the first direct proof appeared~\cite{CQR11}. It was conjectured that identities of this structure hold for all Airy-type processes~\cite{QR13}, as the Airy$_2$ process corresponds to delta initial conditions for the KPZ equation, and it should be possible to derive all other initial profiles by integrating this delta initial condition.

The first proof of all of these identities was given recently in~\cite{CLW14} through studying TASEP with general initial conditions:
\begin{align}
   \max_{x \in \R} \left( \mathcal{A}_2(x) - (x - r)^2 \right) \stackrel{d}{=} {}& 2^{1/3} \mathcal{A}_1(2^{-2/3} r)  %\label{eq:flat_limit}
  \\
   \max_{x \in \R} \left( \mathcal{A}_2(x) - (x - r)^2 + \sqrt{2}B(x) \right) \stackrel{d}{=} {}& \mathcal{A}_\mathrm{stat}(r) , %\label{eq:Bernoulli_limit}
  \\
   \max_{x \ge 0} \left( \mathcal{A}_2(x) - (x -  r)^2 \right) \stackrel{d}{=} {}& \mathcal{A}_{2 \to 1}(r) - r^2 \Id_{r \leq 0}, %\label{eq:step_flat_limit}
  \\
   \max_{x \ge 0} \left( \mathcal{A}_2(x) - (x-r)^2 +\sqrt{2} B(x) \right) \stackrel{d}{=} {}& \mathcal{A}_\mathrm{BM \to 2}(-r) - r^2, %\label{eq:step_Bernoulli_limit}
  \\
   \max_{x \in \R} \left( \mathcal{A}_2(x) - (x-r)^2 +\sqrt{2} \Id_{s\ge 0} B(x) \right) \stackrel{d}{=} {}& \mathcal{A}_{2 \to 1, 1, 0}(-r) - r^2 \Id_{r \geq 0}. %\label{eq:flat_Bernoulli_limit}
\end{align}
Notice that all equalities hold for the one-point distribution only.

\section{Gaussian fluctuations of the Airy\texorpdfstring{$_\textrm{stat}$}{\_stat} process}\label{secGaussFl}
In this section we show that the Airy$_\textrm{stat}$ process has Brownian increments for nonnegative arguments. As a rigorous proof of this fact is already established by indirect arguments, we do not care for full mathematical precision, but rather give a sketch of a proof:
\begin{thm}\label{thmGFl}
 Let $0\leq r_1<r_2<\dots<r_m$. Then
 \begin{equation}
  \Pb\bigg(\bigcap_{k=2}^{m}\{\mathcal{A}_\mathrm{stat}(r_{k})-\mathcal{A}_\mathrm{stat}(r_{k-1})\in\D\sigma_k\}\bigg) =\prod_{k=2}^{m}\frac{e^{-\sigma_k^2/4(r_k-r_{k-1})}}{\sqrt{4\pi(r_k-r_{k-1})}}\,\D\vec\sigma.
 \end{equation}
\end{thm}

\begin{proof}
 Without loss of generality we can assume that $r_1=0$. Denoting the partial derivative with respect to the $i$-th coordinate by $\partial_i$, we have
 \begin{equation}
	\Pb\bigg(\bigcap_{k=1}^m\{\mathcal{A}_\mathrm{stat}(r_k)\leq s_k\}\bigg) = \sum_{i=1}^m\partial_i\Lambda\left(s_1,\dots,s_m\right),
\end{equation}
with
 \begin{equation}
	\Lambda\left(s_1,\dots,s_m\right)=G_m(\vec{r},\vec{s})\det\left(\Id-\mathcal{P}K\right)_{L^2(\R)}.
\end{equation}
With a small abuse of notations, in what follows we will write
\begin{equation}
\Pb\bigg(\bigcap_{k=1}^m\{\mathcal{A}_\mathrm{stat}(r_k)\in\mathrm{d} s_k\}\bigg) \equiv \Pb\bigg(\bigcap_{k=1}^m\{\mathcal{A}_\mathrm{stat}(r_k)=s_k\}\bigg) \mathrm{d} s_1\cdots \mathrm{d} s_m.
\end{equation}
Then,
\begin{equation}\label{eq8.3}
\Pb\bigg(\bigcap_{k=1}^m\{\mathcal{A}_\mathrm{stat}(r_k)=s_k\}\bigg) = \prod_{i=1}^m\partial_i\sum_{j=1}^m\partial_j\Lambda\left(s_1,\dots,s_m\right).
\end{equation}
Inserting (\ref{eq8.3}) into
 \begin{equation}\label{eq8.3b}
 \begin{aligned}
	\Pb&\bigg(\bigcap_{k=2}^{m}\{\mathcal{A}_\mathrm{stat}(r_{k})-\mathcal{A}_\mathrm{stat}(r_{k-1})= \sigma_k\}\bigg) \\
	&=\int_\R\D\sigma_1  \Pb\bigg(\bigcap_{k=1}^m\{\mathcal{A}_\mathrm{stat}(r_k)= \sigma_1+\dots+\sigma_k\}\bigg)
\end{aligned}\end{equation}
we get
 \begin{equation}\label{eq8.4}
 \begin{aligned}
(\ref{eq8.3b}) &=\int_\R\D\sigma_1\bigg(\prod_{i=1}^m\partial_i\sum_{j=1}^m\partial_j\bigg)\Lambda\left(\sigma_1,\sigma_1+\sigma_2,\dots,\sigma_1+\dots+\sigma_m\right)\\ &=\int_\R\D\sigma_1\frac{\D}{\D\sigma_1}\bigg(\prod_{i=1}^m\partial_i\bigg)\Lambda\left(\sigma_1,\sigma_1+\sigma_2,\dots,\sigma_1+\dots+\sigma_m\right)\\
&=\bigg(\prod_{i=1}^m\partial_i\bigg)\Lambda\left(\sigma_1,\sigma_1+\sigma_2,\dots,\sigma_1+\dots+\sigma_m\right)\bigg|_{\sigma_1=-\infty}^{\sigma_1=\infty}.
\end{aligned}\end{equation}
We therefore have to study the asymptotics of $\Lambda$ as $\sigma_1\to\pm\infty$.

First we decompose $\Lambda$ as
\begin{equation}\begin{aligned}
 \Lambda&=\Lambda_1+\Lambda_2,\\
 \Lambda_1&:=\left(\mathcal{R}-1\right)\det\left(\Id-\mathcal{P}K\right)_{L^2(\R)},\\
 \Lambda_2&:=\det\Big(\Id-\mathcal{P}K-\left(\mathcal{P}f^*+\mathcal{P}KP_{s_1}\mathbf{1}+(\mathcal{P}-P_{s_1})\mathbf{1}\right)\otimes g\Big)_{L^2(\R)}.
\end{aligned}\end{equation}
Since $r_1=0$ some functions simplify as
\begin{equation}\begin{aligned}
	\mathcal{R}&=s_1+\int_{s_1}^\infty \mathrm{d} x\int_x^\infty \mathrm{d} y\, \Ai(y),\\
	f^*(s)&=-\int_s^\infty \mathrm{d} x\, \Ai(x),\quad g(s)=1-\int_s^\infty \mathrm{d} x\, \Ai(x),\\
	K(s_1,s_2)&=\int_0^\infty\mathrm{d} x\, \Ai(s_1+x)\Ai(s_2+x),
\end{aligned}\end{equation}
where we used the identity (D.3) from~\cite{FS05a}.
Now consider $\Lambda_1$.
\begin{equation}
 \bigg(\prod_{i=1}^m\partial_i\bigg)\Lambda_1(\vec{s})=(\mathcal{R}-1)\bigg(\prod_{i=1}^m\partial_i\bigg)
 \det(\Id-\mathcal{P}K)+\partial_1\mathcal{R}\bigg(\prod_{i=2}^m\partial_i\bigg)\det(\Id-\mathcal{P}K).
\end{equation}
Regarding the first term, notice that the multiple derivative of the Fredholm determinant gives exactly the multipoint density of the Airy$_2$ process, which is known to decay exponentially for both large positive and negative arguments. This exponential decay dominates over the linear growth of $\mathcal{R}$. Similarly, the $(m-1)$-fold derivative is smaller the $(m-1)$-point density of the Airy$_2$ process, so this contribution vanishes in the limit, too.

Continuing to $\Lambda_2$, using $f^*=-K\mathbf{1}$, we first simplify the expression
\begin{equation}
 \Lambda_2=\det\Big(\Id-\mathcal{P}K+\left(\mathcal{P}K\bar{P}_{s_1}\mathbf{1}-(\mathcal{P}-P_{s_1})\mathbf{1}\right)\otimes g\Big)_{L^2(\R)}.
\end{equation}
We introduce the shift operator $S$, $(Sf)(x)=f(x+\sigma_1)$, which satisfies $SV_{r_i,r_j}S^{-1}=V_{r_i,r_j}$ and $P_{a+\sigma_1}=S^{-1}P_aS$, and consequently also
\begin{equation}
 \Id-\bar{P}_{s_1+\sigma_1}V_{r_1,r_2}\bar{P}_{s_2+\sigma_1}\cdots V_{r_{m-1},r_m}\bar{P}_{s_m+\sigma_1}V_{r_m,r_1}=S^{-1}\mathcal{P}S.
\end{equation}
Using $\det(\Id-AB)=\det(\Id-BA)$, we have
\begin{equation}
 \Lambda_2(\vec{s}+\sigma_1)=\det\Big(\Id-\mathcal{P}SKS^{-1}+\left(\mathcal{P}SKS^{-1}\bar{P}_{s_1}\mathbf{1}-(\mathcal{P}-P_{s_1})\mathbf{1}\right)\otimes Sg\Big)_{L^2(\R)}.
\end{equation}
Now the dependence on the vector $\vec{s}$ is only in the projection operators, while the dependence on $\sigma_1$ is only in these two operators:
\begin{equation}\begin{aligned}
	(Sg)(s)&=\int_{-\infty}^{s+\sigma_1} \mathrm{d} x\, \Ai(x),\\
	(SKS^{-1})(s_1,s_2)&=\int_{\sigma_1}^\infty\mathrm{d} x\, \Ai(s_1+x)\Ai(s_2+x).
\end{aligned}\end{equation}
For large $\sigma_1$, we have $Sg\to\mathbf{1}$ and $SKS^{-1}\to0$ (both strong types of convergence from the superexponential Airy decay). So
\begin{equation}
 \lim_{\sigma_1\to\infty}\Lambda_2(\vec{s}+\sigma_1)=\det\left(\Id-(\mathcal{P}-P_{s_1})\mathbf{1}\otimes \mathbf{1}\right)_{L^2(\R)}=1-\langle(\mathcal{P}-P_{s_1})\mathbf{1}, \mathbf{1}\rangle_{L^2(\R)}
\end{equation}
Applying the expansion \eqref{eqPExp}, we arrive at:
\begin{equation}\label{eq8.13}\begin{aligned}
 \prod_{i=1}^m\partial_i\lim_{\sigma_1\to\infty}\Lambda_2(\vec{s}+\sigma_1)&=
 -\prod_{i=1}^m\partial_i\sum_{k=2}^m\langle\bar{P}_{s_1}V_{r_1,r_2}\dots \bar{P}_{s_{k-1}}V_{r_{k-1},r_k}P_{s_k}\mathbf{1}, \mathbf{1}\rangle\\
 &=-\prod_{i=1}^m\partial_i\langle\bar{P}_{s_1}V_{r_1,r_2}\dots \bar{P}_{s_{m-1}}V_{r_{m-1},r_m}P_{s_m}\mathbf{1}, \mathbf{1}\rangle
\end{aligned}\end{equation}
Writing out this scalar product and applying the fundamental theorem of calculus leads to:
\begin{equation}
 \eqref{eq8.13}=V_{r_1,r_2}(s_1,s_2)V_{r_2,r_3}(s_2,s_3)\dots V_{r_{m-1},r_m}(s_{m-1},s_m),
\end{equation}
which is the desired Gaussian density after setting $s_i=\sum_{k=2}^i\sigma_k$ as in \eqref{eq8.4}.

For large negative $\sigma_1$, we have $Sg\to\mathbf{0}$ and $SKS^{-1}\to\Id$. The rank one contribution is thus
\begin{equation}
 \left(\mathcal{P}\bar{P}_{s_1}\mathbf{1}-(\mathcal{P}-P_{s_1})\mathbf{1}\right)\otimes 0.
\end{equation}
We have to be somewhat careful here, as the convergence is weak (only pointwise) and $(Sg)$ is not even $L^2$-integrable. But the first factor decays superexponentially on both sides for finite $\sigma_1$ and also in the limiting case $\mathcal{P}\bar{P}_{s_1}\mathbf{1}-(\mathcal{P}-P_{s_1})\mathbf{1}=(1-\mathcal{P})P_{s_1}\mathbf{1}$, so one should be able to derive nice convergence properties. Neglecting this rank one contribution we are left with
\begin{equation}
 \lim_{\sigma_1\to-\infty}\Lambda_2(\vec{s}+\sigma_1)=\det\left(\Id-\mathcal{P}\Id\right)_{L^2(\R)}=0.
\end{equation}
\end{proof}

\chapter{Packed and Periodic initial conditions}\label{secPP}
Packed initial conditions turn out to be the most readily accessible case. Although such initial conditions have been studied extensively in the literature, we provide here a complete proof, introducing essentially all methods required later on in more complicated situations. The second treatable case of deterministic initial data are periodic initial conditions.
Here we will be rather brief, in stating only the main results, since the tools
used in~\cite{FSW13} are comparable to the ones employed in case of packed initial conditions.

\section{Packed initial conditions}\label{secStep}

We start with the simplest case of packed initial conditions. Our first result is a determinantal expression for the fixed time distribution.

\begin{prop}\label{propStepKernel}
Let $\{x_n(t),n\geq1\}$ be the system of one-sided reflected Brownian motions with initial condition $\vec{x}(0)=\vec{\zeta}^\textrm{packed}$. Then for any finite subset $S$ of $\Z_{>0}$, it holds
\begin{equation}%\label{eq33}
\Pb\bigg(\bigcap_{n\in S} \{x_n(t)\leq a_n\}\bigg)=\det(\Id-\chi_a \mathcal{K}_\textrm{packed} \chi_a)_{L^2(S\times \R)},
\end{equation}
where $\chi_a(n,\xi)=\Id_{\xi>a_n}$. The kernel $\mathcal{K}_\textrm{packed}$ is given by
\begin{equation}\label{eqKtStep}
\mathcal{K}_\textrm{packed}(n_1,\xi_1;n_2,\xi_2)=-\phi_{n_1,n_2}(\xi_1,\xi_2)\Id_{n_2>n_1}+\mathcal{K}_0(n_1,\xi_1;n_2,\xi_2),
\end{equation}
with
\begin{equation}
\begin{aligned}	\phi_{n_1,n_2}(\xi_1,\xi_2)&=\frac{(\xi_2-\xi_1)^{n_2-n_1-1}}{(n_2-n_1-1)!}\Id_{\xi_1\leq \xi_2}\\
\mathcal{K}_0(n_1,\xi_1;n_2,\xi_2)&=\frac{1}{(2\pi\mathrm{i})^2}\int_{\mathrm{i}\R-\e}\mathrm{d} w\,\oint_{\Gamma_0}\mathrm{d} z\,\frac{e^{{t} w^2/2+\xi_1w}}{e^{{t} z^2/2+\xi_2 z}}\frac{(-w)^{n_1}}{(-z)^{n_2}}\frac{1}{w-z}.
\end{aligned}
\end{equation}
Here $\Gamma_0$ is a simple loop around $0$ and $\e>0$ is chosen large enough that both contours do not touch.
\end{prop}

This proposition is the basis for proving the asymptotic theorem:
\begin{thm}\label{thmAsympStep}
With $\{x_n(t),n\geq1\}$ being the system of one-sided reflected Brownian motions with initial condition $\vec{x}(0)=\vec{\zeta}^\textrm{packed}$, define the rescaled process
\begin{equation}\label{eqStepScaledProcess}
	r\mapsto X_t^\textrm{packed}(r) = t^{-1/3}\big(x_{\lfloor t+2rt^{2/3}\rfloor}(t)-2t-2rt^{2/3} \big).\,
\end{equation}
In the sense of finite-dimensional distributions,
\begin{equation}
	\lim_{t\to\infty}X_t^\textrm{packed}(r)\stackrel{d}{=}\mathcal{A}_2(r)-r^2.
\end{equation}
\end{thm}

\subsection{Determinantal structure}
%\section{Transition density for fixed initial positions}
The first step in proving Proposition~\ref{propStepKernel} is a formula for the transition density of a finite system of one-sided reflected Brownian motions. It generalizes Proposition 4.1~\cite{FSW13}, which has been first shown in~\cite{War07}, to the case of non-zero drifts.
\begin{prop}\label{PropWarren}
Let $W^N=\{\vec{x}\in\R^N|x_1\leq\dots\leq x_N\}$ be the Weyl chamber. The transition probability density of $N$ one-sided reflected Brownian motions with drift $\vec{\mu}$ from $\vec{x}(0)=\vec{\zeta}\in W^N$  to \mbox{$\vec{x}({t})=\vec{\xi}\in W^N$} at time ${t}$ has a continuous version, which is given as follows:
\begin{equation}\label{transDens}
	\Pb\left(\vec{x}({t})\in \D\vec{\xi}\,|\vec{x}(0)=\vec{\zeta}\,\right)=r_{t}(\vec{\zeta},\vec{\xi})\,\D\vec{\xi},
\end{equation}
where
\begin{equation}\label{eqrtau}
r_{t}(\vec{\zeta},\vec{\xi})=\bigg(\prod_{n=1}^Ne^{\mu_n(\xi_n-\zeta_n)-{t}\mu_n^2/2}\bigg)\det_{1\leq k,l\leq N}[F_{k,l}(\xi_{N+1-l}-\zeta_{N+1-k},{t})],
\end{equation}
and
\begin{equation}\label{eqFkl}
	F_{k,l}(\xi,{t})=\frac{1}{2\pi\mathrm{i}}\int_{\mathrm{i}\R+\mu}\mathrm{d} w\,e^{{t} w^2/2+\xi w}\frac{\prod_{i=1}^{k-1}(w+\mu_{N+1-i})}{\prod_{i=1}^{l-1}(w+\mu_{N+1-i})},
\end{equation}
with $\mu>-\min\{\mu_1,\dots,\mu_N\}$.
\end{prop}
\begin{proof}
We follow the proof of Proposition~8 in~\cite{War07}. The strategy is to show that the transition density satisfies three equations, the backwards equation, boundary condition and initial condition, the latter one being contained in Lemma~\ref{lemInCond}. These equations are then used for It\^o's formula to prove that it indeed is the transition density.

We start with the backwards equation and boundary condition:
\begin{align}
	\frac{\partial r_{t}}{\partial{t}}&=\sum_{n=1}^N\left(\frac{1}{2}\frac{\partial^2}{\partial\zeta_n^2}+\mu_n\frac{\partial }{\partial\zeta_n}\right)r_{t}.\label{eqBackwards}\\
	\frac{\partial r_{t}}{\partial\zeta_i}&=0,\qquad\text{whenever } \zeta_i=\zeta_{i-1},\ 2\leq i\leq N\label{eqBoundCond}
\end{align}
To obtain \eqref{eqBoundCond}, move the prefactor $e^{-\mu_{i}\zeta_{i}}$ inside the integral in the \mbox{$(N+1-i)$-}th row of the determinant and notice that the differential operator transforms $F_{k,l}$ into $-F_{k+1,l}$.
 Consequently, $\zeta_i=\zeta_{i-1}$ implies the $(N+1-i)$-th being the negative of the $(N+2-i)$-th row.
\eqref{eqBackwards} can be seen from the computation
\begin{equation}
	\frac{\partial r_{t}}{\partial{t}}=\frac{1}{2}\sum_{n=1}^N\left(-\mu_n^2+e^{-\mu_n\zeta_n}\frac{\partial^2 }{\partial\zeta_n^2}e^{\mu_n\zeta_n}\right)r_{t}.
\end{equation}

Let \mbox{$f\colon W^N\rightarrow\R$} be a $C^\infty$ function, whose support is compact and has a distance of at least some $\e>0$ to the boundary of $W^N$. Define a function $F\colon(0,\infty)\times W^N\rightarrow\R$ as
\begin{equation}
	F({t},\vec{\zeta}\,)=\int_{W^N}r_{t}(\vec{\zeta},\vec{\xi})f(\vec{\xi})\,\mathrm{d} \vec{\xi}.
\end{equation}
The previous identities \eqref{eqBoundCond} and \eqref{eqBackwards} carry over to the function $F$ in the form of:
\begin{align}
	\frac{\partial F}{\partial\zeta_i}&=0,\qquad\text{for } \zeta_i=\zeta_{i-1},\ 2\leq i\leq N\label{eqBoundCond2}\\
	\frac{\partial F}{\partial{t}}&=\sum_{n=1}^N\left(\frac{1}{2}\frac{\partial^2}{\partial\zeta_n^2}+\mu_n\frac{\partial }{\partial\zeta_n}\right)F.\label{eqBackwards2}
\end{align}
Our processes satisfy $x_n({t})=\zeta_n+\mu_n{t}+B_n({t})+L^n({t})$, where $B_n$ are independent Brownian motions, $L^1\equiv0$ and $L^n$, $n\geq2$, are continuous non-decreasing processes increasing only when $x_n({t})=x_{n-1}({t})$. In fact, $L^n$ is twice the semimartingale local time at zero of $x_n-x_{n-1}$. Now fix some $\e>0$, $T>0$, define a process $F(T+\e-{t},\vec{x}({t}))$ for ${t}\in[0,T]$ and apply It\^o's formula:
\begin{equation}\label{eqIto}\begin{aligned}
	F&\left(T+\e-{t},\vec{x}({t})\right)\\&=F\left(T+\e,\vec{x}(0)\right)+\int_0^{t}-\frac{\partial}{\partial s}F\left(T+\e-s,\vec{x}(s)\right)\mathrm{d} s
	\\&\quad+\sum_{n=1}^N\int_0^{t}\frac{\partial}{\partial\zeta_n}F\left(T+\e-s,\vec{x}(s)\right)\mathrm{d} x_n(s)
	\\&\quad+\frac{1}{2}\sum_{m,n=1}^N\int_0^{t}\frac{\partial^2}{\partial\zeta_m\partial\zeta_n}F\left(T+\e-s,\vec{x}(s)\right)\mathrm{d} \left\langle x_m(s),x_n(s)\right\rangle.
\end{aligned}\end{equation}
From the definition it follows that $\mathrm{d} x_n(t)=\mu_n\mathrm{d} t +\mathrm{d} B_n(t)+\mathrm{d} L^n(t)$ and
\begin{equation}
	\mathrm{d} \left\langle x_m(t),x_n(t)\right\rangle=\mathrm{d} \left\langle B_m(t),B_n(t)\right\rangle=\delta_{m,n}\mathrm{d} t,
\end{equation}
because continuous functions of finite variation do not contribute to the quadratic variation. Inserting the differentials, by \eqref{eqBackwards2} the integrals with respect to $\mathrm{d} s$ integrals cancel, which results in:
\begin{equation}\begin{aligned}
	\eqref{eqIto}&=F\left(T+\e,\vec{x}(0)\right)+\sum_{n=1}^N\int_0^{t}\frac{\partial}{\partial\zeta_n}F\left(T+\e-s,\vec{x}(s)\right)\mathrm{d} B_n(s)
	\\&\quad+\sum_{n=1}^N\int_0^{t}\frac{\partial}{\partial\zeta_n}F\left(T+\e-s,\vec{x}(s)\right)\mathrm{d} L^n(s).
\end{aligned}\end{equation}
Since the measure $\mathrm{d} L^n(t)$ is supported on $\{x_n(t)=x_{n-1}(t)\}$, where the spatial derivative of $F$ is zero (see (\ref{eqBoundCond2})), the last term vanishes, too. So $F\left(T+\e-{t},\vec{x}({t})\right)$ is a local martingale and, being bounded, even a true martingale. In particular its expectation is constant, i.e.:
\begin{equation}
	F(T+\e,\vec{\zeta}\,)=\mathbbm{E}\left[F\left(T+\e,\vec{x}(0)\right)\right]=\mathbbm{E}\left[F\left(\e,\vec{x}(T)\right)\right].
\end{equation}
Applying Lemma~\ref{lemInCond} we can take the limit $\e\to0$, leading to
\begin{equation}
	F(T,\vec{\zeta}\,)=\mathbbm{E}\left[f\left(\vec{x}(T)\right)\right].
\end{equation}
Because of the assumptions we made on $f$ it is still possible that the distribution of $\vec{x}(T)$ has positive measure on the boundary. We thus have to show that $r_{t}(\vec{\zeta},\vec{\xi})$ is normalized over the interior of the Weyl chamber.

Start by integrating \eqref{eqrtau} over $\xi_N\in[\xi_{N-1},\infty)$. Pull the prefactor indexed by $n=N$ as well as the integration inside the $l=1$ column of the determinant. The $(k,1)$ entry is then given by:
\begin{equation}
 \begin{aligned}
  &e^{-\mu_N\zeta_N-{t}\mu_N^2/2}\int^{\infty}_{\xi_{N-1}}\D\xi_N e^{\mu_N\xi_N}F_{k,1}(\xi_N-\zeta_{N+1-k},{t})\\
  &\quad=e^{-\mu_N\zeta_N-{t}\mu_N^2/2}  e^{\mu_N x}F_{k,2}(x-\zeta_{N+1-k},{t})\Big|^{x=\infty}_{x=\xi_{N-1}}.
 \end{aligned}
\end{equation}
The contribution from $x=\xi_{N-1}$ is a constant multiple of the second column and thus cancels out. The remaining terms are zero for $k\geq2$, since all these functions $F_{k,2}$ have Gaussian decay. The only non-vanishing term comes from $k=1$ and returns exactly $1$ by an elementary residue calculation.

The determinant can thus be reduced to the index set $2\leq k,l\leq N$. Successively carrying out the integrations of the remaining variables in the same way, we arrive at the claimed normalization. This concludes the proof.
\end{proof}
\begin{lem}\label{lemInCond}
For fixed $\vec{\zeta}\in W^N$, the transition density $r_{t}(\vec{\zeta},\vec{\xi})$ as given by \eqref{eqrtau}, satisfies
\begin{equation}
	\lim_{{t}\to0}\int_{W^N}r_{t}(\vec{\zeta},\vec{\xi})f(\xi)\,\mathrm{d} \vec{\xi}=f(\vec{\zeta})
\end{equation}
for any $C^\infty$ function \mbox{$f\colon W^N\rightarrow\R$}, whose support is compact and has a distance of at least some $\e>0$ to the boundary of $W^N$.
\end{lem}
\begin{proof}
At first consider the contribution to the determinant in \eqref{eqrtau} coming from the diagonal. For $k=l$ the products in \eqref{eqFkl} cancel out, so we are left with a simple Gaussian density. This contribution is thus given by the multidimensional heat kernel, which is well known to converge to the delta distribution. The remaining task is to prove that for all other permutations the integral vanishes in the limit.

Let $\sigma$ be such a permutation. Its contribution is
\begin{equation}\label{eqL2.2a}
	\int_{\R^N}\D\vec{\xi}\,f(\vec{\xi})\prod_{k=1}^NF_{k,\sigma(k)}(\xi_{N+1-\sigma(k)}-\zeta_{N+1-k},{t}),
\end{equation}
where we have extended the domain of $f$ to $\R^N$, being identically zero outside of $W_N$. We also omitted the prefactor since it is bounded for $\xi$ in the compact domain of $f$.

There exist $i<j$ with $\sigma(j)\leq i<\sigma(i)$. Let
\begin{equation}\begin{aligned}
	\widetilde{W}_1&=\{\vec{\xi}\in \R^N\colon \xi_{N+1-\sigma(i)}-\zeta_{N+1-i}<-\e/2\}\\
	\widetilde{W}_2&=\{\vec{\xi}\in \R^N\colon \xi_{N+1-\sigma(j)}-\zeta_{N+1-j}>\e/2\}.
\end{aligned}\end{equation}
It is enough to restrict the area of integration to these two sets, since on the complement of $\widetilde{W}_1\cup\widetilde{W}_2$, we have
\begin{equation}
	\xi_{N+1-\sigma(i)}\geq\zeta_{N+1-i}-\e/2\geq\zeta_{N+1-j}-\e/2\geq\xi_{N+1-\sigma(j)}-\e,
\end{equation}
so we are not inside the support of $f$.

We start with the contribution coming from $\widetilde{W}_1$. Notice that by
\begin{equation}
	F_{k,l}(\xi,{t})=e^{-\xi\mu_{N+1-l}}\frac{\D}{\D\xi}\Big(e^{\xi\mu_{N+1-l}}F_{k,l+1}(\xi,{t})\Big),
\end{equation}
all functions $F_{k,l}$ with $k>l$ can be written as iterated derivatives of $F_{k,k}$ and some exponential functions. For each $k\neq i$ with $k>\sigma(k)$ we write $F_{k,\sigma(k)}$ in this way and then use partial integration to move the exponential factors and derivatives onto $f$. The result is
\begin{equation}\begin{aligned}\label{eq2.20}
	\int_{\widetilde{W}_1}\D\vec{\xi}\,&\tilde{f}(\vec{\xi})F_{i,\sigma(i)}(\xi_{N+1-\sigma(i)}-\zeta_{N+1-i},{t})\\&\prod_{k\neq i}F_{k,\max\{k,\sigma(k)\}}(\xi_{N+1-\sigma(k)}-\zeta_{N+1-k},{t})
\end{aligned}\end{equation}
for a new $C^\infty$ function $\tilde{f}$, which has compact support and is therefore bounded, too. We can bound the contribution by first integrating the variables $\xi_{N+1-\sigma(k)}$ with $k\geq\sigma(k)$, $k\neq i$, where we have a Gaussian factor $F_{k,k}$:
\begin{equation}\begin{aligned}
	\left|\eqref{eq2.20}\right|\leq\sup_{\vec{x}}&\big|\tilde{f}(\vec{x})\big|\int_{\widetilde{W}_1'}\left|F_{i,\sigma(i)}(\xi_{N+1-\sigma(i)}-\zeta_{N+1-i},{t})\right|\D\xi_{N+1-\sigma(i)}\,\\&\prod_{k< \sigma(k),k\neq i}\left|F_{k,\sigma(k)}(\xi_{N+1-\sigma(k)}-\zeta_{N+1-k},{t})\right|\D\xi_{N+1-\sigma(k)}.
\end{aligned}\end{equation}
$\widetilde{W}_1'$ consists of the yet to be integrated $\xi$-components that are contained in the set \mbox{$\widetilde{W}_1\cap\supp(\widetilde{f})$}. In particular, $\widetilde{W}_1'$ is compact, so the functions $F_{k,\sigma(k)}$, $k\neq i$, are bounded uniformly in ${t}$ by Lemma~\ref{lemFkl}. The remaining integral gives:
\begin{equation}
	\left|\eqref{eq2.20}\right|\leq\const\int^{-\e/2}_{-\infty}\left|F_{i,\sigma(i)}(x,{t})\right|\mathrm{d} x,
\end{equation}
which converges to $0$ as ${t}\to0$ by \eqref{eqFkllimn}.

The contribution of $\widetilde{W}_2$ can be bounded analogously with $j$ playing the role of $i$. The final convergence is then given by \eqref{eqFkllimp}.
\end{proof}

\begin{lem}\label{lemFkl}
For each $\e>0$ we have
\begin{align}
	\lim_{{t}\to0}\int_\e^\infty\left|F_{k,l}(x,{t})\right|\mathrm{d} x&=0, & 1&\leq l\leq k\leq N,\label{eqFkllimp}\\
	\lim_{{t}\to0}\int^{-\e}_{-\infty}\left|F_{k,l}(x,{t})\right|\mathrm{d} x&=0, & 1&\leq k,l\leq N.\label{eqFkllimn}
\end{align}
In addition, for each $1\leq k<l\leq N$ the function $F_{k,l}(x,{t})$ is bounded uniformly in ${t}$ on compact sets.
\end{lem}

\begin{proof}
Let $x<-\e$, and choose a $\mu$ which is positive. By a transformation of variable we have
\begin{equation}\begin{aligned}\label{eqFklbound}
	\left|F_{k,l}(x,{t})\right|&=\bigg|\frac{1}{2\pi\mathrm{i}}\int_{\mathrm{i}\R+\mu}\mathrm{d} w\,e^{{t} w^2/2+x w}\frac{\prod_{i=1}^{k-1}(w+\mu_{N+1-i})}{\prod_{i=1}^{l-1}(w+\mu_{N+1-i})}\bigg|\\
	&=\bigg|\frac{1}{2\pi\mathrm{i}}\int_{\mathrm{i}\R+\mu}\mathrm{d} v\,\sqrt{{t}}^{l-k-1}e^{ v^2/2+x v/\sqrt{{t}}}\frac{\prod_{i=1}^{k-1}(v+\sqrt{{t}}\mu_{N+1-i})}{\prod_{i=1}^{l-1}(v+\sqrt{{t}}\mu_{N+1-i})}\bigg|\\
	&\leq(2\pi)^{-1}\sqrt{{t}}^{l-k-1}e^{x\mu/\sqrt{{t}}}\int_{\mathrm{i}\R+\mu}|\mathrm{d} v|\,e^{ \Re(v^2/2)}g(|v|),
\end{aligned}\end{equation}
where $g(|v|)$ denotes a bound on the fraction part of the integrand, which grows at most polynomial in $|v|$. Convergence of the integral is ensured by the exponential term, so integrating and taking the limit ${t}\to0$ gives \eqref{eqFkllimn}. To see \eqref{eqFkllimp}, notice that by $l\leq k$ the integrand has no poles, so we can shift the contour to the right, such that $\mu$ is negative, and obtain the convergence analogously.

We are left to prove uniform boundedness of $F_{k,l}$ on compact sets for $k<l$. For $x\leq0$ we can use \eqref{eqFklbound} to get
\begin{equation}
	\left|F_{k,l}(x,{t})\right|\leq(2\pi)^{-1} \int_{\mathrm{i}\R+\mu}|\mathrm{d} v|\,e^{ \Re(v^2/2)}g(|v|)
\end{equation}
for ${t}\leq1$. In the case $x>0$ we shift the contour to negative $\mu$, thus obtaining contributions from residua as well as from the remaining integral. The latter can be bounded as before, while the residua are well-behaved functions, which converge uniformly on compact sets.
\end{proof}

\begin{proof}[Proof of Proposition~\ref{propStepKernel}]
Applying Proposition~\ref{PropWarren} for $\vec{\mu}=0$ and $\vec{\zeta}=0$ gives
\begin{equation}\label{eqPstep}
  \Pb\left(\vec{x}({t})\in \D\vec{\xi}|\vec{x}(0)=0\right)=\det_{1\leq k,l\leq N}[F_{k-l}(\xi_{N+1-l},{t})],
\end{equation}
where
\begin{equation}
	F_{k}(\xi,{t})=\frac{1}{2\pi\mathrm{i}}\int_{\mathrm{i}\R+1}\mathrm{d} w\,e^{{t} w^2/2+\xi w}w^k.
\end{equation}

Using repeatedly the identity
\begin{equation}
	F_{k}(\xi,{t})=\int^\xi_{-\infty} \mathrm{d} x\,F_{k+1}(x,{t}),
\end{equation}
relabeling $\xi_1^k:=\xi_k$, and introducing new variables $\xi_l^k$ for $2\leq l\leq k\leq N$,
we can write
\begin{equation}\label{eq5.1.36}
	\det_{1\leq k,l\leq N}\big[F_{k-l}(\xi_1^{N+1-l},{t})\big]=\int_{\mathcal D'} \det_{1\leq k,l\leq N}\big[F_{k-1}(\xi_l^{N},{t})\big]\prod_{2\leq l\leq k\leq N} \D\xi_l^k,
\end{equation}
where $\mathcal D' = \{\xi_l^k\in\R,2\leq l\leq k\leq N|\xi_l^k\leq \xi_{l-1}^{k-1}\}$.
Using the antisymmetry of the determinant we can apply Lemma~\ref{AppLemma3.3} to change the area of integration into ${\mathcal D}=\{\xi_l^k,2 \leq l \leq k \leq N | \xi_l^k<\xi_l^{k+1},\xi_l^k\leq \xi_{l-1}^{k-1}\}$.

The next step is to encode the constraint of the integration over $\mathcal D$ into
a formula and then consider the measure over $\{\xi_l^k , 1 \leq l \leq k \leq N \}$, which
turns out to have determinantal correlations functions. At this point the
allowed configuration are such that $\xi_l^k \leq \xi_{l+1}^k$. For a while, we still consider
ordered configurations at each level, i.e., with $\xi_1^k \leq \xi_2^k \leq \dots \leq \xi_k^k$ for
$1 \leq k \leq N$. Let us set
\begin{equation}
  \widetilde{ \mathcal D}=\{\xi_l^k,1 \leq l \leq k \leq N | \xi_l^k<\xi_l^{k+1},\xi_l^k\leq \xi_{l-1}^{k-1}\}
\end{equation}
Defining $\phi_n(x,y)=\Id_{x<y}$ and using the convention that $\Id_{\xi_n^{n-1}<y}=1$ for any $n$, it is easy to verify that
\begin{equation}
\prod_{n=2}^{N}\det_{1\leq k,l\leq n}\big[\phi_n(\xi_k^{n-1},\xi_l^n)\big]= \left\{
\begin{array}{ll}
1,&\textrm{if }\{\xi_l^k,1 \leq l \leq k \leq N\}\in \widetilde{\cal D},\\
0,&\textrm{otherwise}.
\end{array}
\right.
\end{equation}
Thus, the measure \eqref{eqPstep} is a marginal of
\begin{equation}\label{SExtM}\begin{aligned}
	\const \cdot\prod_{n=1}^{N}\det_{1\leq k,l\leq n}\big[\phi_n(\xi_k^{n-1},\xi_l^n)\big]\det_{1\leq k,l\leq N}\big[ F_{k-1}(\xi_l^{N},{t})\big].
\end{aligned}\end{equation}
Notice that the measure \eqref{SExtM} is symmetric in the
$\xi_l^k$'s since by permuting two of them (at the same level $k$) one gets twice a factor $−1$. Thus, we relax the constraint of ordered configurations at each level. The only effect is a modification of the normalization constant.

The measure \eqref{SExtM} has the appropriate form for applying Lemma~\ref{lemDetMeasure}. The composition of the $\phi$ functions can be evaluated explicitly as
\begin{equation}
\phi_{m,n}(x,y)=(\phi_{m+1}*\dots*\phi_{n})(x,y)=\frac{(y-x)^{n-m-1}}{(n-m-1)!}\Id_{x< y},
\end{equation}
for $n>m\geq0$. Define
\begin{equation}
	\Psi^n_{n-k}(\xi):=\frac{(-1)^{n-k}}{2\pi\mathrm{i}}\int_{\mathrm{i}\R-\e}\mathrm{d} w\,\,e^{{t} w^2/2+\xi w}w^{n-k},
\end{equation}
for some $\e>0$. In the case $n\geq k$ the integrand has no poles, which implies $\Psi^n_{n-k}=(-1)^{n-k}F_{n-k}$. The straightforward recursion
\begin{equation}
	(\phi_n*\Psi^n_{n-k})(\xi)=\Psi^{n-1}_{n-1-k}(\xi)
\end{equation}
eventually leads to condition \eqref{Sasdef_psi} being satisfied.

The space $V_n$ is generated by
\begin{equation}
	\{\phi_{0,n}(\xi_1^0,x),\dots,\phi_{n-2,n}(\xi_{n-1}^{n-2},x),\phi_{n-1,n}(\xi_n^{n-1},x)\},
\end{equation}
so a basis for $V_n$ is given by
\begin{equation}
	\{x^{n-1},x^{n-2},\dots,x,1\}.
\end{equation}
Choose functions $\Phi^n_{n-k}$ as follows
\begin{equation}
  \Phi^n_{n-k}(\xi)=\frac{(-1)^{n-k}}{2\pi\mathrm{i}}\oint_{\Gamma_0}\mathrm{d} z\,e^{-{t} z^2/2-\xi z}z^{-n+k-1},
\end{equation}
which are polynomials of order $n-k$ by elementary residue calculating rules, so these functions indeed generate $V_n$. To show \eqref{Sasortho}, we decompose the scalar product as follows:
\begin{equation}\label{eqS2.41}
\int_{\R_-} \mathrm{d} \xi\, \Psi^{n}_{n-k}(\xi) \Phi^n_{n-\ell}(\xi) + \int_{\R_+} \mathrm{d} \xi\, \Psi^{n}_{n-k}(\xi) \Phi^n_{n-\ell}(\xi).
\end{equation}
Since $n-k\geq0$ we are free to choose the sign of $\e$ as necessary. For the first term, we choose $\e<0$ and the path $\Gamma_0$ close enough to zero, such that always \mbox{$\Re(w-z)>0$}. Then, we can take the integral over $\xi$ inside and obtain
\begin{equation}
\int_{\R_-} \mathrm{d} \xi\, \Psi^{n}_{n-k}(\xi) \Phi^n_{n-\ell}(\xi) =\frac{(-1)^{k-l}}{(2\pi\mathrm{i})^2}\int_{\mathrm{i}\R-\e}\mathrm{d} w \oint_{\Gamma_0}\mathrm{d} z\, \frac{e^{{t} w^2/2} w^{n-k}}{e^{{t} z^2/2}z^{n-\ell+1}(w-z)}.
\end{equation}
For the second term, we choose $\e>0$ to obtain \mbox{$\Re(w-z)<0$}. Then again, we can take the integral over $\xi$ inside and arrive at the same expression up to a minus sign. The net result of \eqref{eqS2.41} is a residue at $w=z$, which is given by
\begin{equation}
\frac{(-1)^{k-l}}{2\pi\mathrm{i}}\oint_{\Gamma_0}\mathrm{d} z\, z^{\ell-k-1}=\delta_{k,\ell}.
\end{equation}

Furthermore, both $\phi_n(\xi_{n}^{n-1},x)$ and $\Phi_0^{n}(\xi)$ are constants, so the kernel has a simple form (compare with \eqref{SasK})
\begin{equation}
\mathcal{K}_\textrm{packed}(n_1,\xi_1;n_2,\xi_2)=-\phi_{n_1,n_2}(\xi_1,\xi_2)\Id_{n_2>n_1} + \sum_{k=1}^{n_2} \Psi_{n_1-k}^{n_1}(\xi_1) \Phi_{n_2-k}^{n_2}(\xi_2).
\end{equation}

Note that we are free to extend the summation over $k$ up to infinity, since the integral expression for $\Phi_{n-k}^{n}(\xi)$ vanishes for $k>n$ anyway. Taking the sum inside the integrals we can write
\begin{equation}\label{eqStepKernSum}
	\sum_{k\geq1} \Psi_{n_1-k}^{n_1}(\xi_1) \Phi_{n_2-k}^{n_2}(\xi_2)=\frac{1}{(2\pi\mathrm{i})^2}\int_{\mathrm{i}\R-\e}\hspace{-0.4cm}\mathrm{d} w \oint_{\Gamma_0}\mathrm{d} z\,	\frac{e^{{t} w^2/2+\xi_1 w}(-w)^{n_1}}{e^{{t} z^2/2+\xi_2 z}(-z)^{n_2}}\sum_{k\geq1}\frac{z^{k-1}}{w^k}.
\end{equation}

By choosing contours such that $|z|<|w|$, we can use the formula for a geometric series, resulting in
\begin{equation}
	\eqref{eqStepKernSum}=\frac{1}{(2\pi\mathrm{i})^2}\int_{\mathrm{i}\R-\e}\hspace{-0.4cm}\mathrm{d} w \oint_{\Gamma_0}\mathrm{d} z\,	\frac{e^{{t} w^2/2+\xi_1 w}(-w)^{n_1}}{e^{{t} z^2/2+\xi_2 z}(-z)^{n_2}(w-z)}=\mathcal{K}_0(n_1,\xi_1;n_2,\xi_2).
\end{equation}
\end{proof}

The lemma needed to alter the area of integration from $\cal D'$ to $\cal D$ is a continuous version of Lemma 3.3,~\cite{BFPS06}. As the proof is identical to the discrete case, we do not state it here.
\begin{lem}\label{AppLemma3.3} Let $f$ be an antisymmetric function of $\{x_1^{N},\ldots,x_N^{N}\}$. Then, whenever $f$ has enough decay to make the integrals finite,
\begin{equation}
\int_{\cal D} f(x_1^{N},\ldots,x_N^{N}) \prod_{2\leq l\leq k \leq N} \D x_l^k=\int_{\cal D'} f(x_1^{N},\ldots,x_N^{N}) \prod_{2\leq l\leq k \leq N} \D x_l^k
\end{equation}
where
\begin{eqnarray}
{\cal D}&=&\{x_l^k,2 \leq l \leq k \leq N | x_l^k<x_l^{k+1},x_l^k\leq x_{l-1}^{k-1}\},\nonumber \\
{\cal D'}&=&\{x_l^k,2 \leq l \leq k \leq N | x_l^k\leq x_{l-1}^{k-1}\},
\end{eqnarray}
and the positions $x_1^1<x_1^2<\ldots<x_1^N$ being fixed.
\end{lem}

\subsection{Asymptotic analysis}\label{secStepAsy}

According to \eqref{eqStepScaledProcess} we use the scaled variables
\begin{equation}\label{eqStepScaling}\begin{aligned}
	n_i&={t}+2{t}^{2/3}r_i\\
	\xi_i&=2{t}+2{t}^{2/3}r_i+{t}^{1/3}s_i.
\end{aligned}\end{equation}
Correspondingly, consider the rescaled (and conjugated) kernel
\begin{equation}
	\mathcal{K}^\textrm{resc}_\textrm{packed}(r_1,s_1;r_2,s_2)={t}^{1/3}e^{\xi_1-\xi_2}\mathcal{K}_\textrm{packed}(n_1,\xi_1;n_2,\xi_2),
\end{equation}
which naturally decomposes into
\begin{equation}
	\mathcal{K}^\textrm{resc}_\textrm{packed}(r_1,s_1;r_2,s_2)=-\phi_{r_1,r_2}^\textrm{resc}(s_1,s_2)\Id_{r_1<r_2}+\mathcal{K}_0^\textrm{resc}(r_1,s_1;r_2,s_2).
\end{equation}

In order to establish the asymptotics of the joint distributions, one needs both a pointwise limit of the kernel, as well as uniform bounds to ensure convergence of the Fredholm determinant expansion. The first time this approach was used is in~\cite{GTW00}. These results are contained in the following propositions.
\begin{prop}\label{propStepPointw}
Consider any $r_1,r_2$ in a bounded set and fixed $L$. Then, uniformly for $(s_1,s_2)\in[-L,L]^2$, the kernel converges as
\begin{equation}
	\lim_{{t}\to\infty}\mathcal{K}_\textrm{packed}^\textrm{resc}(r_1,s_1;r_2,s_2)=K_{\mathcal{A}_2}(r_1,s_1;r_2,s_2).
\end{equation}
\end{prop}

\begin{cor}\label{corStepBound}
Consider $r_1,r_2$ fixed. For any $L$ there exists ${t}_0$ such that for ${t}>{t}_0$ the bound
\begin{equation}
	\left|\mathcal{K}_\textrm{packed}^\textrm{resc}(r_1,s_1;r_2,s_2)\right|\leq \const_L
\end{equation}
holds for all $(s_1,s_2)\in[-L,L]^2$.
\end{cor}

\begin{prop}\label{propStepK0Bound}
For fixed $r_1,r_2,L$ there exists ${t}_0>0$ such that the estimate
\begin{equation}
	\left|\mathcal{K}_0^\textrm{resc}(r_1,s_1;r_2,s_2)\right|\leq\frac{1}{2} e^{-(s_1+s_2)}
\end{equation}
holds for any ${t}>{t}_0$ and $s_1,s_2>0$.
\end{prop}

\begin{prop}[Proposition 5.4 of~\cite{FSW13}]\label{propStepPhiBound}
For fixed $r_1<r_2$ there exists ${t}_0>0$ and $C>0$ such that
\begin{equation}
	\left|\phi_{r_1,r_2}^\textrm{resc}(s_1,s_2)\right|\leq Ce^{-|s_1-s_2|}
\end{equation}
holds for any ${t}>{t}_0$ and $s_1,s_2\in\R$.
\end{prop}

Now we can prove the asymptotic theorem:
\begin{proof}[Proof of Theorem~\ref{thmAsympStep}]
 The joint distributions of the rescaled process $X_{t}(r)$ are given by the Fredholm determinant with series expansion
\begin{equation}\label{eqStepFredExp}\begin{aligned}
	&\Pb\bigg(\bigcap_{k=1}^m\big\{X_{t}(r_k)\leq s_k\big\}\bigg)\\
	&\ =\sum_{N\geq0}\frac{(-1)^N}{N!}\sum_{i_1,\dots,i_N=1}^m\int\prod_{k=1}^N\mathrm{d} x_k\,\Id_{x_k>\xi_{i_k}}\det_{1\leq k,l\leq N}\left[\mathcal{K}_\textrm{packed}(n_{i_k},x_k;n_{i_l},x_l)\right],
\end{aligned}\end{equation}
where $n_i$ and $\xi_i$ are given in (\ref{eqStepScaling}). By employing the change of variables \mbox{$\sigma_k={t}^{-1/3}(x_k-2{t}-2{t}^{2/3}r_{i_k})$} we obtain
\begin{equation}\label{eqStepFredExp2}\begin{aligned}
	\eqref{eqStepFredExp}&=\sum_{N\geq0}\frac{(-1)^N}{N!}\sum_{i_1,\dots,i_N=1}^m\int\prod_{k=1}^N\mathrm{d} \sigma_k\,\Id_{\sigma_k>s_{i_k}}\\
	&\quad\times\det_{1\leq k,l\leq N}\left[\mathcal{K}_\textrm{packed}^\textrm{resc}(r_k,\sigma_k;r_l,\sigma_l)\frac{(1+\sigma_l^2)^{m+1-i_l}}{(1+\sigma_k^2)^{m+1-i_k}}\right],
\end{aligned}\end{equation}
where the fraction inside the determinant is a new conjugation, which does not change the value of the determinant. This additional conjugation is necessary because without it the first part of the kernel does not decay as $\sigma_k\to\infty$, if there is some $l$, with $i_k<i_l$ and $\sigma_l$ close to $\sigma_k$.
Using Corollary~\ref{corStepBound} and Propositions~\ref{propStepK0Bound},~\ref{propStepPhiBound}, we can bound the $(k,l)$-coefficient inside the determinant by
\begin{equation}\label{eqStepCoeff}
	\const_1\left(e^{-|\sigma_k-\sigma_l|}\Id_{i_k<i_l}+e^{-(\sigma_k+\sigma_l)}\right)\frac{(1+\sigma_l^2)^{m+1-i_l}}{(1+\sigma_k^2)^{m+1-i_k}},
\end{equation}
assuming the $r_k$ are ordered. The bounds
\begin{equation}\begin{aligned}
	\frac{(1+x^2)^i}{(1+y^2)^j}e^{-|x-y|}&\leq \const_2\frac{1}{1+y^2},& \text{for } i&<j,\\
	\frac{(1+x^2)^i}{(1+y^2)^j}e^{-(x+y)}&\leq \const_3\frac{1}{1+y^2},& \text{for } j&\geq1,
\end{aligned}\end{equation}
which hold for $x$, $y$ bounded from below, lead to
\begin{equation}
	\eqref{eqStepCoeff}\leq\const_4\frac{1}{1+\sigma_k^2}.
\end{equation}
Using the Hadamard bound on the determinant, the integrand of \eqref{eqStepFredExp2} is therefore bounded by
\begin{equation}
	\const_4^NN^{N/2}\prod_{k=1}^N\Id_{\sigma_k>s_{i_k}}\frac{\D\sigma_k}{1+\sigma_k^2},
\end{equation}
which is integrable. Furthermore,
\begin{equation}
	|\eqref{eqStepFredExp}|\leq\sum_{N\geq0}\frac{\const_5^NN^{N/2}}{N!},
\end{equation}
which is summable, since the factorial grows like $(N/e)^N$, i.e., much faster than the numerator. Dominated convergence thus allows to interchange the limit ${t}\to\infty$ with the integral and the infinite sum. The pointwise convergence comes from Proposition~\ref{propStepPointw}, thus
\begin{equation}\begin{aligned}
	\lim_{{t}\to\infty} \Pb\bigg(\bigcap_{k=1}^m\big\{X_{t}(r_k)\leq s_k\big\}\bigg)&=\det\left(\Id-\chi_s K_{\mathcal{A}_2}\chi_s\right)_{L^2(\{r_1,\dots,r_m\}\times\R)}\\
	&=\Pb \biggl(\bigcap_{k=1}^m \bigl\{
\mathcal{A}_2(r_k)-r_k^2\leq s_k  \bigr\} \biggr).
\end{aligned}\end{equation}\end{proof}

By using just the exponential bounds on the kernel instead of the pointwise convergence, we can now prove the concentration inequality required for Proposition~\ref{propMconv} in a similar way:
\begin{proof}[Proof of Proposition~\ref{concIn}]
First notice that by $Y_{k,m}(t)\stackrel{d}{=}Y_{1,m-k+1}(t)$ we can restrict ourselves to $k=1$ without loss of generality, and thus have to show
\begin{equation}
 	\Pb\bigg(\frac{x_m(T)}{\sqrt{mT}}\geq 2+\delta\bigg)\leq \const\cdot e^{-m^{2/3}\delta},
\end{equation}
with $\vec{x}(t)$ being the system of reflected Brownian motions with packed initial conditions.

Applying Proposition~\ref{propStepKernel}, we get by the Fredholm series expansion
\begin{equation}\begin{aligned}
	\Pb&\bigg(x_{t}(t)\leq \xi\bigg)=\sum_{N\geq0}\frac{(-1)^N}{N!}\int_{(\xi,\infty)^N}\mathrm{d} \vec{x}\,\det_{1\leq k,l\leq N}\left[\mathcal{K}_\textrm{packed}(t,x_k;t,x_l)\right].
\end{aligned}\end{equation}
We recognize that the $N=0$ term is exactly $1$, so the probability of the complementary event is simply the negative of the series started at $N=1$.
Setting $\xi=2t+t^{1/3}s$ and using the change of variables \mbox{$\sigma_k={t}^{-1/3}(x_k-2{t})$}, we recognize the scaling \eqref{eqStepScaling} and obtain
\begin{equation}\begin{aligned}
	\Pb&\bigg(\frac{x_{t}(t)}{t}\geq 2+\frac{s}{t^{2/3}}\bigg)\leq\sum_{N\geq1}\frac{1}{N!}\int_{(s,\infty)^N}\mathrm{d} \vec{\sigma}\,\left|\det_{1\leq k,l\leq N}\left[\mathcal{K}_\textrm{packed}^\textrm{resc}(0,\sigma_k;0,\sigma_l)\right]\right|.
\end{aligned}\end{equation}
By Proposition~\ref{propStepK0Bound} the $(k,l)$-coefficient inside the determinant is bounded by a constant times $e^{-(\sigma_k+\sigma_l)}$. Using the Hadamard bound on the determinant, the integrand is therefore bounded by
\begin{equation}
	\const_1^NN^{N/2}\prod_{k=1}^N\Id_{\sigma_k>s}e^{-\sigma_k}\D\sigma_k,
\end{equation}
leading to
\begin{equation}\begin{aligned}
	\Pb&\bigg(\frac{x_{t}(t)}{t}\geq 2+t^{-2/3}s\bigg)\leq\sum_{N\geq1}\frac{\const_1^NN^{N/2}}{N!}e^{-Ns}.
\end{aligned}\end{equation}
For positive $s$ we have $e^{-Ns}\leq e^{-s}$, and the remaining sum over $N$ is finite, since the factorial grows like $(N/e)^N$, i.e., much faster than the numerator. We rename $m:=t$ and notice that for any positive $T$, by Brownian scaling:
\begin{equation}
 \frac{x_{m}(m)}{m}\stackrel{d}{=}\frac{x_{m}(T)}{\sqrt{mT}},
\end{equation}
which implies
\begin{equation}
	\Pb\bigg(\frac{x_{m}(T)}{\sqrt{mT}}\geq 2+m^{-2/3}s\bigg)\leq\const\cdot e^{-s}.
\end{equation}
Inserting $s=m^{2/3}\delta$ finishes the proof.
\end{proof}

Before showing Propositions~\ref{propStepPointw} and~\ref{propStepK0Bound}, we introduce some auxiliary functions and establish asymptotic results for them.
\begin{defin}
Using the scaling
\begin{equation}\begin{aligned}
 n(t,r)&=t+2t^{2/3}r\\
 \xi(t,r,s)&=2t+2t^{2/3}r+t^{1/3}s,
\end{aligned}\end{equation}
define the functions
\begin{equation}\begin{aligned}
	\alpha_{t}(r,s)&:=\frac{{t}^{1/3}}{2\pi\mathrm{i}}\int_{\mathrm{i}\R}\mathrm{d} w\,e^{{t}(w^2-1)/2+\xi(w+1)}(-w)^{n}\\
	&=\frac{{t}^{1/3}}{2\pi\mathrm{i}}\int_{\mathrm{i}\R}\mathrm{d} w\,e^{{t}(w^2-1)/2+(2{t}+2{t}^{2/3}r+{t}^{1/3}s)(w+1)}(-w)^{{t}+2{t}^{2/3}r},\\
	\beta_{t}(r,s)&:=\frac{{t}^{1/3}}{2\pi\mathrm{i}}\oint_{\Gamma_0}\mathrm{d} z\,e^{-{t}(z^2-1)/2-\xi(z+1)}(-z)^{-n}\\
	&=\frac{{t}^{1/3}}{2\pi\mathrm{i}}\oint_{\Gamma_0}\mathrm{d} z\,e^{-{t}(z^2-1)/2-(2{t}+2{t}^{2/3}r+{t}^{1/3}s)(z+1)}(-z)^{-{t}-2{t}^{2/3}r}.	
\end{aligned}\end{equation}
\end{defin}

These functions will be very useful in the rest of this work, since often times kernels can be written as integrals of them. We can then derive pointwise convergence and uniform bounds of a kernel by analogous results for the auxiliary functions, which are the subject of the two subsequent lemmas. The proofs of Propositions~\ref{propStepPointw} and~\ref{propStepK0Bound} as well as their counterparts in the Poisson case, Propositions~\ref{propPointw} and~\ref{propK0Bound}, are considerably shortened in this way, avoiding repeated steepest descent analysis.

\begin{lem}\label{lemAlphaLimit}
The limits
\begin{equation}\begin{aligned}
	\alpha(r,s)&:=\lim_{{t}\to\infty}\alpha_{t}(r,s)=\Ai(r^2+s)e^{-\frac{2}{3}r^3-rs}\\
	\beta(r,s)&:=\lim_{{t}\to\infty}\beta_{t}(r,s)=-\Ai(r^2+s)e^{\frac{2}{3}r^3+rs}
\end{aligned}\end{equation}
hold uniformly for $s$ and $r$ in a compact set.
\end{lem}

\begin{proof}
 We start by analyzing $\beta_{t}$. Defining functions as
\begin{equation}\begin{aligned}
	f_3(z)&=-(z^2-1)/2-2(z+1)-\ln(-z)\\
	f_2(z)&=-2r(z+1+\ln(-z))\\
	f_1(z)&=-s(z+1),
\end{aligned}\end{equation}
we can write $G(z)={t} f_3(z)+{t}^{2/3}f_2(z)+{t}^{1/3}f_1(z)$, leading to
\begin{equation}
	\beta_{t}(r,s)=\frac{{t}^{1/3}}{2\pi\mathrm{i}}\oint_{\Gamma_0}\mathrm{d} z\, e^{G(z)}.
\end{equation}
These types of limits of contour integrals appear frequently when studying models in the KPZ universality class. As usual, they will be computed via \emph{steep descent analysis}. The idea is that the limit is dominated by the leading order term $tf_3(z)$ at a point where $\Re(z)$ is maximal. One therefore chooses a contour that passes through a \emph{critical point} $z_0$, satisfying $f'(z_0)=0$, in such a way that $\Re(z)$ is strictly decreasing when moving away from $z_0$ along the contour. The vanishing derivative also ensures that the imaginary part is stationary at the maximum of the real part, such that no rapid oscillations occur, which might cause cancellations.
\begin{figure}
 \centering
 \psfrag{omega}[lb]{}
 \psfrag{-1}[lb]{\hspace{-0.3em}$-1$}
 \psfrag{Gamma}[lb]{$\Gamma$}
 \psfrag{theta}[lb]{$\theta$}
 \psfrag{R}[lb]{$R$}
 \includegraphics[height=5cm]{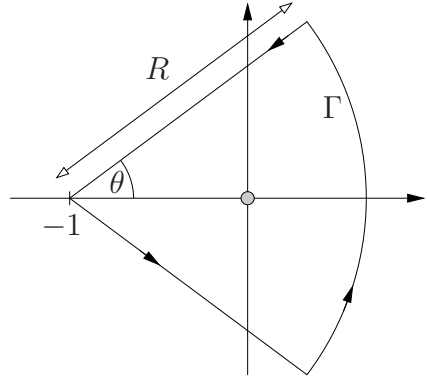}
 \caption{The contour $\Gamma=\gamma_1(R)\cup\overline{\gamma_1(R)}\cup\gamma_2(R)$ used in the steep descent analysis.}
 \label{figpointwContour}
\end{figure}

Let $\theta\in(\pi/6,\pi/4)$. We choose \mbox{$\Gamma=\gamma_1\cup\overline{\gamma_1(R)}\cup\gamma_2(R)$} as our steep descent contour, where
\begin{equation}\begin{aligned}
	\gamma_1(R)&=\{-1+ue^{\mathrm{i}\theta}, u\in[0,R]\},\\
	\gamma_2(R)&=\{-1+Re^{\mathrm{i} u}, u\in[-\theta,\theta]\},
\end{aligned}\end{equation}
with the direction of integration as in Figure~\ref{figpointwContour}.

The integrand is dominated by the $\exp(-z^2)$ term for large $|z|$ as \mbox{$\theta<\pi/4$}. Thus the contribution coming from $\gamma_2(R)$ converges to $0$ as \mbox{$R\to\infty$}. With $\gamma_1=\lim_{R\to\infty}\gamma_1(R)$ the remaining contour of integration is now \mbox{$\gamma_1\cup\overline{\gamma_1}$}. Let us show that the real part is indeed decreasing along $\gamma_1$:

\begin{equation}\begin{aligned}
\frac{\mathrm{d} \Re f_3(-1+ue^{\mathrm{i}\theta})}{\mathrm{d} u}&=\frac{\D}{\mathrm{d} u}\Re\left(-\frac{u^2}{2}e^{2\mathrm{i}\theta}-ue^{\mathrm{i}\theta}-\log(1-ue^{\mathrm{i}\theta})\right)\\
&=\Re\frac{u^2e^{3\mathrm{i}\theta}}{1-ue^{\mathrm{i}\theta}}=\frac{u^2}{||1-ue^{\mathrm{i}\theta}||^2}\Re\left(e^{3\mathrm{i}\theta}-ue^{2\mathrm{i}\theta}\right)<0.
\end{aligned}\end{equation}
Replacing $\theta$ by $-\theta$ gives the result for $\overline{\gamma_1}$. It is also evident that $f_3$ decreases quadratically in $u$, while $f_2$ and $f_1$ increase at most linearly in $u$. Convergence of the integral is therefore clear for arbitrary finite $t$.

We want to restrict the contour of integration to a small neighbourhood of our critical point $\Gamma_\delta=\{z\in\Gamma|\, |z+1|\leq\delta\}$. Since $f_3$ is a steep descent contour, the contribution coming from the remaining part will be bounded by a constant times $e^{-\mu t}$ where $\mu$ is a constant of order $\delta^3$,
\begin{equation}
\beta_t(r,s) = \Or\left(t^{1/3}e^{-\mu t}\right)+\frac{t^{1/3}}{2\pi\mathrm{i}}\int_{\Gamma_\delta}\mathrm{d} z\,e^{G(z)}.
\end{equation}

For the computation of the integral along $\Gamma_\delta$ we can now use Taylor expansion:
\begin{equation}\label{eqTaylf}\begin{aligned}
	{t} f_3(-1+\omega)&={t}\left(\omega^3/3+\Or(\omega^4)\right)\\
	{t}^{2/3}f_2(-1+\omega)&={t}^{2/3}r\left(\omega^2+\Or(\omega^3)\right)\\
	{t}^{1/3}f_1(-1+\omega)&=-\omega s{t}^{1/3}.
\end{aligned}\end{equation}
All error terms are to be understood uniformly in $s,{t},r$. Note that our critical point $z_0=-1$ is actually a \emph{doubly critical} point, i.e.\ the second derivative $f''_3(z_0)$ vanishes, too. This is a core feature of formulas appearing in the KPZ universality class. In a usual Gaussian scaling setting, the leading term of $tf_3$ would be $t\omega^2$, which required rescaling the integration variable $\omega$ by $t^{-1/2}$. This in turn led to the fluctuation scaling $st^{1/2}$ in order to obtain a non-degenerate limit. The limiting function had a second order polynomial in the exponent, resulting in the Fourier transform of a Gaussian function, i.e.\ again a Gaussian function. So a non-vanishing second derivative is connected to both the scaling exponent $1/2$ and the Gaussian limiting distribution. The actual term $t\omega^3$ appearing here is the reason for the fluctuation exponent $1/3$ as well as a third order polynomial in the exponent, giving an Airy function.

Let $\tilde{f}_i(-1+\omega)$ be the expression $f_i(-1+\omega)$ but without the error term, and define also $\tilde{G}(-1+\omega)$ correspondingly. We use the inequality \mbox{$|e^x-1|\leq|x|e^{|x|}$} to estimate the error we make by integrating over $\exp(\tilde{G})$ instead of $\exp(G)$:
\begin{equation}\label{eqTaylErr}\begin{aligned}
 \Big|\frac{t^{1/3}}{2\pi\mathrm{i}}&\int_{\Gamma_\delta}\mathrm{d} z\,\left(e^{G(z)}-e^{\tilde{G}(z)}\right)\Big|\\
 &\leq\int_{\Gamma_\delta+1}\mathrm{d} \omega\,\left|e^{\tilde{G}(-1+\omega)}\right|e^{\Or(\omega^4t+\omega^3t^{2/3}+\omega^2t^{1/3})}\Or\left(\omega^4t+\omega^3t^{2/3}+\omega^2t^{1/3}\right)\\
 &\leq\int_{\Gamma_\delta+1}\mathrm{d} \omega\,\left|e^{t\tilde{f}_3(-1+\omega)(1+\chi_3)+t^{2/3}\tilde{f}_2(-1+\omega)(1+\chi_2)+t^{1/3}\tilde{f}_1(-1+\omega)(1+\chi_1)}\right|\\&\quad\times\Or\left(\omega^4t+\omega^3t^{2/3}+\omega^2t^{1/3}\right),
\end{aligned}\end{equation}
where $\chi_1$, $\chi_2$ and $\chi_3$ are constants, which can be made as small as desired for $\delta$ small enough. The leading term in the exponential is
\begin{equation}
 t\widetilde{f}_3(-1+\omega)(1+\chi_3)=\frac{1}{3}\omega^3(1+\chi_3)t,
\end{equation}
which has negative real part and therefore ensures the integral to stay bounded for $t\to\infty$. By the change of variables $\omega=t^{-1/3}Z$ the prefactor $t^{1/3}$ cancels and the remaining $\Or$-terms imply that the overall error is $\Or(t^{-1/3})$.

\begin{equation}\label{eq3.2.26}\begin{aligned}
 \frac{t^{1/3}}{2\pi\mathrm{i}}\int_{\Gamma_\delta}\mathrm{d} z\,e^{\tilde{G}(z)}&=\frac{t^{1/3}}{2\pi\mathrm{i}}\int_{\Gamma_\delta+1}\mathrm{d} \omega\,e^{t\omega^3/3+t^{2/3}r\omega^2-t^{1/3}s\omega}\\&=\frac{1}{2\pi\mathrm{i}}\int_{e^{\theta\mathrm{i}}\delta t^{1/3}}^{-e^{\theta\mathrm{i}}\delta t^{1/3}}\mathrm{d} Z\,e^{Z^3/3+rZ^2-sZ}.
\end{aligned}\end{equation}
Letting $t\to\infty$ now just extends the integration contour up to infinity. Noticing that we are free to choose an arbitrary angle $\pi/6<\theta<\pi/2$, this is indeed the integral expression for $\beta(r,s)$.

One can carry out an analogous analysis for $\alpha_t$. Notice that with the same definition of the function $G$ we now have:
\begin{equation}
 \alpha_t(r,s)=\frac{{t}^{1/3}}{2\pi\mathrm{i}}\oint_{\mathrm{i}\R}\mathrm{d} z\, e^{-G(z)}.
\end{equation}
We choose the contour $\Gamma'=\{-1+|u|e^{\sgn(u)2\pi\mathrm{i}/3}, u\in\R\}$ and show that it is a steep descent curve:
\begin{equation}\begin{aligned}
\frac{\mathrm{d} \Re \left(-f_3(-1+ue^{2\pi\mathrm{i}/3})\right)}{\mathrm{d} u}=\frac{u^2}{||1-ue^{2\pi\mathrm{i}/3}||^2}\Re\left(-1+ue^{4\pi\mathrm{i}/3}\right)<0.
\end{aligned}\end{equation}
Repeating the other steps of the steep descent analysis in the obvious way one arrives at:
\begin{equation}
 \lim_{t\to\infty}\alpha_t(r,s)=\frac{1}{2\pi\mathrm{i}}\int_{e^{-2\pi\mathrm{i}/3}\infty}^{e^{2\pi\mathrm{i}/3}\infty}\mathrm{d} W\,e^{-W^3/3-rW^2+sW}.
\end{equation}

 \end{proof}

\begin{lem}\label{lemAlphaBound}
For fixed $r$ and $L$, there exist ${t}_0$, $c_L$ such that for all ${t}>{t}_0$ and $s>-L$ the following bounds hold
\begin{equation}\begin{aligned}
	|\alpha_{t}(r,s)|&\leq c_L e^{-s}\\
	|\beta_{t}(r,s)|&\leq c_L e^{-s}
\end{aligned}\end{equation}
\end{lem}

\begin{proof}
In the case $s<L$ the result follows from the previous lemma. Let us assume $s\geq L$ from now on. Notice that we can require $L$ to be as large as necessary, since the claim of the lemma is stronger for $L$ large.

Start by analyzing $\beta_{t}$. We have again
\begin{equation}
	\beta_{t}(r,s)=\frac{{t}^{1/3}}{2\pi\mathrm{i}}\oint_{\Gamma_0}\mathrm{d} z\, e^{G(z)},
\end{equation}
with $G(z)$ as in the proof of Proposition~\ref{lemAlphaLimit}.
\begin{figure}
 \centering
 \psfrag{-1}[lb]{$-1$}
 \psfrag{omega}[lb]{$\omega$}
 \psfrag{Gamma}[lb]{$\overline{\Gamma}$}
 \psfrag{theta}[lb]{$\theta$}
 \psfrag{R}[lb]{$R$}
 \includegraphics[height=5cm]{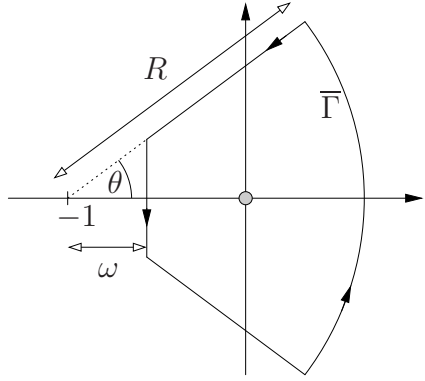}
 \caption{The contour $\overline{\Gamma}=\gamma_1(R)\cup\overline{\gamma_1(R)}\cup\gamma_2(R)\cup\gamma_3$ used for obtaining the uniform bounds.}
 \label{figboundContour}
\end{figure}

Define a new parameter $\omega$ given by
\begin{equation}\label{eqOmega}
	\omega=\min\left\{{t}^{-1/3}\sqrt{s},\e\right\},
\end{equation}
for some small, positive $\e$ chosen in the following, and let $\theta\in(\pi/6,\pi/4)$. We change the contour $\Gamma_0$ to \mbox{$\overline{\Gamma}=\gamma_1(R)\cup\overline{\gamma_1(R)}\cup\gamma_2(R)\cup\gamma_3$} as shown in Figure~\ref{figboundContour}, with
\begin{equation}\begin{aligned}
	\gamma_1(R)&=\{-1+ue^{\mathrm{i}\theta}, u\in[\omega/\cos\theta,R]\}\\
	\gamma_2(R)&=\{-1+Re^{\mathrm{i} u}, u\in[-\theta,\theta]\}\\
	\gamma_3&=\{-1+\omega(1+\mathrm{i} u\tan\theta), u\in[-1,1] \}.
\end{aligned}\end{equation}

If ${t}$ and $s$ are fixed, the integrand is dominated by the $\exp(-z^2)$ term for large $|z|$. Thus the contribution coming from $\gamma_2(R)$ converges to $0$ as $R\to\infty$. With $\gamma_1=\lim_{R\to\infty}\gamma_1(R)$ our choice for the contour of integration is now \mbox{$\gamma_1\cup\overline{\gamma_1}\cup\gamma_3$}.

We start by analyzing the contribution coming from $\gamma_3$,
\begin{equation}\label{eqAlphFac}
	\frac{{t}^{1/3}}{2\pi\mathrm{i}}\int_{\gamma_3}\mathrm{d} z\, e^{G(z)}=e^{G(z_0)}\frac{{t}^{1/3}}{2\pi}\int_{[-\omega\tan\theta,\omega\tan\theta]}\mathrm{d} u\, e^{G(z_0+\mathrm{i} u)-G(z_0)},
\end{equation}
where $z_0=-1+\omega$.

Let us consider the prefactor $e^{G(z_0)}$ at first. Since $\omega$ is small we can use Taylor expansion, as well as \eqref{eqOmega}, to obtain the bounds
\begin{equation}\begin{aligned}
	{t} f_3(z_0)&={t}\left(\omega^3/3+\Or(\omega^4)\right)\leq\frac{1}{3}\omega s{t}^{1/3}\left(1+\Or(\e)\right)\\
	{t}^{2/3}f_2(z_0)&={t}^{2/3}r\left(\omega^2+\Or(\omega^3)\right)\leq\omega \sqrt{s}{t}^{1/3}|r|\left(1+\Or(\e)\right)\\
	{t}^{1/3}f_1(z_0)&=-\omega s{t}^{1/3}.
\end{aligned}\end{equation}
All error terms are to be understood uniformly in $s,{t},r$. The $f_1$ term dominates both $f_2$, if $L$ is chosen large enough, and $f_3$, for $\e$ being small. This results in
\begin{equation}\label{eqG0bound}
	e^{G(z_0)}\leq e^{-\frac{1}{2}\omega{t}^{1/3}s}\leq e^{-s},
\end{equation}
since $\omega t^{1/3}$ can be made as large as desired by increasing $t_0$ and $L$ while keeping $\e$ fixed.

To show convergence of the integral part of \eqref{eqAlphFac} we first bound the real part of the exponent:
\begin{equation}\begin{aligned}
\Re&\left(G(z_0+\mathrm{i} u)-G(z_0)\right)\\
&=\Re\bigg[{t}\left(\frac{u^2-2z_0\mathrm{i} u}{2}-2\mathrm{i} u-\ln\frac{z_0+\mathrm{i} u}{z_0}\right)\\&\qquad\quad+{t}^{2/3}\cdot 2r\left(\mathrm{i} u+\ln\frac{z_0+\mathrm{i} u}{z_0}\right)-{t}^{1/3}s\mathrm{i} u\bigg]\\
&={t}\left(\frac{u^2}{2}-\frac{1}{2}\ln\left(1+\frac{u^2}{z_0^2}\right)\right)+{t}^{2/3}r\ln\left(1+\frac{u^2}{z_0^2}\right)\\
&\leq{t}\frac{u^2}{2}\left(1-\frac{1}{z_0^2}+\frac{u^2}{2z_0^4}\right)+{t}^{2/3}r\frac{u^2}{z_0^2}=:-\eta{t}^{2/3}u^2.
\end{aligned}\end{equation}
$\eta$ satisfies:
\begin{equation}\begin{aligned}
\eta&=\frac{{t}^{1/3}}{2}\left(\frac{1}{(1-\omega)^2}-1-\frac{u^2}{2(1-\omega)^4}\right)-\frac{r}{(1-\omega)^2}\\
&={t}^{1/3}\omega\left(1+\Or(\omega)\right)-r\left(1+\Or(\omega)\right),
\end{aligned}\end{equation}
where we used $|u|<\omega$. Given any $\e$ we can now choose both $L$ and ${t}_0$ large, such that the first term dominates. Consequently $\eta$ will be bounded from below by some positive constant $\eta_0$. The integral contribution coming from $\gamma_3$ can thus be bounded as
\begin{equation}\label{eqSDB1}\begin{aligned}
	|\eqref{eqAlphFac}|&= e^{G(z_0)}\frac{{t}^{1/3}}{2\pi}\left|\int_{[-\omega\tan\theta,\omega\tan\theta]}\mathrm{d} u\,e^{G(z_0+iu)-G(z_0)}\right|\\
	&\leq e^{-s}\frac{{t}^{1/3}}{2\pi}\int_\R \mathrm{d} u\,e^{-\eta_0{t}^{2/3}u^2}=\frac{e^{-s}}{2\pi}\int_\R \mathrm{d} u\,e^{-\eta_0u^2}
	=\frac{e^{-s}}{2\sqrt{\pi\eta_0}}.
\end{aligned}\end{equation}

Finally we need a corresponding bound on the $\gamma_1$ contribution to the integral. By symmetry this case covers also the contour $\overline{\gamma_1}$. Write
\begin{equation}\label{eqAlphFac2}
	\frac{{t}^{1/3}}{2\pi\mathrm{i}}\int_{\gamma_1}\mathrm{d} z\, e^{G(z)}=e^{G(z_1)}\frac{{t}^{1/3}e^{\mathrm{i}\theta}}{2\pi\mathrm{i}}\int_{\R_+}\mathrm{d} u\, e^{G(z_1+ue^{\mathrm{i}\theta})-G(z_1)},
\end{equation}
with $z_1=-1+\omega(1+\mathrm{i}\tan\theta)$. From the previous estimates one easily gets
\begin{equation}
	\left|e^{G(z_1)}\right|\leq e^{G(z_0)}\leq e^{-s},
\end{equation}
so the remaining task is to show boundedness of the integral part of \eqref{eqAlphFac2}.

At first notice that the real part of the $f_1$ contribution in the exponent is negative, so we can omit it, avoiding the problem of large $s$. By elementary calculus, we have for all $u\geq\omega/\cos\theta$,
\begin{equation}
	\frac{\D}{\mathrm{d} u}\Re\left(f_3(-1+ue^{\mathrm{i}\theta})\right)<0,
\end{equation}
that is, $\gamma_1$ is a steep descent curve for $f_3$. We can therefore restrict the contour to a neighbourhood of the critical point $z_1$, which we choose of size $\delta$. The error we make is exponentially small in $t$, so can be bounded by $1$ through choosing $t_0$ large enough:
\begin{equation}\label{eq3.43}\begin{aligned}
	\left|\frac{{t}^{1/3}}{2\pi}\int_{\R_+}\mathrm{d} u\, e^{G(z_1+ue^{\mathrm{i}\theta})-G(z_1)}\right|&\leq 1+\int_0^\delta\mathrm{d} u\, \left|e^{{t}\hat{f}_3(ue^{\mathrm{i}\theta})+{t}^{2/3}\hat{f}_2(ue^{\mathrm{i}\theta})}\right|
\end{aligned}\end{equation}
where $\hat{f}_i(z)=f_i(z_1+z)-f_i(z_1)$. Taylor expanding these functions leads to
\begin{equation}\begin{aligned}
	\Re({t}\hat{f}_3(ue^{\mathrm{i}\theta}))&={t} \Re(e^{3\mathrm{i}\theta})u\frac{\omega^2}{\cos^2\theta}\left(1+\Or(\delta)\right)\left(1+\Or(\e)\right)\\
&\leq-\chi_3{t}^{1/3}\omega \cdot{t}^{2/3}u\omega\\	\Re({t}^{2/3}\hat{f}_2(ue^{\mathrm{i}\theta}))&=2\Re(e^{2\mathrm{i}\theta}){t}^{2/3}ru\frac{\omega}{\cos\theta}\left(1+\Or(\delta)\right)\left(1+\Or(\e)\right)\\
&\leq\chi_2|r|\cdot{t}^{2/3}u\omega,
\end{aligned}\end{equation}
for some positive constants $\chi_2$, $\chi_3$, by choosing $\delta$ and $\e$ small enough. For large $L$ and ${t}_0$, $-\chi_3{t}^{1/3}\omega$ dominates over $\chi_2|r|$, so we can further estimate:
\begin{equation}
	\int_0^\delta\mathrm{d} u\, \left|e^{{t}\hat{f}_3(ue^{\mathrm{i}\theta})+{t}^{2/3}\hat{f}_2(ue^{\mathrm{i}\theta})+\hat{f}_0(ue^{\mathrm{i}\theta})}\right|\leq\int_0^\infty\mathrm{d} u\, e^{-\chi_3{t}\omega^2 u/2}\leq\frac{2}{\chi_3{t}\omega^2}.
\end{equation}
Combining this with \eqref{eqSDB1} gives us
\begin{equation}
 |\beta_{t}(r,s)|\leq e^{-s}\left(\frac{1}{2\sqrt{\pi\eta_0}}+2\left(1+\frac{2}{\chi_3{t}\omega^2}\right)\right) \leq c_L e^{-s}
\end{equation}

The bound on $\alpha_t$ can be obtained by the same line of arguments. In this case choose the contour $\Gamma'=\gamma'_1\cup\overline{\gamma'_1}\cup\gamma'_3$, with
\begin{equation}\begin{aligned}
	\gamma'_1&=\{-1+ue^{2\pi\mathrm{i}/3}, u\in[2\omega,R]\}\\
	\gamma'_3&=\{-1+\omega(-1+\mathrm{i} u\sqrt{3}), u\in[-1,1] \}.
\end{aligned}\end{equation}
\end{proof}

\begin{proof}[Proof of Proposition~\ref{propStepPointw}]
We start with the first part of the kernel. It has an integral representation:
\begin{equation}
 \phi_{n_1,n_2}(\xi_1,\xi_2)=\int_{\mathrm{i}\R-\delta}\mathrm{d} z\, \frac{e^{z(\xi_1-\xi_2)}}{(-z)^{n_2-n_1}}.
\end{equation}
Inserting the scaling gives
\begin{equation}\label{eq62}
	t^{1/3}e^{\xi_1-\xi_2}\phi_{n_1,n_2}(\xi_1,\xi_2) = \frac{t^{1/3}}{2\pi\mathrm{i}} \int_{\mathrm{i}\R-\delta}\mathrm{d} z\, \frac{e^{(z+1)(\xi_1-\xi_2)}}{(-z)^{n_2-n_1}}.
\end{equation}
Setting $\delta=1$ and using the change of variables $z=-1+t^{-1/3}\zeta$ as well as the shorthand $r=r_2-r_1$ and $s=s_2-s_1$, we have
\begin{equation}
		\eqref{eq62}= \frac{1}{2\pi\mathrm{i}} \int_{\mathrm{i}\R}\mathrm{d} \zeta\, \frac{e^{t^{-1/3}\zeta(\xi_1-\xi_2)}}{(1-t^{-1/3}\zeta)^{n_2-n_1}}=
		\frac{1}{2\pi\mathrm{i}} \int_{\mathrm{i}\R}\mathrm{d} \zeta\,e^{-s\zeta}f_t(\zeta,r)
\end{equation}
with
\begin{equation}
	f_t(\zeta,r)=\frac{e^{-2t^{1/3}r\zeta}}{(1-t^{-1/3}\zeta)^{2t^{2/3}r}}=e^{-2t^{1/3}r\zeta-2t^{2/3}r\log(1-t^{-1/3}\zeta)}.
\end{equation}
Since this integral is $0$ for $r\leq0$ we can assume $r>0$ from now on. The function $f_t(\zeta,r)$ satisfies the pointwise limit $\lim_{t\to\infty}f_t(\zeta,r)=e^{r\zeta^2}$, which is easy to see by Taylor expanding the logarithm in the exponent. Applying Bernoulli's inequality, we also obtain a $t$-independent integrable bound
\begin{equation}\begin{aligned}
	|f_t(\zeta,r)|&=|1-t^{-1/3}\zeta|^{-2t^{2/3}r}=\left(1+t^{-2/3}|\zeta|^2\right)^{-t^{2/3}r}
	\\&\leq(1+r|\zeta|^2)^{-1}.
\end{aligned}\end{equation}
Thus by dominated convergence
\begin{equation}
	\left|\frac{1}{2\pi\mathrm{i}} \int_{\mathrm{i}\R}\mathrm{d} \zeta\,\left(e^{-s\zeta}f_t(\zeta,r)-e^{-s\zeta+r\zeta^2}\right)\right|\leq \frac{1}{2\pi} \int_{\mathrm{i}\R}|\mathrm{d} \zeta|\,\big|f_t(\zeta,r)-e^{r\zeta^2}\big|\stackrel{t\to\infty}{\longrightarrow}0.
\end{equation}
This implies that the convergence of the integral is uniform in $s$. The limit is easily identified as
\begin{equation}\begin{aligned}
	\lim_{t\to\infty}-\phi^\textrm{resc}_{r_1,r_2}(s_1,s_2)=-\frac{1}{2\pi\mathrm{i}} \int_{\mathrm{i}\R}\mathrm{d} \zeta\,e^{-s\zeta+r\zeta^2}\Id_{r>0} =
		-\frac{1}{\sqrt{4\pi r}}e^{-s^2/4r}\Id_{r>0},
\end{aligned}\end{equation}
which is the first part of the kernel $K_{\mathcal{A}_2}$.

The remaining kernel can be rewritten as integrals over the previously defined functions $\alpha$ and $\beta$. Therefore, choose the contours in such a way that \mbox{$\Re(z-w)>0$} is ensured.
\begin{equation}\label{eqStep26}\begin{aligned}
&\mathcal{K}_0^\textrm{resc}(r_1,s_1;r_2,s_2)=  {t}^{1/3}e^{\xi_1-\xi_2}\mathcal{K}_0(n_1,\xi_1;n_2,\xi_2)\\
&=\frac{{t}^{1/3}}{(2\pi\mathrm{i})^2}\int_{\mathrm{i}\R-\e}\mathrm{d} w\oint_{\Gamma_0}\mathrm{d} z\frac{e^{{t} w^2/2+\xi_1(w+1)}}{e^{{t} z^2/2+\xi_2 (z+1)}}\frac{(-w)^{n_1}}{(-z)^{n_2}}\frac{1}{w-z}\\
	&=\frac{-{t}^{1/3}}{(2\pi\mathrm{i})^2}\int_{\mathrm{i}\R-\e}\mathrm{d} w\oint_{\Gamma_0}\mathrm{d} z\frac{e^{{t} (w^2-1)/2+\xi_1(w+1)}}{e^{{t} (z^2-1)/2+\xi_2 (z+1)}}\frac{(-w)^{n_1}}{(-z)^{n_2}}\int_0^\infty \mathrm{d} x\,{t}^{1/3}e^{-{t}^{1/3}x(z-w)}\\
	&=-\int_0^\infty \mathrm{d} x\,\alpha_{t}(r_1,s_1+x)\beta_{t}(r_2,s_2+x)
\end{aligned}\end{equation}

Using the previous lemmas we can deduce compact convergence of the kernel. Indeed (omitting the $r$-dependence for greater clarity) we can write:
\begin{equation}\begin{aligned}\label{eqStepSupK0}
	\sup_{s_1,s_2\in[-L,L]}&\left|\int_0^\infty \mathrm{d} x\,\alpha_{t}(s_1+x)\beta_{t}(s_2+x)-\int_0^\infty \mathrm{d} x\,\alpha(s_1+x)\beta(s_2+x)\right|\\
	\leq&\int_0^\infty \mathrm{d} x\,\sup_{s_1,s_2\in[-L,L]}\left|\alpha_{t}(s_1+x)\beta_{t}(s_2+x)-\alpha(s_1+x)\beta(s_2+x)\right|.
\end{aligned}\end{equation}
By Lemma~\ref{lemAlphaLimit} the integrand converges to zero for every $x>0$. Using Lemma~\ref{lemAlphaBound} we can bound it by $\const\cdot e^{-2x}$, thus ensuring that \eqref{eqStepSupK0} goes to zero, i.e., $\mathcal{K}_0^\textrm{resc}$ converges compactly. Applying the limit in \eqref{eqStep26} and inserting the expressions for $\alpha$ and $\beta$ finishes the proof.
\end{proof}
\begin{proof}[Proof of Proposition~\ref{propStepK0Bound}]
Inserting the bounds from Lemma~\ref{lemAlphaBound} into \eqref{eqStep26} results in
\begin{equation}
		\left|\mathcal{K}_0^\textrm{resc}(r_1,s_1;r_2,s_2)\right|\leq\int_0^\infty \mathrm{d} x\,e^{-(s_1+x)}e^{-(s_2+x)}=\frac{1}{2}e^{-(s_1+s_2)}.
\end{equation}
\end{proof}
\section{Periodic initial conditions}\label{secPer}
The periodic initial conditions have been analyzed in detail in~\cite{FSW13}. We refer the reader to this work for details and give only the main results as well as a sketch of the idea of the arguments. Notice that here the direction of space is reversed as compared to~\cite{FSW13}.
\subsection{Determinantal structure}
The first result is an expression for the joint distribution at fixed time $t$.
\begin{prop}\label{propFlat}
Let $\{x_n(t),n\in\Z\}$ be the system of one-sided reflected Brownian motions with initial condition $\vec{x}(0)=\vec{\zeta}^\textrm{flat}$. Then, for any finite subset $S$ of $\Z$, it holds
\begin{equation}
\Pb\left(\bigcap_{n\in S} \{x_n(t)\leq a_k\}\right)=\det(\Id-P_a K_t^\textrm{flat} P_a)_{L^2(\R\times S)},
\end{equation}
where $P_a(x,k)=\Id_{(a_k,\infty)}(x)$ and the kernel $K_t^\textrm{flat}$ is given by
\begin{equation}\label{eqKtflat}
\begin{aligned}
K_t^\textrm{flat}&(x_1,n_1;x_2,n_2)=-\frac{(x_2-x_1)^{n_2-n_1-1}}{(n_2-n_1-1)!}\Id(x_2\geq x_1)\Id(n_2>n_1)\\
&+\frac{1}{2\pi\mathrm{i}} \int_{\Gamma_-} \mathrm{d} z\frac{e^{t z^2/2} e^{z x_1}(-z)^{n_1}}{e^{t \varphi(z)^2/2} e^{\varphi(z) x_2} (-\varphi(z))^{n_2}}.
\end{aligned}
\end{equation}
Here $\Gamma_-$ is any path going from $\infty e^{-\theta \mathrm{i}}$ to $\infty e^{\theta \mathrm{i}}$ with $\theta\in [\pi/2,3\pi/4)$, crossing the real axis to the left of $-1$, and such that the function
\begin{equation}
\varphi(z)=L_0(z e^{z})
\end{equation}
is continuous and bounded. Here $L_0$ is the Lambert-W function, \textit{i.e.}, the principal solution for $w$ in $z=w e^w$, see Figure~\ref{FigContoursProp}.
\end{prop}
\begin{figure}
\begin{center}
\includegraphics[height=5cm]{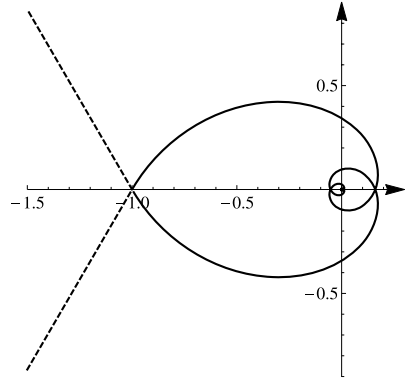}
\caption{A possible choice for the contour $\Gamma_-$ (dashed line) and its image by $\varphi$ (solid line).}
\label{FigContoursProp}
\end{center}
\end{figure}

Interesting and quite unexpected is the appearance of the Lambert function, defined as the multi-valued inverse of the function $z\mapsto ze^z$. It has a branch structure similar to the logarithm, but slightly more complicated. The Lambert function is of use in many different areas like combinatorics, exponential towers, delay-differential equations~\cite{Che02} and several problems from physics~\cite{BJ00,JK04,CJV00}. This function has been studied in detail, \textit{e.g.}  see~\cite{BPLPCS00,CJK97,CGHJ93}, with~\cite{CGHJK96} the standard reference. However, the specific behaviour needed for  our asymptotic analysis does not seem to be covered in the literature.

The proof of Proposition~\ref{propFlat} proceeds in several steps. The first step develops a determinantal for the half-periodic system $\{x^\hf_n(t),n\geq1\}$ with initial condition $\vec{x}^\shf(0)=\vec{\zeta}^\shf$ (see Section~\ref{secHF}), in a similar way as in the proof of Proposition~\ref{propStepKernel} but with a slightly more complex choice of the orthogonal polynomials. This system is subsequently scaled to the full periodic case by
\begin{equation}
 x_n(t)=\lim_{M\to\infty} \big(x^\hf_{n+M}(t)-M\big).
\end{equation}
The resulting kernel has a main part that is given by a double contour integral,
\begin{equation}
 -\frac{1}{(2\pi\mathrm{i})^2} \oint_{\Gamma_0}\mathrm{d} w \int_{\Gamma_-} \mathrm{d} z\frac{e^{t z^2/2} e^{z x_1}(-z)^{n_1}}{e^{t w^2/2} e^{w x_2} (-w)^{n_2}}\frac{(1+ w) e^{ w}}{z e^{ z}-w e^{ w}}\left(\frac{w e^{ w}}{z e^{ z}}\right)^{n_2},
\end{equation}
which can then be simplified to the form of \eqref{eqKtflat} using the Lambert-W function.
\subsection{Asymptotics}
The second main result contains a characterization of the law for the positions of the interacting Brownian motions in the large time limit. Due to the asymmetric reflections, the particles have an average velocity $1$, so that the macroscopic position of $x_{\lfloor t\rfloor}(t)$ is around $2t$. For large time $t$ the KPZ scaling theory predicts the positional fluctuations relative to the characteristic to be of order $t^{1/3}$. Nontrivial correlations between particles occur if the particle indices are of order $t^{2/3}$ apart from each other. The scaling is thus the same as the one used in the packed initial conditions case.
\begin{thm}\label{thmAsympFixedTime}
With $\{x_n(t),n\in\Z\}$ being the system of one-sided reflected Brownian motions with initial condition $\vec{x}(0)=\zeta^\textrm{flat}$, define the rescaled process
\begin{equation}\label{defXt}
	r\mapsto X_t^\textrm{flat}(r)=t^{-1/3}\left(x_{\lfloor t+2t^{2/3}r\rfloor}(t)-2t-2t^{2/3}r\right).
\end{equation}
In the sense of finite-dimensional distributions,
\begin{equation}
	\lim_{t\to\infty}X_t^\textrm{flat}(r) = 2^{1/3}\mathcal{A}_1(2^{-2/3}r).
\end{equation}
\end{thm}
The proof of this theorem relies on the usual asymptotic analysis of the kernel via steep descent. However, rather unconventional claims about the behaviour of the Lambert-W function $L_k(z)$ are necessary, which do not seem to be covered in the literature. For a special choice of the integration contour,  these claims can be derived from the differential identity
\begin{equation}
 L'_k(z)=\frac{L_k(z)}{z(1+L_k(z))}
\end{equation}
as well as some information about the branch structure (see also the proof of Proposition~\ref{propK1Pointw}).
\begin{remark}
 Due to the translational invariance, other choices of the rescaled process are possible, for which the asymptotic theorem still holds. For example, one could choose
 \begin{equation}
  	r\mapsto \bar{X}_t^\textrm{flat}(r)=t^{-1/3}\left(x_{\lfloor -t+2t^{2/3}r\rfloor}(t)-2t^{2/3}r\right),
 \end{equation}
which is the way it has been proven in~\cite{FSW13} (up to constant factors).
\end{remark}

\chapter{Stationary initial conditions}\label{secPoi}
As stated more precisely in Proposition~\ref{propBM}, stationary initial conditions are given through a Poisson point process with uniform density.
Our tools can be adjusted to also cover such random initial data. The novel difficulty comes from the determinantal expressions being seemingly  ill-defined. This forces a more elaborate approximation scheme, in which one starts from a Poisson point process, which is uniform except for a jump at $0$.

\section{Poisson initial conditions}\label{secPoi1}
In the trio of basic initial conditions for KPZ we turn to the stationary initial condition as given by a Poisson point process. The final result is again an asymptotic theorem.
\begin{thm}\label{thmAsymp0}
With $\{x_n(t),n\in\Z\}$ being the system of one-sided reflected Brownian motions with initial condition $\vec{x}(0)=\vec{\zeta}^\textrm{stat}(1,1)$, define the rescaled process
\begin{equation}\label{eqScaledProcessOriginal}
	r\mapsto X_t^\textrm{stat}(r) = t^{-1/3}\big(x_{\lfloor t+2rt^{2/3}\rfloor}(t)-2t-2rt^{2/3} \big).\,
\end{equation}
In the sense of finite-dimensional distributions,
\begin{equation}\label{eqX0limit}
	\lim_{t\to\infty}X_t^\textrm{stat}(r)\stackrel{d}{=}\mathcal{A}_\mathrm{stat}(r).
\end{equation}
\end{thm}

A crucial step towards this theorem is a determinantal formula for the fixed time distribution of the process under the initial condition with two different densities on $\R_+$ and $\R_-$:

\begin{prop}\label{propKernel}
Let $\{x_n(t),n\in\Z\}$ be the system of one-sided reflected Brownian motions with initial condition $\vec{x}(0)=\zeta^\textrm{stat}(\lambda,\rho)$ for any $\lambda>\rho>0$. For any finite subset $S$ of $\Z_{\geq0}$, it holds
\begin{equation}\label{eq33}
\Pb\bigg(\bigcap_{n\in S} \{x_n(t)\leq a_n\}\bigg)=\bigg(1+\frac{1}{\lambda-\rho}\sum_{n\in S}\frac{\D}{\mathrm{d} a_n}\bigg)\det(\Id-\chi_a \mathcal{K}_\textrm{stat} \chi_a)_{L^2(S\times \R)},
\end{equation}
where $\chi_a(n,\xi)=\Id_{\xi>a_n}$. The kernel $\mathcal{K}_\textrm{stat}$ is given by
\begin{equation}\label{eqKtStat}
\begin{aligned}	\mathcal{K}_\textrm{stat}(n_1,\xi_1;n_2,\xi_2)&=-\phi_{n_1,n_2}(\xi_1,\xi_2)\Id_{n_2>n_1}+\mathcal{K}_0(n_1,\xi_1;n_2,\xi_2)\\&\quad+(\lambda-\rho)\mathpzc{f}(n_1,\xi_1)\mathpzc{g}(n_2,\xi_2).
\end{aligned}
\end{equation}
where
\begin{equation}\begin{aligned}
	\phi_{0,n_2}(\xi_1,\xi_2)&=\rho^{-n_2}e^{\rho \xi_2},&\quad& \text{for } n_2\geq0,\\
	\phi_{n_1,n_2}(\xi_1,\xi_2)&=\frac{(\xi_2-\xi_1)^{n_2-n_1-1}}{(n_2-n_1-1)!}\Id_{\xi_1\leq \xi_2},&& \text{for } 1\leq n_1<n_2,
	\end{aligned}\end{equation}
and	
	\begin{equation}\begin{aligned}	\mathcal{K}_0(n_1,\xi_1;n_2,\xi_2)&=\frac{1}{(2\pi\mathrm{i})^2}\int_{\mathrm{i}\R-\e}\mathrm{d} w\,\oint_{\Gamma_0}\mathrm{d} z\,\frac{e^{{t} w^2/2+\xi_1w}}{e^{{t} z^2/2+\xi_2 z}}\frac{(-w)^{n_1}}{(-z)^{n_2}}\frac{1}{w-z},\\
	\mathpzc{f}(n_1,\xi_1)&=\frac{1}{2\pi\mathrm{i}}\int_{\mathrm{i}\R-\e}\mathrm{d} w\,\frac{e^{{t} w^2/2+\xi_1w}(-w)^{n_1}}{w+\lambda},\\
	\mathpzc{g}(n_2,\xi_2)&=\frac{1}{2\pi\mathrm{i}}\oint_{\Gamma_{0,-\rho}}\mathrm{d} z\,\frac{e^{-{t} z^2/2-\xi_2 z}(-z)^{-n_2}}{z+\rho},
\end{aligned}\end{equation}
for any fixed $0<\e<\lambda$.
\end{prop}

Notice that this result holds for $\lambda>\rho$ only and not for the most interesting case $\lambda=\rho$. The latter can be accessed through a careful analytic continuation of the formula \eqref{eq33}. One of the novelties of this work is to perform the analytic continuation \emph{after} the scaling limit. This allows us to discover a new process, called \emph{finite-step Airy$_\textrm{{stat}}$ process},  describing the large time limit close to stationarity.

In general, the limits $t\to\infty$ and $\lambda-\rho\downarrow 0$ do not commute. Therefore we have to consider $\lambda-\rho>0$ (to be able to apply Proposition~\ref{propKernel}), but vanishing with a tuned scaling exponent as $t\to\infty$, a critical scaling. We set $\lambda-\rho=\delta t^{-1/3}$ for $\delta>0$.  These considerations give rise to the following positive-step asymptotic theorem:
\begin{thm}\label{thmAsymp}
With $\{x^{(\rho)}_n(t),n\in\Z\}$ being the system of one-sided reflected Brownian motions with initial condition $\vec{x}^{(\rho)}(0)=\vec{\zeta}^\textrm{stat}(1,\rho)$, define the rescaled process
\begin{equation}\label{eqScaledProcess}
	r\mapsto X^{(\delta)}_t(r) = t^{-1/3}\Big(x^{(1-t^{-1/3}\delta)}_{\lfloor t+2rt^{2/3}\rfloor}(t)-2t-2rt^{2/3}\Big).\,
\end{equation}
For every $\delta>0$, the rescaled process converges to the \emph{finite-step Airy$_\textrm{{stat}}$ process}
\begin{equation}
	\lim_{t\to\infty}X^{(\delta)}_t(r)\stackrel{d}{=}\mathcal{A}^{(\delta)}_\textrm{stat}(r),
\end{equation}
in the sense of finite-dimensional distributions.
\end{thm}
Recognizing $X^{(0)}_t(r)=X^\textrm{stat}_t(r)$, in order to finally arrive at Theorem~\ref{thmAsymp0}, we have to prove that the limit $t\to\infty$ commutes with $\delta\downarrow 0$ on the left hand side, as well as convergence of the finite-step Airy$_\textrm{stat}$ process to the standard Airy$_\textrm{stat}$ process. This is the content of Section~\ref{SectAnCont}.

\section{Determinantal structure}

To obtain a representation as a signed determinantal point process we have to introduce a new measure. This measure $\Pb_+$ coincides with $\Pb$ on the sigma algebra which is generated by $\zeta_{k+1}^\textrm{stat}-\zeta_{k}^\textrm{stat}$, $k\in\Z$, and the driving Brownian motions $B_k$, $k\in\Z$. But under $\Pb_+$, $\zeta_0^\textrm{stat}$ is a random variable with an exponential distribution instead of being fixed at zero. Formally, \mbox{$\Pb_+=\Pb\otimes\Pb_0$}, with $\Pb_0$ giving rise to \mbox{$\zeta_0^\textrm{stat}\sim\exp(\lambda-\rho)$}, so that $\Pb$ is the result of conditioning $\Pb_+$ on the event $\{\zeta_0^\textrm{stat}=0\}$. This new measure satisfies a determinantal formula for the joint distribution at a fixed time.

\begin{prop}\label{PropRandomIC}
Under the modified initial condition specified by $\Pb_+$, the joint density of the positions of the asymmetrically reflected Brownian motions \mbox{$\{x_n({t}),0\leq n\leq N-1\}$} is given by
	\begin{equation}\begin{aligned}
	\Pb_+(\vec{x}({t})\in \D\vec{\xi})=(\lambda-\rho)\lambda^{N-1}e^{-{t}\rho^2/2+\rho \xi_0}
	 \det_{1\leq k,l\leq N}[\widetilde{F}_{k-l}(\xi_{N-l},{t})]\,\D\vec{\xi}
\end{aligned}\label{Prnd}\end{equation}
with
\begin{equation}
	\widetilde{F}_{k}(\xi,{t}):=\frac{1}{2\pi\mathrm{i}}\int_{\mathrm{i}\R+\e}\mathrm{d} w\,\frac{e^{{t} w^2/2+\xi w}w^k}{w+\lambda}.
\end{equation}
\end{prop}

\begin{proof}[Proof of Proposition~\ref{PropRandomIC}]
The fixed time distribution can be obtained by integrating the transition density \eqref{transDens} over the initial condition. Denote by $p_+(\vec{\xi})$ the probability density of $\vec{x}(t)$, i.e., $\Pb_+(\vec{x}({t})\in \D\vec{\xi})=p_+(\vec{\xi})\D\vec{\xi}$. Then $p_+(\vec{\xi})$ equals
\begin{equation}\label{eqPrnd1}\begin{aligned}
	&\int_{W^N\cap\{\zeta_0>0\}}\D\vec{\zeta}\, e^{\rho(\xi_0-\zeta_0)-{t}\rho^2/2}(\lambda-\rho)\lambda^{N-1}e^{\rho\zeta_0}e^{-\lambda\zeta_{N-1}}\\
        &\quad\times\det_{1\leq k,l\leq N}[F_{k,l}(\xi_{N-l}-\zeta_{N-k},{t})]\\
&=(\lambda-\rho)\lambda^{N-1}e^{-{t}\rho^2/2+\rho\xi_0}\int_{W^N\cap\{\zeta_0>0\}}\D\vec{\zeta}\, e^{\lambda\zeta_{N-1}} \\
	&\quad\times\det_{1\leq k,l\leq N}\left[\frac{1}{2\pi\mathrm{i}}\int_{\mathrm{i}\R+\mu}\mathrm{d} w_k\,e^{{t} w_k^2/2}e^{\xi_{N-l}w_k}e^{-\zeta_{N-k}w_k}w_k^{k-l}\right]\\
	&=\frac{(\lambda-\rho)\lambda^{N-1}}{e^{{t}\rho^2/2-\rho\xi_0}}\sum_{\sigma\in \mathcal{S}_N}(-1)^{|\sigma|}\prod_{k=1}^N\int_{\mathrm{i}\R+\mu}\frac{\mathrm{d} w_k}{2\pi\mathrm{i}}\,e^{{t} w_k^2/2}e^{\xi_{N-\sigma(k)}w_k}w_k^{k-\sigma(k)}\\ &\quad\times\int^{\infty}_0\mathrm{d} \zeta_0\dots\int^{\infty}_{\zeta_{N-2}}\mathrm{d} \zeta_{N-1}\,e^{-\lambda\zeta_{N-1}}e^{-\zeta_{N-1}w_1}e^{-\zeta_{N-2}w_2}\dots e^{-\zeta_0w_N}\\
&=\frac{(\lambda-\rho)\lambda^{N-1}}{e^{{t}\rho^2/2-\rho\xi_0}}\sum_{\sigma\in \mathcal{S}_N}(-1)^{|\sigma|}\prod_{k=1}^N\int_{\mathrm{i}\R+\mu}\frac{\mathrm{d} w_k}{2\pi\mathrm{i}}\,\frac{e^{{t} w_k^2/2}e^{\xi_{N-\sigma(k)}w_k}w_k^{k-\sigma(k)}}{w_1+\dots+w_k+\lambda}.
\end{aligned}\end{equation}
Since all $w_k$ are integrated over the same contour, we can replace $w_k$ by $w_{\sigma(k)}$:
\begin{equation}\label{eqPrnd2}\begin{aligned}
\eqref{eqPrnd1}&=\frac{(\lambda-\rho)\lambda^{N-1}}{e^{{t}\rho^2/2-\rho\xi_0}}\\
&\quad\times\sum_{\sigma\in \mathcal{S}_N}(-1)^{|\sigma|}\prod_{k=1}^N\int_{\mathrm{i}\R+\mu}\frac{\mathrm{d} w_k}{2\pi\mathrm{i}}\,\frac{e^{{t} w_k^2/2}e^{\xi_{N-\sigma(k)}w_{\sigma(k)}}w_{\sigma(k)}^{k-\sigma(k)}}{w_{\sigma(1)}+\dots+w_{\sigma(k)}+\lambda}\\ &=\frac{(\lambda-\rho)\lambda^{N-1}}{e^{{t}\rho^2/2-\rho\xi_0}}\prod_{k=1}^N\int_{\mathrm{i}\R+\mu}\frac{\mathrm{d} w_k}{2\pi\mathrm{i}}\,e^{{t} w_k^2/2}e^{\xi_{N-k}w_k}w_k^{-k}\\
&\quad\times\sum_{\sigma\in \mathcal{S}_N}(-1)^{|\sigma|}\prod_{k=1}^N\frac{w_{\sigma(k)}^{k}}{w_{\sigma(1)}+\dots+w_{\sigma(k)}+\lambda}.
\end{aligned}\end{equation}
We apply Lemma~\ref{lemDetIdent} below to the sum and finally obtain
\begin{equation}\begin{aligned}
p_+(\vec{x})&=\frac{(\lambda-\rho)\lambda^{N-1}}{e^{{t}\rho^2/2-\rho\xi_0}}\prod_{k=1}^N\int_{\mathrm{i}\R+\mu}\frac{\mathrm{d} w_k}{2\pi\mathrm{i}}\,e^{{t} w_k^2/2}e^{\xi_{N-k}w_k}w_k^{-k} \det_{1\leq k,l\leq N}\left[\frac{w_l^k}{w_l+\lambda}\right]\\
&=(\lambda-\rho)\lambda^{N-1}e^{-{t}\rho^2/2+\rho\xi_0}\det_{1\leq k,l\leq N}\left[\widetilde{F}_{k-l}(\xi_{N-l},{t})\right].
\end{aligned}\end{equation}
\end{proof}

\begin{lem}\label{lemDetIdent}
Given $N\in\Z_{>0}$, $\lambda>0$ and $w_1,\dots,w_N\in\C\setminus\R_-$, the following identity holds:
\begin{equation}\label{eqDetIdent}
	\sum_{\sigma\in \mathcal{S}_N}(-1)^{|\sigma|}\prod_{k=1}^N\frac{w_{\sigma(k)}^{k}}{w_{\sigma(1)}+\dots+w_{\sigma(k)}+\lambda}=\det_{1\leq k,l\leq N}\left[\frac{w_l^k}{w_l+\lambda}\right].
\end{equation}
\end{lem}

\begin{proof}
We use induction on $N$. For $N=1$ the statement is trivial. For arbitrary $N$, rearrange the left hand side of \eqref{eqDetIdent} as
\begin{equation}\begin{aligned}\label{eqDI1}
	&\sum_{\sigma\in \mathcal{S}_N}(-1)^{|\sigma|}\prod_{k=1}^N\frac{w_{\sigma(k)}^{k}}{w_{\sigma(1)}+\dots+w_{\sigma(k)}+\lambda}\\&\quad=\sum_{l=1}^N\frac{w_l^N}{w_1+\dots+w_N+\lambda}\sum_{\sigma\in S_{N},\sigma(N)=l}(-1)^{|\sigma|}\prod_{k=1}^{N-1}\frac{w_{\sigma(k)}^{k}}{w_{\sigma(1)}+\dots+w_{\sigma(k)}+\lambda}\\
	&\quad=\sum_{l=1}^N\frac{w_l^N}{w_1+\dots+w_N+\lambda}\sum_{\sigma\in S_{N},\sigma(N)=l}(-1)^{|\sigma|}\prod_{k=1}^{N-1}\frac{w_{\sigma(k)}^{k}}{w_{\sigma(k)}+\lambda},
\end{aligned}\end{equation}
where we applied the induction hypothesis to the second sum. Further,
\begin{equation}\begin{aligned}\label{eqDI2}
	\eqref{eqDI1}=&\sum_{\sigma\in \mathcal{S}_N}(-1)^{|\sigma|}\frac{w_{\sigma(N)}+\lambda}{w_1+\dots+w_N+\lambda}\prod_{k=1}^{N}\frac{w_{\sigma(k)}^{k}}{w_{\sigma(k)}+\lambda}\\
	&=\frac{1}{w_1+\dots+w_N+\lambda}\prod_{l=1}^N\frac{w_l}{w_l+\lambda}\\
&\quad\times\bigg(\sum_{\sigma\in \mathcal{S}_N}(-1)^{|\sigma|}w_{\sigma(N)}\prod_{k=1}^Nw_{\sigma(k)}^{k-1}+\lambda\sum_{\sigma\in \mathcal{S}_N}(-1)^{|\sigma|}\prod_{k=1}^Nw_{\sigma(k)}^{k-1}\bigg)\\
	&=\frac{1}{w_1+\dots+w_N+\lambda}\prod_{l=1}^N\frac{w_l}{w_l+\lambda}\\
&\quad \times\left(\det_{1\leq k,l\leq N}\left[w_l^{k-1+\delta_{k,N}}\right]+\lambda\det_{1\leq k,l\leq N}\left[w_l^{k-1}\right]\right),
\end{aligned}\end{equation}
with $\delta_{k,N}$ being the Kronecker delta. Inserting the identity
\begin{equation}\label{eqMuir}
	\det_{1\leq k,l\leq N}\left[w_l^{k-1+\delta_{k,N}}\right]=(w_1+\dots+w_N)\det_{1\leq k,l\leq N}\left[w_l^{k-1}\right],
\end{equation}
we arrive at
\begin{equation}
\eqref{eqDI1}=\bigg(\prod_{l=1}^N\frac{w_l}{w_l+\lambda}\bigg)\det_{1\leq k,l\leq N}\left[w_l^{k-1}\right]=\det_{1\leq k,l\leq N}\bigg[\frac{w_l^k}{w_l+\lambda}\bigg].
\end{equation}

To show \eqref{eqMuir} we introduce the variable $w_{N+1}$ and consider the factorization
\begin{equation}
	\det_{1\leq k,l\leq N+1}\left[w_l^{k-1}\right]=\prod_{i=1}^N(w_{N+1}-w_i)\det_{1\leq k,l\leq N}\left[w_l^{k-1}\right],
\end{equation}
which follows directly from the explicit formula for a Vandermonde determinant. Expanding the determinant on the left hand side along the $(N+1)$-th column gives an explicit expression in terms of monomials in $w_{N+1}$. Examining the coefficient of $w_{N+1}^{N-1}$ on the left and right hand side respectively provides \eqref{eqMuir}.
\end{proof}

%\section{Proof of Proposition~\ref{propKernel}}
We can rewrite the measure in Proposition~\ref{PropRandomIC} in terms of a conditional \mbox{$L$-ensemble} (see Lemma 3.4 of~\cite{BFPS06} reported here as Lemma~\ref{lemDetMeasure}) and obtain a Fredholm determinant expression for the joint distribution of any subsets of particles position. Then it remains to relate the law under $\Pb_+$ and $\Pb$, which is the law of the reflected Brownian motions specified by the initial condition \eqref{statModel}. This is made using a \emph{shift argument}, analogue to the one used for the polynuclear growth model with external sources~\cite{BR00,SI04} or in the totally asymmetric simple exclusion process~\cite{PS02b,FS05a,BFP09}.

\begin{proof}[Proof of Proposition~\ref{propKernel}]
The proof is divided into two steps. In \emph{Step 1} we determine the distribution under $\Pb_+$ and in \emph{Step 2} we extend this result via a shift argument to $\Pb$.

\emph{Step 1.} We consider the law of the process under $\Pb_+$ for now. The first part of the proof is identical to the proof of Proposition 3.5~\cite{FSW13}, so it is only sketched here.
Using repeatedly the identity
\begin{equation}
	\widetilde{F}_{k}(\xi,{t})=\int^\xi_{-\infty} \mathrm{d} x\,\widetilde{F}_{k+1}(x,{t}),
\end{equation}
relabeling $\xi_1^k:=\xi_{k-1}$, and introducing new variables $\xi_l^k$ for $2\leq l\leq k\leq N$,
we can write
\begin{equation}
	\det_{1\leq k,l\leq N}\big[\widetilde{F}_{k-l}(\xi_1^{N+1-l},{t})\big]=\int_{\mathcal D'} \det_{1\leq k,l\leq N}\big[\widetilde{F}_{k-1}(\xi_l^{N},{t})\big]\prod_{2\leq l\leq k\leq N} \D\xi_l^k,
\end{equation}
where $\mathcal D' = \{\xi_l^k\in\R,2\leq l\leq k\leq N|x_l^k\leq x_{l-1}^{k-1}\}$. Proceeding as in the proof of Proposition~\ref{propStepKernel}, we obtain that the measure \eqref{Prnd} is a marginal of
\begin{equation}\label{ExtM}\begin{aligned}
	\const&\cdot e^{\rho\xi^1_1}\prod_{n=2}^{N}\det_{1\leq i,j\leq n}\left[\Id_{\xi_i^{n-1}\leq\xi_j^n}\right]\det_{1\leq k,l\leq N}\big[ \widetilde{F}_{k-1}(\xi_l^{N},{t})\big]\\=
	\const&\cdot \prod_{n=1}^{N}\det_{1\leq i,j\leq n}\big[\tilde{\phi}_n(\xi_i^{n-1},\xi_j^n)\big]\det_{1\leq k,l\leq N}\big[ \widetilde{F}_{k-1}(\xi_l^{N},{t})\big]
\end{aligned}\end{equation}
with
\begin{equation}\begin{aligned}
	\tilde{\phi}_n(x,y)&=\Id_{x\leq y}, \quad \text{for } n\geq2\\
	\tilde{\phi}_1(x,y)&=e^{\rho y},
\end{aligned}\end{equation}
and using the convention that $\xi_n^{n-1}\leq y$ always holds.

The measure \eqref{ExtM} has the appropriate form for applying Lemma~\ref{lemDetMeasure}. The composition of the $\tilde{\phi}$ functions can be evaluated explicitly as
\begin{equation}\begin{aligned}
	\tilde{\phi}_{0,n}(x,y)&=(\tilde{\phi}_1*\dots*\tilde{\phi}_{n})(x,y)=\rho^{1-n}e^{\rho y}, & &\text{for } n\geq1,\\
	\tilde{\phi}_{m,n}(x,y)&=(\tilde{\phi}_{m+1}*\dots*\tilde{\phi}_{n})(x,y)=\frac{(y-x)^{n-m-1}}{(n-m-1)!}\Id_{x\leq y}, & &\text{for } n>m\geq1.
\end{aligned}\end{equation}
Define
\begin{equation}
	\Psi^n_{n-k}(\xi):=\frac{(-1)^{n-k}}{2\pi\mathrm{i}}\int_{\mathrm{i}\R-\e}\mathrm{d} w\,\,\frac{e^{{t} w^2/2+\xi w}w^{n-k}}{w+\lambda},
\end{equation}
for $n,k\geq1$ and some $0<\e<\lambda$. In the case $n\geq k$ the integrand has no poles in the region $|w|<\lambda$, which implies $\Psi^n_{n-k}=(-1)^{n-k}\widetilde{F}_{n-k}$. The straightforward recursion
\begin{equation}
	(\tilde{\phi}_n*\Psi^n_{n-k})(\xi)=\Psi^{n-1}_{n-1-k}(\xi)
\end{equation}
eventually leads to condition \eqref{Sasdef_psi} being satisfied.

The space $V_n$ is generated by
\begin{equation}
	\{\tilde{\phi}_{0,n}(\xi_1^0,x),\dots,\tilde{\phi}_{n-2,n}(\xi_{n-1}^{n-2},x),\tilde{\phi}_{n-1,n}(\xi_n^{n-1},x)\},
\end{equation}
so a basis for $V_n$ is given by
\begin{equation}
	\{e^{\rho x},x^{n-2},x^{n-3},\dots,x,1\}.
\end{equation}
Choose functions $\Phi^n_{n-k}$ as follows
\begin{equation}
		\Phi^n_{n-k}(\xi)=\begin{cases}
\frac{(-1)^{n-k}}{2\pi\mathrm{i}}\oint_{\Gamma_0}\mathrm{d} z\,\frac{z+\lambda}{e^{{t} z^2/2+\xi z}z^{n-k+1}}& 2\leq k\leq n,\\
\frac{(-1)^{n-1}}{2\pi\mathrm{i}}\oint_{\Gamma_{0,-\rho}}\mathrm{d} z\,\frac{z+\lambda}{e^{{t} z^2/2+\xi z}z^{n-1}(z+\rho)}& k=1.
\end{cases}
\end{equation}
By residue calculating rules, $\Phi^n_{n-k}$ is a polynomial of order $n-k$ for $k\geq2$ and a linear combination of $1$ and $e^{\rho\xi}$
for $k=1$, so these functions indeed generate $V_n$. To show \eqref{Sasortho} for $\ell\geq2$, we decompose the scalar product as follows:
\begin{equation}\label{eq2.41}
\int_{\R_-} \mathrm{d} \xi\, \Psi^{n}_{n-k}(\xi) \Phi^n_{n-\ell}(\xi) + \int_{\R_+} \mathrm{d} \xi\, \Psi^{n}_{n-k}(\xi) \Phi^n_{n-\ell}(\xi).
\end{equation}
Since $n-k\geq0$ we are free to choose the sign of $\e$ as necessary. For the first term, we choose $\e<0$ and the path $\Gamma_0$ close enough to zero, such that always \mbox{$\Re(w-z)>0$}. Then, we can take the integral over $\xi$ inside and obtain
\begin{equation}\begin{aligned}
\int_{\R_-} \mathrm{d} \xi\,&\Psi^{n}_{n-k}(\xi) \Phi^n_{n-\ell}(\xi)\\
&=\frac{(-1)^{k-l}}{(2\pi\mathrm{i})^2}\int_{\mathrm{i}\R-\e}\mathrm{d} w \oint_{\Gamma_0}\mathrm{d} z\, \frac{e^{{t} w^2/2} w^{n-k}(z+\lambda)}{e^{{t} z^2/2}z^{n-\ell+1}(w+\lambda)(w-z)}.
\end{aligned}\end{equation}
For the second term, we choose $\e>0$ to obtain \mbox{$\Re(w-z)<0$}. Then again, we can take the integral over $\xi$ inside and arrive at the same expression up to a minus sign. The net result of \eqref{eq2.41} is a residue at $w=z$, which is given by
\begin{equation}
\frac{(-1)^{k-l}}{2\pi\mathrm{i}}\oint_{\Gamma_0}\mathrm{d} z\, z^{\ell-k-1}=\delta_{k,\ell}.
\end{equation}
The case $\ell=1$ uses the same decomposition and requires the choice $\e>\rho$ resp.\ $\e<0$, finally leading to
\begin{equation}
\eqref{eq2.41}=\frac{(-1)^{k-1}}{2\pi\mathrm{i}}\oint_{\Gamma_{0,-\rho}}\mathrm{d} z\, \frac{z^{1-k}}{z+\rho}=\delta_{k,1}.
\end{equation}

Furthermore, both $\tilde{\phi}_n(\xi_{n}^{n-1},x)$ and $\Phi_0^{n}(\xi)$ are constants, so the kernel has a simple form (compare with \eqref{SasK})
\begin{equation}
\tilde{\mathcal{K}}(n_1,\xi_1;n_2,\xi_2)=-\tilde{\phi}_{n_1,n_2}(\xi_1,\xi_2)\Id_{n_2>n_1} + \sum_{k=1}^{n_2} \Psi_{n_1-k}^{n_1}(\xi_1) \Phi_{n_2-k}^{n_2}(\xi_2).
\end{equation}
However, the relabeling $\xi_1^k:=\xi_{k-1}$ included an index shift, so the kernel of our system is actually
\begin{equation}\begin{aligned}
\mathcal{K}_\textrm{stat}(n_1,\xi_1;n_2,\xi_2)&=\tilde{\mathcal{K}}(n_1+1,\xi_1;n_2+1,\xi_2)\\&=-\phi_{n_1,n_2}(\xi_1,\xi_2)\Id_{n_2>n_1} + \sum_{k=1}^{n_2} \Psi_{n_1-k+1}^{n_1+1}(\xi_1) \Phi_{n_2-k+1}^{n_2+1}(\xi_2).
\end{aligned}\end{equation}

Note that we are free to extend the summation over $k$ up to infinity, since the integral expression for $\Phi_{n-k}^{n}(\xi)$ vanishes for $k>n$ anyway. Taking the sum inside the integrals we can write
\begin{equation}\label{eqKernSum}\begin{aligned}
	\sum_{k\geq1} \Psi_{n_1-k+1}^{n_1+1}(\xi_1) \Phi_{n_2-k+1}^{n_2+1}(\xi_2)=\frac{1}{(2\pi\mathrm{i})^2}\int_{\mathrm{i}\R-\e}\hspace{-0.5cm}\mathrm{d} w \oint_{\Gamma_{0,-\rho}}\hspace{-0.5cm}\mathrm{d} z\,	\frac{e^{{t} w^2/2+\xi_1 w}(-w)^{n_1}}{e^{{t} z^2/2+\xi_2 z}(-z)^{n_2}}\eta(w,z),
\end{aligned}\end{equation}
with
\begin{equation}
	\eta(w,z)=\frac{z+\lambda}{(w+\lambda)(z+\rho)}+\sum_{k\geq2}\frac{z^{k-2}(z+\lambda)}{w^{k-1}(w+\lambda)}.
\end{equation}
By choosing contours such that $|z|<|w|$, we can use the formula for a geometric series, resulting in
\begin{equation}\begin{aligned}
	\eta(w,z)&=\frac{z+\lambda}{(w+\lambda)(z+\rho)}+\frac{z+\lambda}{(w+\lambda)w}\frac{1}{1-z/w}\\
	&=\frac{1}{w-z}+\frac{\lambda-\rho}{(w+\lambda)(z+\rho)}.
\end{aligned}\end{equation}
Inserting this expression back into \eqref{eqKernSum} gives the kernel \eqref{eqKtStat}, which governs the multidimensional distributions of $x_n({t})$ under the measure $\Pb_+$, namely
\begin{equation}\label{eqDistP+}
\Pb_+\bigg(\bigcap_{n\in S} \{x_n({t})\leq a_n\}\bigg)=\det(\Id-\chi_a \mathcal{K}_\textrm{stat} \chi_a)_{L^2(S\times\R)}.
\end{equation}

\emph{Step 2.} The distributions under $\Pb$ and under $\Pb_+$ can be related via the following \emph{shift argument}. Introducing the shorthand
\begin{equation}
	\widetilde{\mathcal{E}}(S,\vec{a}):=\bigcap_{n\in S} \{x_n({t})\leq a_n\},
\end{equation}
we have
\begin{equation}\begin{aligned}
		\Pb_+&(\widetilde{\mathcal{E}}(S,\vec{a}))=\int_{\R_+} \mathrm{d} x\, \Pb_+(x_0(0)\in \mathrm{d} x)\Pb_+(\widetilde{\mathcal{E}}(S,\vec{a})|x_0(0)=x)\\
		&=\int_{\R_+} \mathrm{d} x\,(\lambda-\rho)e^{-(\lambda-\rho)x}\Pb(\widetilde{\mathcal{E}}(S,\vec{a}-x))\\
		&=-e^{-(\lambda-\rho)x}\,\Pb(\widetilde{\mathcal{E}}(S,\vec{a}-x))\Big|_0^\infty+\int_{\R_+} \mathrm{d} x\,e^{-(\lambda-\rho)x}\frac{\D}{\mathrm{d} x}\Pb(\widetilde{\mathcal{E}}(S,\vec{a}-x))\\
		&=\Pb(\widetilde{\mathcal{E}}(S,\vec{a}))-\int_{\R_+} \mathrm{d} x\,e^{-(\lambda-\rho)x}\sum_{n\in S}\frac{\D}{\mathrm{d} a_n}\Pb(\widetilde{\mathcal{E}}(S,\vec{a}))\\
		&=\Pb(\widetilde{\mathcal{E}}(S,\vec{a}))-\frac{1}{\lambda-\rho}\sum_{n\in S}\frac{\D}{\mathrm{d} a_n}\Pb_+(\widetilde{\mathcal{E}}(S,\vec{a})).
\end{aligned}\end{equation}
Combining the identity
\begin{equation}
	\Pb(\widetilde{\mathcal{E}}(S,\vec{a}))=\bigg(1+\frac{1}{\lambda-\rho}\sum_{n\in S}\frac{\D}{\mathrm{d} a_n}\bigg)\Pb_+(\widetilde{\mathcal{E}}(S,\vec{a}))
\end{equation}
with \eqref{eqDistP+} finishes the proof.
\end{proof}

\section{Asymptotic analysis}
Noticing that the change in variables
\begin{equation}\begin{aligned}
	x&\mapsto\lambda^{-1}x&\quad {t}&\mapsto\lambda^{-2}{t}
\end{aligned}\end{equation}
reproduces the same system with new parameters $\widetilde{\lambda}=1$ and $\widetilde{\rho}=\frac{\rho}{\lambda}$, we can restrict our considerations to $\lambda=1$ without loss of generality.
According to \eqref{eqScaledProcess} we use the scaled variables
\begin{equation}\label{eqScaling}\begin{aligned}
	n_i&={t}+2{t}^{2/3}r_i\\
	\xi_i&=2{t}+2{t}^{2/3}r_i+{t}^{1/3}s_i\\
	\rho&=1-{t}^{-1/3}\delta,
\end{aligned}\end{equation}
with $\delta>0$. Correspondingly, consider the rescaled (and conjugated) kernel
\begin{equation}
	\mathcal{K}^\textrm{resc}(r_1,s_1;r_2,s_2)={t}^{1/3}e^{\xi_1-\xi_2}\mathcal{K}(n_1,\xi_1;n_2,\xi_2),
\end{equation}
which decomposes into
\begin{equation}\begin{aligned}
	\mathcal{K}^\textrm{resc}(r_1,s_1;r_2,s_2)&=-\phi_{r_1,r_2}^\textrm{resc}(s_1,s_2)\Id_{r_1<r_2}+\mathcal{K}_0^\textrm{resc}(r_1,s_1;r_2,s_2)\\
	&\quad+\delta\mathpzc{f}^\textrm{resc}(r_1,s_1)\mathpzc{g}^\textrm{resc}(r_2,s_2),
\end{aligned}\end{equation}
by
\begin{equation}\begin{aligned}
	\mathpzc{f}^\textrm{resc}(r_1,s_1)&=e^{-{t}/2+\xi_1}\mathpzc{f}(n_1,\xi_1)\\
	\mathpzc{g}^\textrm{resc}(r_2,s_2)&=e^{{t}/2-\xi_2}\mathpzc{g}(n_2,\xi_2).
\end{aligned}\end{equation}

As before, the proof of Theorem~\ref{thmAsymp} relies both on the pointwise convergence as well as a uniform bound for the rescaled kernel:
%In order to establish the asymptotics of the joint distributions, one needs both a pointwise limit of the kernel, as well as uniform bounds to ensure convergence of the Fredholm determinant expansion. The first time this approach was used is in~\cite{GTW00}. These results are contained in the following propositions.
\begin{prop}\label{propPointw}
Consider any $r_1,r_2$ in a bounded set and fixed $L$. Then, the kernel converges as
\begin{equation}
	\lim_{{t}\to\infty}\mathcal{K}_\textrm{stat}^\textrm{resc}(r_1,s_1;r_2,s_2)=K^\delta(r_1,s_1;r_2,s_2)
\end{equation}
uniformly for $(s_1,s_2)\in[-L,L]^2$.
\end{prop}

\begin{cor}\label{corBound}
Consider $r_1,r_2$ fixed. For any $L$ there exists ${t}_0$ such that for ${t}>{t}_0$ the bound
\begin{equation}
	\left|\mathcal{K}_\textrm{stat}^\textrm{resc}(r_1,s_1;r_2,s_2)\right|\leq \const_L
\end{equation}
holds for all $(s_1,s_2)\in[-L,L]^2$.
\end{cor}

\begin{prop}\label{propK0Bound}
For fixed $r_1,r_2,L$ and $\delta>0$ there exists ${t}_0>0$ such that the estimate
\begin{equation}
	\left|\delta\mathpzc{f}^\textrm{resc}(r_1,s_1)\mathpzc{g}^\textrm{resc}(r_2,s_2)\right|\leq \const\cdot e^{-\min\{\delta,1\} s_2}
\end{equation}
holds for any ${t}>{t}_0$ and $s_1,s_2>0$.
\end{prop}

Now we can prove the asymptotic theorem:
\begin{proof}[Proof of Theorem~\ref{thmAsymp}]The joint distributions of the rescaled process $X^{(\delta)}_{t}(r)$ under the measure $\Pb_+$ are given by the Fredholm determinant with series expansion
\begin{equation}\label{eqFredExp}\begin{aligned}
	\Pb_+&\bigg(\bigcap_{k=1}^m\big\{X_{t}^{(\delta)}(r_k)\leq s_k\big\}\bigg)\\
	&=\sum_{N\geq0}\frac{(-1)^N}{N!}\sum_{i_1,\dots,i_N=1}^m\int\prod_{k=1}^N\mathrm{d} x_k\,\Id_{x_k>\xi_{i_k}}\det_{1\leq k,l\leq N}\left[\mathcal{K}_\textrm{stat}(n_{i_k},\xi_k;n_{i_l},\xi_l)\right],
\end{aligned}\end{equation}
where $n_i$ and $\xi_i$ are understood as in (\ref{eqScaling}). Using the change of variables \mbox{$\sigma_k={t}^{-1/3}(x_k-2{t}-2{t}^{2/3}r_{i_k})$}, we obtain
\begin{equation}\label{eqFredExp2}\begin{aligned}
	\eqref{eqFredExp}&=\sum_{N\geq0}\frac{(-1)^N}{N!}\sum_{i_1,\dots,i_N=1}^m\int\prod_{k=1}^N\mathrm{d} \sigma_k\,\Id_{\sigma_k>s_{i_k}}\\
	&\quad\times\det_{1\leq k,l\leq N}\left[\mathcal{K}_\textrm{stat}^\textrm{resc}(r_k,\sigma_k;r_l,\sigma_l)\frac{(1+\sigma_l^2)^{m+1-i_l}}{(1+\sigma_k^2)^{m+1-i_k}}\right],
\end{aligned}\end{equation}
where the fraction inside the determinant is a new conjugation, which does not change the value of the determinant. This conjugation helps bounding the first part of the kernel, which would otherwise not decay when some $\sigma_k$ and $\sigma_l$ are close. By Proposition~\ref{propStepK0Bound} we have
\begin{equation}
 \left|\mathcal{K}_0^\textrm{resc}(r_1,s_1;r_2,s_2)\right|\leq\frac{1}{2} e^{-(s_1+s_2)}\leq \const\cdot e^{-\min\{\delta,1\}s_2}.
\end{equation}
Using this estimate together with Corollary~\ref{corBound} and Propositions~\ref{propK0Bound},~\ref{propStepPhiBound}, we can bound the $(k,l)$-coefficient inside the determinant by
\begin{equation}\label{eqCoeff}
	\const_1\left(e^{-|\sigma_k-\sigma_l|}\Id_{i_k<i_l}+e^{-\min\{\delta,1\}\sigma_l}\right)\frac{(1+\sigma_l^2)^{m+1-i_l}}{(1+\sigma_k^2)^{m+1-i_k}},
\end{equation}
assuming the $r_k$ are ordered. The bounds
\begin{equation}\begin{aligned}
	\frac{(1+x^2)^i}{(1+y^2)^j}e^{-|x-y|}&\leq \const_2\frac{1}{1+y^2},& \text{for } i&<j,\\
	\frac{(1+x^2)^i}{(1+y^2)^j}e^{-\min\{\delta,1\} x}&\leq \const_3\frac{1}{1+y^2},& \text{for } j&\geq1,
\end{aligned}\end{equation}
lead to
\begin{equation}
	\eqref{eqCoeff}\leq\const_4\frac{1}{1+\sigma_k^2}.
\end{equation}
Using the Hadamard bound on the determinant, the integrand of \eqref{eqFredExp2} is therefore bounded by
\begin{equation}
	\const_4^NN^{N/2}\prod_{k=1}^N\Id_{\sigma_k>s_{i_k}}\frac{\D\sigma_k}{1+\sigma_k^2},
\end{equation}
which is integrable. Furthermore,
\begin{equation}
	|\eqref{eqFredExp}|\leq\sum_{N\geq0}\frac{\const_5^NN^{N/2}}{N!},
\end{equation}
which is summable, since the factorial grows like $(N/e)^N$, i.e., much faster than the numerator. Dominated convergence thus allows to interchange the limit ${t}\to\infty$ with the integral and the infinite sum. The pointwise convergence comes from Proposition~\ref{propPointw}, thus
\begin{equation}
	\lim_{{t}\to\infty} \Pb_+\bigg(\bigcap_{k=1}^m\big\{X_{t}^{(\delta)}(r_k)\leq s_k\big\}\bigg)=\det\left(\Id-\chi_s K^\delta\chi_s\right)_{L^2(\{r_1,\dots,r_m\}\times\R)}.
\end{equation}

It remains to show that the convergence carries over to the measure $\Pb$. The identity
\begin{equation}
	\frac{\mathrm{d} s_i}{\D\xi_i}={t}^{-1/3}=\delta^{-1}(1-\rho)
\end{equation}
leads to
\begin{equation}\label{eqShift}
\Pb\bigg(\bigcap_{k=1}^m\big\{X_{t}^{(\delta)}(r_k)\leq s_k\big\}\bigg)=\bigg(1+\frac{1}{\delta}\sum_{i=1}^m\frac{\D}{\mathrm{d} s_i}\bigg)\Pb_+\bigg(\bigcap_{k=1}^m\big\{X_{t}^{(\delta)}(r_k)\leq s_k\big\}\bigg).
\end{equation}
Notice that in \eqref{eqFredExp2}, $s_i$ appears only in the indicator function, so differentiation just results in one of the $\sigma_k$ not being integrated but instead being set to $s_i$. Using the same bounds as before we can again show interchangeability of the limit ${t}\to\infty$ with the remaining integrals and the infinite sum.
\end{proof}

\begin{proof}[Proof of Proposition~\ref{propPointw}]
The kernel $\mathcal{K}_\textrm{stat}^\textrm{resc}$ consists of three parts, where compact convergence for the first two parts of the kernel comes directly from Proposition~\ref{propStepPointw}.

As we did with $\mathcal{K}_0^\textrm{resc}$, we rewrite the third part, which is the product of $\mathpzc{f}^\textrm{resc}$ and $\mathpzc{g}^\textrm{resc}$, as integrals over the previously defined functions $\alpha$ and $\beta$:
\begin{equation}\begin{aligned}\label{eqfInt}
&\mathpzc{f}^\textrm{resc}(r_1,s_1)=\frac{1}{2\pi\mathrm{i}}\int_{\mathrm{i}\R-\e}\mathrm{d} w\,\frac{e^{{t} (w^2-1)/2+\xi_1(w+1)}(-w)^{n_1}}{w+1}\\
	&=1+\frac{1}{2\pi\mathrm{i}}\int_{\mathrm{i}\R-\e-1}\mathrm{d} w\,\frac{e^{{t} (w^2-1)/2+\xi_1(w+1)}(-w)^{n_1}}{w+1}\\
	&=1-\frac{1}{2\pi\mathrm{i}}\int_{\mathrm{i}\R-\e-1}\mathrm{d} w\,e^{{t} (w^2-1)/2+\xi_1(w+1)}(-w)^{n_1}\int_0^\infty \mathrm{d} x\, {t}^{1/3}e^{{t}^{1/3}x(w+1)}\\&=1-\int_0^\infty \mathrm{d} x\, \alpha_{{t}}(r_1,s_1+x).
\end{aligned}\end{equation}

Now,
\begin{equation}\begin{aligned}\label{eqfSup}
	\sup_{s_1\in[-L,L]}&\left|\int_0^\infty \mathrm{d} x\, \alpha_{{t}}(r_1,s_1+x)-\int_0^\infty \mathrm{d} x\, \alpha(r_1,s_1+x)\right|\\
	\leq&\int_0^\infty \mathrm{d} x\,\sup_{s_1,s_2\in[-L,L]}\left|\alpha_{{t}}(r_1,s_1+x)-\alpha(r_1,s_1+x)\right|.
\end{aligned}\end{equation}
By Lemma~\ref{lemAlphaLimit} the integrand converges to zero for every $x>0$. Using Lemma~\ref{lemAlphaBound} we can bound it by $\const\cdot e^{-2x}$, thus ensuring that \eqref{eqfSup} goes to zero, i.e., $\mathpzc{f}^\textrm{resc}(r_1,s_1)$ converges to $f_{r_1}(s_1)$ uniformly on compact sets.

Similarly,
\begin{equation}\begin{aligned}\label{eqgInt}
	\mathpzc{g}^\textrm{resc}(r_2,s_2)&=\textrm{Res}_{\mathpzc{g},-\rho}+\int_0^\infty \mathrm{d} x\, \beta_{{t}}(r_2,s_2+x)e^{\delta x},
\end{aligned}\end{equation}
with
\begin{equation}
	\textrm{Res}_{\mathpzc{g},-\rho}=e^{{t}^{2/3}\delta-{t}^{1/3}\delta^2/2-\xi_2{t}^{-1/3}\delta}(1-{t}^{-1/3}\delta)^{-n_2}.
\end{equation}
The residuum satisfies the limit
\begin{equation}\label{eqResLimit}
	\lim_{{t}\to\infty}\textrm{Res}_{\mathpzc{g},-\rho}=e^{\delta^3/3+r_2\delta^2-s_2\delta}
\end{equation}
uniformly in $s_2$. By the same argument, uniform convergence holds again.
\end{proof}

\begin{proof}[Proof of Proposition~\ref{propK0Bound}]
The product $\delta\mathpzc{f}^\textrm{resc}(r_1,s_1)\mathpzc{g}^\textrm{resc}(r_2,s_2)$ can be bounded using Lemma~\ref{lemAlphaBound} in the representations \eqref{eqfInt} and \eqref{eqgInt}:

\begin{equation}\begin{aligned}\label{eqfgbound}
		\left|\delta\mathpzc{f}^\textrm{resc}(r_1,s_1)\mathpzc{g}^\textrm{resc}(r_2,s_2)\right|&\leq\delta\left(1+\int_0^\infty \mathrm{d} x\, e^{-(s_1+x)}\right)\\&\quad\cdot\left(\textrm{Res}_{\mathpzc{g},-\rho}+\int_0^\infty \mathrm{d} x\, e^{-(s_2+x)}e^{\delta x}\right)\\
		&=\delta\left(1+e^{-s_1}\right)\left(\textrm{Res}_{\mathpzc{g},-\rho}+\frac{e^{-s_2}}{1-\delta}\right)
\end{aligned}\end{equation}
Since the convergence \eqref{eqResLimit} is uniform in $s_2$ we can deduce
\begin{equation}
	\left|\textrm{Res}_{\mathpzc{g},-\rho}\right|\leq \const_1\cdot e^{-s_2\delta},
\end{equation}
resulting in
\begin{equation}
 |\eqref{eqfgbound}|\leq\const\cdot e^{-\min\{\delta,1\} s_2}.
\end{equation}

\end{proof}
\section{Path-integral style formula}
Using the results from~\cite{BCR13} we can transform the formula for the joint probability distribution of the finite-step Airy$_\textrm{stat}$ process from the current form involving a Fredholm determinant over the space \mbox{$L^2(\{r_1,\dots,r_m\}\times\R)$} into a path-integral style form, where the Fredholm determinant is over the simpler space $L^2(\R)$. The result of~\cite{BCR13} can not be applied at the stage of finite time as one of the assumption is not satisfied.
\begin{prop}\label{propPathInt}
For any parameters $\chi_k\in\R$, $1\leq k\leq m$, satisfying
\begin{equation}
  0<\chi_m<\dots<\chi_2<\chi_1<\max_{i<j}\left\{r_j-r_i,\delta\right\},
\end{equation}
define the multiplication operator $(M_{r_i}f)(x)=m_{r_i}(x)f(x)$, with
\begin{equation}
 m_{r_i}(x)=\begin{cases}e^{-\chi_ix} &\mbox{for } x\geq0 \\ e^{x^2} &\mbox{for } x<0.\end{cases}
\end{equation}
Writing $K^\delta_{r_i}(x,y):=K^\delta(r_i,x;r_i,y)$, the finite-dimensional distributions of the finite-step Airy$_\textrm{stat}$ process are given by
\begin{equation}\label{PathIntForm}\begin{aligned}
	\Pb &\bigg(\bigcap_{k=1}^m\big\{\mathcal{A}_\mathrm{stat}^{(\delta)}(r_k)\leq s_k\big\}\bigg)=\bigg(1+\frac{1}{\delta}\sum_{i=1}^m\frac{\D}{\mathrm{d} s_i}\bigg)
\det\Big(\Id+M_{r_1} Q M_{r_1}^{-1}\Big)_{L^2(\R)},
\end{aligned}\end{equation}
with
\begin{equation}
Q=-K^{\delta}_{r_1}+\bar{P}_{s_1}V_{r_1,r_2}\bar{P}_{s_2}\cdots V_{r_{m-1},r_m}\bar{P}_{s_m}V_{r_m,r_1}K^\delta_{r_1},
\end{equation}
where $\bar{P}_s=\Id-P_s$ denotes the projection operator on $(-\infty,s)$.
\end{prop}
\begin{remark}\label{remV}
 The operator $V_{r_j,r_i}$ for $r_i<r_j$ is defined only on the range of $K^\delta_{r_i}$ and acts on it in the following way:
 \begin{equation}
   V_{r_j,r_i}K_{r_i,r_k}=K_{r_j,r_k},\qquad
   V_{r_j,r_i}f_{r_i}=f_{r_j}.
 \end{equation}
In particular, we have also $V_{r_j,r_i}\mathbf{1}=\mathbf{1}$.
\end{remark}

\begin{proof}
We will denote conjugations by the operator $M$ by a hat in the following way:
\begin{equation}
 \begin{aligned}
  \widehat{V}_{r_i,r_j}&=M_{r_i}V_{r_i,r_j}M^{-1}_{r_j}, &\quad \widehat{f}_{r_i}&=M_{r_i}f_{r_i},\\
  \widehat{K}^\delta_{r_i}&=M_{r_i}K^\delta_{r_i}M^{-1}_{r_i}, &\quad \widehat{g}_{r_i}&=g_{r_i}M_{r_i}^{-1},\\
  \widehat{K}_{r_i,r_j}&=M_{r_i}K_{r_i,r_j}M^{-1}_{r_j}.
 \end{aligned}
\end{equation}
Applying the conjugation also in the determinant in \eqref{eq3.6}, the identity we have to show is:
\begin{equation}
\begin{aligned}
  \det&\left(\Id-\chi_s \widehat{K}^\delta\chi_s\right)_{L^2(\{r_1,\dots,r_m\}\times\R)}\\&= \det\Big(\Id-\widehat{K}^{\delta}_{r_1}+\bar{P}_{s_1}\widehat{V}_{r_1,r_2}\bar{P}_{s_2}\cdots \widehat{V}_{r_{m-1},r_m}\bar{P}_{s_m}\widehat{V}_{r_m,r_1}\widehat{K}^\delta_{r_1}\Big)_{L^2(\R)}
\end{aligned}
\end{equation}
This is done by applying Theorem 1.1~\cite{BCR13}.

It has three groups of assumptions we have to prove. We merged them into two by choosing the multiplication operators of Assumption 3 to be the identity.
\paragraph{Assumption 1}
\renewcommand{\theenumi}{\roman{enumi}}
\renewcommand{\labelenumi}{(\theenumi)}
\begin{enumerate}
	\item The operators $P_{s_i}\widehat{V}_{r_i,r_j}$, $P_{s_i}\widehat{K}^\delta_{r_i}$, $P_{s_i}\widehat{V}_{r_i,r_j}\widehat{K}^\delta_{r_j}$ and $P_{s_j}\widehat{V}_{r_j,r_i}\widehat{K}^\delta_{r_i}$ for $r_i<r_j$ preserve $L^2(\R)$ and are trace class in $L^2(\R)$.
	\item The operator $\widehat{V}_{r_i,r_1}\widehat{K}^{\delta}_{r_1}-\bar{P}_{s_i}\widehat{V}_{r_i,r_{i+1}}\bar{P}_{s_{i+1}}\cdots \widehat{V}_{r_{m-1},r_m}\bar{P}_{s_m}\widehat{V}_{r_m,r_1}\widehat{K}^\delta_{r_1}$ preserves $L^2(\R)$ and is trace class in $L^2(\R)$.
\end{enumerate}
\paragraph{Assumption 2}
\begin{enumerate}
	\item Right-invertibility: $\widehat{V}_{r_i,r_j}\widehat{V}_{r_j,r_i}\widehat{K}^\delta_{r_i}=\widehat{K}^\delta_{r_i}$
	\item Semigroup property: $\widehat{V}_{r_i,r_j}\widehat{V}_{r_j,r_k}=\widehat{V}_{r_i,r_k}$
	\item Reversibility relation: $\widehat{V}_{r_i,r_j}\widehat{K}^\delta_{r_j}=\widehat{K}^\delta_{r_i}\widehat{V}_{r_i,r_j}$
\end{enumerate}

The semigroup property is clear. To see the reversibility relation, start from the contour integral representation \eqref{contIntK} of $K_{r_j,r_j}$ and \eqref{contIntfg} of $f_{r_j}$ and use the Gaussian identity:
\begin{equation}\begin{aligned}
	\int_\R \mathrm{d} z \frac{1}{\sqrt{4\pi (r_j-r_i)}}e^{-(z-x)^2/4(r_j-r_i)}e^{-r_jW^2+zW}=e^{-r_iW^2+xW}.
\end{aligned}\end{equation}
This results in $\widehat{V}_{r_i,r_j}\widehat{K}^\delta_{r_j}=\widehat{K}_{r_i,r_j}+\delta \widehat{f}_{r_i}\otimes \widehat{g}_{r_j}$.
On the other hand we have
\begin{equation}\begin{aligned}
	\int_\R \mathrm{d} z \frac{1}{\sqrt{4\pi (r_j-r_i)}}e^{-(z-y)^2/4(r_j-r_i)}e^{r_iZ^2-zZ}=e^{r_jZ^2-yZ},
\end{aligned}\end{equation}
so $\widehat{K}^\delta_{r_i}\widehat{V}_{r_i,r_j}=\widehat{K}_{r_i,r_j}+\delta \widehat{f}_{r_i}\otimes \widehat{g}_{r_j}$, which proves Assumption 2 (iii). Noticing Remark~\ref{remV}, the right-invertibility follows immediately.

Assumption 1 (ii) can be deduced from Assumption 1 (i) as shown in Remark~3.2,~\cite{BCR13}. Using the previous identities we thus are left to show that the three operators $P_{s_i}\widehat{V}_{r_i,r_j}$, for $r_i<r_j$, as well as $P_{s_i}\widehat{K}_{r_i,r_j}$ and $P_{s_i}\widehat{f}_{r_i}\otimes \widehat{g}_{r_j}$, for arbitrary $r_i$, $r_j\in\R$, are all $L^2$-bounded and trace class.

First notice $V_{r_i,r_j}(x,y)=V_{0,r_j-r_i}(-x,-y)$. Using the shorthand $r=r_j-r_i$ and inserting this into the integral representation \eqref{contIntK} of $V$, we have
\begin{equation}
 V_{r_i,r_j}(x,y)=e^{\frac{2}{3}r^3}\int_\R\mathrm{d} \lambda\, \Ai(-x+\lambda)e^{r(-y+\lambda)}\Ai(r^2-y+\lambda)=\left(V^{(1)}V^{(2)}_r\right)(x,y),
\end{equation}
with the new operators
\begin{equation}\begin{aligned}
 V^{(1)}(x,y)&=\Ai(-x+y)\\
 V^{(2)}_r(x,y)&=e^{\frac{2}{3}r^3}e^{r(x-y)}\Ai(r^2+x-y).
\end{aligned}\end{equation}
Introducing yet another operator, $(Nf)(x)=\exp\left(-(\chi_i+\chi_j)x/2\right)f(x)$, we can write
\begin{equation}
 P_{s_i}\widehat{V}_{r_i,r_j}=(P_{s_i}M_{r_i}V^{(1)}N^{-1})(NV^{(2)}_rM_{r_j}^{-1}).
\end{equation}
The Hilbert-Schmidt norm of the first factor is given by
\begin{equation}\begin{aligned}
	\int_{\R^2}\mathrm{d} x\,\mathrm{d} y\, &\left|(P_{s_i}M_{r_i}V^{(1)}N^{-1})(x,y)\right|^2\\&=\int_{s_1}^{\infty}\mathrm{d} x\,\int_\R\mathrm{d} y\,m^2_{r_i}(x)\Ai^2(-x+y)e^{(\chi_i+\chi_j)y}\\
	&=\int_{s_1}^{\infty}\mathrm{d} x\,m^2_{r_i}(x)e^{(\chi_i+\chi_j)x}\int_\R\mathrm{d} z\,\Ai^2(z)e^{(\chi_i+\chi_j)z}.
\end{aligned}\end{equation}
The asymptotic behaviour of the Airy function and  \mbox{$\chi_i>\chi_j>0$} imply that both integrals are finite. Similarly,
\begin{equation}\begin{aligned}
	\int_{\R^2}&\mathrm{d} x\,\mathrm{d} y\, \left|(NV^{(2)}_rM_{r_j}^{-1})(x,y)\right|^2\\
	&=e^{\frac{4}{3}r^3}\int_{\R^2}\mathrm{d} x\,\mathrm{d} y\, e^{-(\chi_i+\chi_j)x}e^{2r(x-y)}\Ai^2(r^2+x-y)m^{-2}_{r_j}(y)\\
	&=e^{\frac{4}{3}r^3}\int_\R\mathrm{d} z\,e^{-(\chi_i+\chi_j)z}e^{2rz}\Ai^2(r^2+z)\int_\R\mathrm{d} y\,m^{-2}_{r_j}(y)e^{-(\chi_i+\chi_j)y}<\infty,
\end{aligned}\end{equation}
where we used $2r>\chi_i+\chi_j$ as well.
As a product of two Hilbert-Schmidt operators, $P_{s_i}\widehat{V}_{r_i,r_j}$ is thus $L^2$-bounded and trace class.

We decompose the operator $\widehat{K}_{r_i,r_j}$ as
\begin{equation}
 P_{s_i}\widehat{K}_{r_i,r_j}=(P_{s_i}M_{r_i}K'_{-r_i}P_0)(P_0K'_{r_j}M_{r_j}^{-1})
\end{equation}
where
\begin{equation}
 K'_r(x,y)=e^{\frac{2}{3}r^3}e^{r(x+y)}\Ai(r^2+x+y).
\end{equation}
Again, we bound the Hilbert-Schmidt norms,
\begin{equation}\begin{aligned}
	\int_{\R^2}\mathrm{d} x\,\mathrm{d} y\, &\left|(P_{s_i}M_{r_i}K'_{-r_i}P_0)(x,y)\right|^2\\
	&=e^{-\frac{4}{3}r_i^3}\int_{s_i}^\infty\mathrm{d} x\,\int_0^\infty\mathrm{d} y\,m^2_{r_i}(x)e^{-2r_i(x+y)}\Ai^2(r_i^2+x+y)\\
	&\leq e^{-\frac{4}{3}r_i^3}\int_{s_i}^\infty\mathrm{d} x\,m^2_{r_i}(x)\int_{s_i}^\infty\mathrm{d} z\,e^{-2r_iz}\Ai^2(r_i^2+z)<\infty,
\end{aligned}\end{equation}
as well as
\begin{equation}\begin{aligned}
	\int_{\R^2}\mathrm{d} x\,\mathrm{d} y\, &\left|(P_0K'_{r_j}M_{r_j}^{-1})(x,y)\right|^2\\
	&=e^{\frac{4}{3}r_j^3}\int_0^\infty\mathrm{d} x\,\int_\R\mathrm{d} y\,e^{2r_j(x+y)}\Ai^2(r_j^2+x+y)m^{-2}_{r_j}(y)\\
	&= e^{\frac{4}{3}r_j^3}\int_\R\mathrm{d} y\,m^{-2}_{r_j}(y)\int_y^\infty\mathrm{d} z\,e^{2r_jz}\Ai^2(r_j^2+z).
\end{aligned}\end{equation}
The superexponential decay of the Airy function implies that for every \mbox{$c_1>|r_j|$} we can find $c_2$ such that $e^{2r_jz}\Ai^2(r_j^2+z)\leq c_2e^{-c_1z}$. This proves finiteness of the integrals.

Regarding the last operator, start by decomposing it as
\begin{equation}
 P_{s_i}\widehat{f}_{r_i}\otimes \widehat{g}_{r_j}=(P_{s_i}\widehat{f}_{r_i}\otimes\phi)(\phi \otimes \widehat{g}_{r_j})
\end{equation}
for some function $\phi$ with $L^2$-norm $1$. Next, notice that
\begin{equation}\begin{aligned}
	\int_{\R^2}\mathrm{d} x\,\mathrm{d} y\, &\left|(P_{s_i}M_{r_i}f_{r_i}\otimes\phi)(x,y)\right|^2=\int_{s_i}^\infty\mathrm{d} x\,m^2_{r_i}(x)f^2_{r_i}(x)
\end{aligned}\end{equation}
It is easy to see that $\lim_{s\to\infty}f_{r_i}(s)=1$, so $f_{r_i}$ is bounded on the area of integration. But then the $m^2_{r_i}$ term ensures the decay, implying that the integral is finite. Furthermore,
\begin{equation}\begin{aligned}
	\int_{\R^2}\mathrm{d} x\,\mathrm{d} y\, &\left|(\phi \otimes g_{r_j}M_{r_j}^{-1})(x,y)\right|^2=\int_\R\mathrm{d} y\,m^{-2}_{r_j}(y)g^2_{r_j}(y).
\end{aligned}\end{equation}
Analyzing the asymptotic behaviour of $g_{r_j}$ we see that for large positive arguments, the first part decays exponentially with rate $-\delta$ and the second part even superexponentially. $\delta>\chi_j$ thus gives convergence on the positive half-line. For negative arguments, it is sufficient to see that $g_{r_j}$ does not grow faster than exponentially.
\end{proof}

\section{Analytic continuation}\label{SectAnCont}
We know from Theorem~\ref{thmAsymp} and Proposition~\ref{propPathInt} that:
\begin{equation}
\lim_{{t}\to\infty}\Pb\bigg(\bigcap_{k=1}^m\big\{X_{t}^{(\delta)}(r_k)\leq s_k\big\}\bigg)=\bigg(1+\frac{1}{\delta}\sum_{i=1}^m\frac{\D}{\mathrm{d} s_i}\bigg)\det(\Id-\widehat{\mathcal{P}}\widehat{K}^\delta_{r_1}).
\end{equation}
In this section we prove the main Theorem~\ref{thmAsymp0} by extending this equation to $\delta=0$. The right hand side can actually be analytically continued for all $\delta\in\R$ (see Proposition~\ref{propAnalyt}). Additionally we have to show that the left hand side is continuous at $\delta=0$. This proof relies mainly on Proposition~\ref{propExPoint}, which gives a bound on the exit point of the maximizing path from the lower boundary in the last passage percolation model.

\begin{proof}[Proof of Theorem~\ref{thmAsymp0}]
We adopt the point of view of last passage percolation discussed in the Sections~\ref{secStepOther} and~\ref{secInfLPP}. By the stationarity property, we know that $x_n(t)\stackrel{d}{=}\widetilde{x}^{(0)}_n(t)$ for $n\geq0$. We use the latter interpretation, i.e.\ consider $x_n(t)$ as being constructed from $\zeta_n$ with $n\geq1$, $B_n$ with $n\geq1$ and $\widetilde{B}_0$, and base coupling arguments also on these variables being fixed. In this way, we have
\begin{equation}
 x_n(t)=L_{(0,0)\to(t,n)},
\end{equation}
with background weights on both boundaries, i.e.\ Dirac weights $\zeta_k-\zeta_{k-1}$ on $(0,k)$, $k\geq1$, and a Lebesgue measure of density $\rho$ on the line $\{0\}\times\R_+$ additionally to the white noise $\mathrm{d} \widetilde{B}_0$.

We add superscripts to $x$, $L$ and $w$ indicating the choice of $\rho$, while $\lambda$ is always fixed at $1$. It is clear that for any path $\vec{\pi}$ the weight $w^{(\rho)}(\vec{\pi})$ is non-decreasing in $\rho$. But then the supremum is non-decreasing, too, and:
\begin{equation}\label{eq5.6a}
 x_n^{(\rho)}(t)\leq x_n^{(1)}(t),
\end{equation}
for $\rho<1$.
We know there exists a unique maximizing path \mbox{$\vec{\pi}^*\in\Pi(0,0;t;n)$}. We can therefore define $Z_n(t):=s_0^*$, the exit point from the lower boundary specifically with $\rho=1$. We want to derive the inequality
\begin{equation}\label{eq5.7}
   x_n^{(1)}(t)\leq x_n^{(\rho)}(t)+(1-\rho)Z_n(t).
\end{equation}
It can be seen as follows:
\begin{equation}\begin{aligned}
 L^{(1)}_{(0,0)\to(t,n)}-(1-\rho)Z_n(t)&=\sup_{\vec{\pi}\in\Pi(0,0;t,n)}w^{(1)}(\vec{\pi})-(1-\rho)Z_n(t)\\
 &=w^{(1)}(\vec{\pi}^*)-(1-\rho)s_0^*=w^{(\rho)}(\vec{\pi}^*).
\end{aligned}\end{equation}
Note that $\vec{\pi}^*$ maximizes $w^{(1)}(\vec\pi)$ and not necessarily $w^{(\rho)}(\vec{\pi})$. In particular we have
\begin{equation}
 w^{(\rho)}(\vec{\pi}^*)\leq\sup_{\vec{\pi}\in\Pi(0,0;t,n)}w^{(\rho)}(\vec{\pi})= L^{(\rho)}_{(0,0)\to(t,n)}.
\end{equation}
Combining the last two equations results in \eqref{eq5.7}.

\eqref{eq5.6a} and \eqref{eq5.7} imply that for the rescaled processes $X_{t}^{(\delta)}$, see (\ref{eqScaledProcess}), we have
\begin{equation}
 X^{(\delta)}_{t}(r)\leq X^\textrm{stat}_{t}(r)\leq X^{(\delta)}_{t}(r)+\delta{t}^{-2/3}Z_{{t}+2{t}^{2/3}r}({t}).
\end{equation}
For any $\e>0$ it holds
\begin{equation}\begin{aligned}
  \Pb &\bigg(\bigcap_{k=1}^m\{X^{(\delta)}_{t}(r_k)\leq s_k\}\bigg)\geq\Pb\bigg(\bigcap_{k=1}^m\{X^\textrm{stat}_{t}(r_k)\leq s_k\}\bigg)\\
  &\geq\Pb\bigg(\bigcap_{k=1}^m\{X^{(\delta)}_{t}(r_k)+\delta{t}^{-2/3}Z_{{t}+2{t}^{2/3}r}({t})\leq s_k\}\bigg)\\
  &\geq\Pb\bigg(\bigcap_{k=1}^m\{X^{(\delta)}_{t}(r_k)\leq s_k-\e\}\bigg)-\sum_{k=1}^m\Pb\left(\delta{t}^{-2/3}Z_{{t}+2{t}^{2/3}r}({t})>\e\right).
\end{aligned}\end{equation}
Then, taking ${t}\to\infty$, we obtain
\begin{equation}\begin{aligned}
  \Pb &\bigg(\bigcap_{k=1}^m\{\mathcal{A}_\mathrm{stat}^{(\delta)}(r_k)\leq s_k\}\bigg)
  \geq\limsup_{{t}\to\infty}\Pb\bigg(\bigcap_{k=1}^m\{X^\textrm{stat}_{t}(r_k)\leq s_k\}\bigg)\\
  &\geq\liminf_{{t}\to\infty}\Pb\bigg(\bigcap_{k=1}^m\{X^\textrm{stat}_{t}(r_k)\leq s_k\}\bigg)\\
  &\geq\Pb\bigg(\bigcap_{k=1}^m\{\mathcal{A}_\mathrm{stat}^{(\delta)}(r_k)\leq s_k-\e\}\bigg) -\sum_{k=1}^m\limsup_{{t}\to\infty}\Pb\left(Z_{{t}+2{t}^{2/3}r}({t})>{t}^{2/3}\e/\delta\right).
\end{aligned}\end{equation}
Using Proposition~\ref{propExPoint} on the last term and Proposition~\ref{propAnalyt} on the other terms, we can now take the limit $\delta\downarrow0$, resulting in
\begin{equation}\begin{aligned}
  \Pb &\bigg(\bigcap_{k=1}^m\{\mathcal{A}_\mathrm{stat}(r_k)\leq s_k\}\bigg)
  \geq\limsup_{{t}\to\infty}\Pb\bigg(\bigcap_{k=1}^m\{X^\textrm{stat}_{t}(r_k)\leq s_k\}\bigg)\\
  &\geq\liminf_{{t}\to\infty}\Pb\bigg(\bigcap_{k=1}^m\{X^\textrm{stat}_{t}(r_k)\leq s_k\}\bigg)
  \geq\Pb\bigg(\bigcap_{k=1}^m\{\mathcal{A}_\mathrm{stat}(r_k)\leq s_k-\e\}\bigg).
\end{aligned}\end{equation}
Continuity of \eqref{eqAiryDef} in the $s_k$ finishes the proof.
\end{proof}

\begin{prop}\label{propExPoint}
For any $r\in\R$,
 \begin{equation}\label{eq5.14}
  \lim_{\beta\to\infty}\limsup_{{t}\to\infty}\Pb\left(Z_{{t}+2{t}^{2/3}r}({t})>\beta{t}^{2/3}\right)=0.
 \end{equation}
\end{prop}
\begin{proof}
 By scaling of ${t}$ and $\beta$, \eqref{eq5.14} is equivalent to
  \begin{equation}\label{eq5.15}
  \lim_{\beta\to\infty}\limsup_{{t}\to\infty}\Pb\left(Z_{{t}}({t}+2{t}^{2/3}r)>\beta{t}^{2/3}\right)=0,
 \end{equation}
for any $r\in\R$, which is the limit we are showing. We introduce some new events:
\begin{equation}
 \begin{aligned}
  M_\beta&:=\{Z_{{t}}({t}+2{t}^{2/3}r)>\beta{t}^{2/3}\}\\
  E_\beta&:=\{L_{(0,0)\to(\beta{t}^{2/3},0)}+L_{(\beta{t}^{2/3},0)\to({t}+2{t}^{2/3}r,{t})}\leq2{t}+2{t}^{2/3}r+s{t}^{1/3}\}\\
  N_\beta&:=\{L_{(0,0)\to({t}+2{t}^{2/3}r,{t})}\leq2{t}+2{t}^{2/3}r+{t}^{1/3}s\},
 \end{aligned}
\end{equation}
where $L$ is to be understood as in the proof of Theorem~\ref{thmAsymp0}, with $\rho=1$.
Notice that if $M_\beta$ occurs, then
\begin{equation}
L_{(0,0)\to({t}+2{t}^{2/3}r,{t})}=L_{(0,0)\to(\beta{t}^{2/3},0)}+L_{(\beta{t}^{2/3},0)\to({t}+2{t}^{2/3}r,{t})},
\end{equation} resulting in $M_\beta\cap E_\beta\subseteq N_\beta$. We arrive at the inequality:
\begin{equation}\label{eqMbeta}
 \Pb(M_\beta)=\Pb(M_\beta\cap E_\beta)+\Pb(M_\beta\cap E_\beta^c)\leq\Pb(N_\beta)+\Pb(E_\beta^c).
\end{equation}
We further define new random variables
\begin{equation}\label{eq5.18}\begin{aligned}
 \xi_\textrm{spiked}^{({t})}& =\frac{L_{(\beta{t}^{2/3},0)\to({t}+2{t}^{2/3}r,{t})}-2{t}-2{t}^{2/3}(r-\beta/2)}{{t}^{1/3}}+(r-\beta/2)^2,\\
 \xi_\textrm{GUE}^{({t})}&=\frac{L^\textrm{packed}_{(0,1)\to({t}+2{t}^{2/3}r,{t})}-2{t}-2{t}^{2/3}r}{{t}^{1/3}}+r^2,\\
 \xi_\textrm{N}^{({t})}&=\frac{L_{(0,0)\to(\beta{t}^{2/3},0)}-\beta{t}^{2/3}}{\sqrt{\beta}{t}^{1/3}},
\end{aligned}\end{equation}
where $L^\textrm{packed}$ is to be understood as the last passage percolation time \emph{without} boundary weights.
By Theorem~\ref{thmTagged}, for any fixed $r\in\R$,
\begin{equation}\label{eqGUE}
 \xi_\textrm{GUE}^{({t})}\stackrel{d}{\to} \xi_\textrm{GUE},
\end{equation}
where $\xi_\textrm{GUE}$ has the GUE Tracy-Widom distribution. $\xi_\textrm{spiked}^{({t})}$ follows the distribution of the largest eigenvalue of a critically 	spiked GUE matrix, as will be shown in Lemma~\ref{lemSpike}. $\xi_\textrm{N}^{({t})}$ has the distribution of a standard normal random variable $\xi_\textrm{N}$ for any $\beta>0$, ${t}>0$.

Combining these definitions, we have:
\begin{equation}
 \begin{aligned}
  \Pb(E_\beta)=\Pb\big(\sqrt{\beta}\xi_\textrm{N}^{({t})}+\xi_\textrm{spiked}^{({t})}\leq  (r-\beta/2)^2+s\big).
 \end{aligned}
\end{equation}
Fix $s=3r^2-\beta^2/16$, such that:
\begin{equation}\begin{aligned}
  \left(r-\frac{\beta}{2}\right)^2+s&=4r^2-r\beta+\frac{\beta^2}{16}+\frac{\beta^2}{8}\geq \frac{\beta^2}{8}.
 \end{aligned}\end{equation}
Using the independence of $\xi_\textrm{N}^{({t})}$ and $\xi_\textrm{spiked}^{({t})}$, we obtain
\begin{equation}\label{eqEbeta}\begin{aligned}
  \Pb(E_\beta)&\geq\Pb\left(\sqrt{\beta}\xi_\textrm{N}^{({t})}+\xi_\textrm{spiked}^{({t})}\leq  \frac{\beta^2}{16}+\frac{\beta^2}{16}\right)\\&\geq\Pb\left(\xi_\textrm{N}^{({t})}\leq  \frac{\beta^{3/2}}{16}\text{ and }\xi_\textrm{spiked}^{({t})}\leq\frac{\beta^2}{16}\right)\\
  &=\Pb\left(\xi_\textrm{N}^{({t})}\leq  \frac{\beta^{3/2}}{16}\right)\Pb\left(\xi_\textrm{spiked}^{({t})}\leq\frac{\beta^2}{16}\right)
 \end{aligned}\end{equation}
Further, the inequality
\begin{equation}
 L^\textrm{packed}_{(0,1)\to({t}+2{t}^{2/3}r,{t})}\leq L_{(0,0)\to({t}+2{t}^{2/3}r,{t})}
\end{equation}
leads to
\begin{equation}\label{eqNbeta}
 \Pb(N_\beta)\leq\Pb\big(\xi_\textrm{GUE}^{({t})}\leq 4r^2-\beta^2/16\big).
\end{equation}
Inserting \eqref{eqEbeta} and \eqref{eqNbeta} into \eqref{eqMbeta}, we arrive at
\begin{equation}\begin{aligned}
 \Pb(M_\beta)\leq\Pb\left(\xi_\textrm{GUE}^{({t})}\leq 4r^2-\frac{\beta^2}{16}\right)+1-\Pb\left(\xi_\textrm{N}^{({t})}\leq  \frac{\beta^{3/2}}{16}\right)\Pb\left(\xi_\textrm{spiked}^{({t})}\leq\frac{\beta^2}{16}\right)
\end{aligned}\end{equation}
By \eqref{eqGUE} and Lemma~\ref{lemSpike} we can take limits:
\begin{equation}\begin{aligned}
0&\leq\limsup_{\beta\to\infty}\limsup_{{t}\to\infty}\Pb\left(M_\beta\right)\\
&\leq\lim_{\beta\to\infty}\bigg[\Pb\left(\xi_\textrm{GUE}\leq 4r^2-\frac{\beta^2}{16}\right)\\
&\qquad\qquad+1-\Pb\left(\xi_\textrm{N}\leq  \frac{\beta^{3/2}}{16}\right)\Pb\left(\xi_\textrm{spiked}(\beta)\leq\frac{\beta^2}{16}\right)\bigg]\\&=0.
\end{aligned}\end{equation}
\end{proof}

\begin{lem}\label{lemSpike}
Let $r\in\R$ be fixed. For any $\beta>2(r+1)$, as ${t}\to\infty$, the random variable
 \begin{equation}
   \xi_\textrm{spiked}^{({t})}=\frac{L_{(\beta{t}^{2/3},0)\to({t}+2{t}^{2/3}r,{t})}-2{t}-2{t}^{2/3}(r-\beta/2)}{{t}^{1/3}}+(r-\beta/2)^2
 \end{equation}
converges in distribution,
\begin{equation}
 \xi_\textrm{spiked}^{({t})}\stackrel{d}{\to} \xi_\textrm{spiked}(\beta).
\end{equation}
In addition, $\xi_\textrm{spiked}(\beta)$ satisfies
\begin{equation}
 \lim_{\beta\to\infty}\Pb\left(\xi_\textrm{spiked}(\beta)\leq\beta^2/16\right)=1.
\end{equation}

\end{lem}
\begin{proof}
 The family of processes $L_{(\beta{t}^{2/3},0)\to(\beta{t}^{2/3}+t,n)}$ indexed by $n\in\Z_{\geq0}$ and time parameter $t\geq0$ is precisely a marginal of Warren's process with drifts, starting at zero, as defined in~\cite{FF13}. In our case only the first particle has a drift of $1$, and all the others zero. By Theorem 2~\cite{FF13}, the fixed time distribution of this process is given by the distribution of the largest eigenvalue of a spiked $n\times n$ GUE matrix, where the spikes are given by the drifts.

 Thus we can apply the results on spiked random matrices, more concretely we want to apply Theorem 1.1~\cite{BW10}, with the potential given by \mbox{$V(x)=-x^2/2$}. Since
 \begin{equation}
  L^*:=L_{(\beta{t}^{2/3},0)\to({t}+2{t}^{2/3}r,n)}
 \end{equation}
 represents a $n\times n$ GUE matrix diffusion $M(t)$ at time $t={t}+2{t}^{2/3}(r-\beta/2)$, it is distributed according to the density
 \begin{equation}
  p_n(M)=\frac{1}{Z_n}\exp\left(-\frac{\Tr(M-t{\textrm I}_{11})^2}{2t}\right),
 \end{equation}
 where ${\textrm I}_{11}$ is a $n\times n$ matrix with a one at entry $(1,1)$ and zeros elsewhere. In order to apply the theorem we need the density given in equation (1)~\cite{BW10}, i.e., consider the scaled quantity $L^*/\sqrt{nt}$. The size of the first-order spike is then:
 \begin{equation}
  a= t/\sqrt{nt} = \sqrt{1+2{t}^{-1/3}(r-\beta/2)} = 1+(r-\beta/2){t}^{-1/3}+\Or({t}^{-2/3}).
 \end{equation}
We are thus in the neighbourhood of the critical value ${\mathbf a}_c=1$. For $\alpha\geq0$, let
\begin{equation}
 C_\alpha(\xi)=\int_{-\infty}^0e^{\alpha x}\Ai(x+\xi) \mathrm{d} x.
\end{equation}
With $F_\textrm{GUE}(s)$ being the cumulative distribution function of the GUE Tracy-Widom distribution, and $K_{0,0}(s_1,s_2)$ as in \eqref{eqKernel2Def}, define:
\begin{equation}\label{eqF1}
 F_1(s;\alpha)=F_\textrm{GUE}(s)\Big(1-\left\langle(1-P_s K_{0,0}P_s)^{-1}C_\alpha,P_s\Ai\right\rangle\Big).
\end{equation}

Applying (28)~\cite{BW10}, we have
 \begin{equation}
  n^{2/3}(L^*/\sqrt{nt}-2)\to \xi_\textrm{spiked}(\beta),
 \end{equation}
with
\begin{equation}
 \Pb\left(\xi_\textrm{spiked}(\beta)\leq\beta^2/16\right)=F_1(\beta^2/16,\alpha),
\end{equation}
where $\alpha=\beta/2-r$. Since in our case $\alpha>1$, we can estimate:
\begin{equation}
 \left|C_\alpha(\xi)\right|\leq\int_{-\infty}^0e^{\alpha x}e^{-x-\xi} \mathrm{d} x=e^{-\xi}\frac{1}{\alpha-1}.
\end{equation}
Combining this with the usual bounds on the Airy kernel and the Airy function, we see that as $\beta\to\infty$, the scalar product in \eqref{eqF1} converges to zero and we are left with the limit of $F_0$ which is one.

On the other hand,
 \begin{equation}\begin{aligned}
  n^{2/3}(L^*/\sqrt{nt}-2)\leq s\quad\Leftrightarrow\quad L^*\leq\sqrt{nt}(2+n^{-2/3}s),
 \end{aligned}\end{equation}
  and
  \begin{equation}
   \sqrt{nt}(2+n^{-2/3}s)=2{t}+2{t}^{2/3}(r-\beta/2)+{t}^{1/3}\left(s-(r-\beta/2)^2\right)+\Or(1),
  \end{equation}
from which the claim follows.
\end{proof}

\begin{prop}\label{propAnalyt}
 The function $\delta\mapsto\delta^{-1}\det(\Id-\widehat{\mathcal{P}}\widehat{K}^\delta_{r_1})$ can be extended analytically in the domain $\delta\in\R$. Its value at $\delta=0$ is given by
 \begin{equation}
  G_m(\vec{r},\vec{s})\det\left(\Id-\mathcal{P}K\right)_{L^2(\R)}.
 \end{equation}
\end{prop}

\begin{proof}
We use the identity $\det(\Id+A)\det(\Id+B)=\det(\Id+A+B+AB)$ and Lemma~\ref{lemInv}  to factorize
\begin{equation}\label{eqFac}\begin{aligned}
	\delta^{-1}\det\big(\Id-\widehat{\mathcal{P}}\widehat{K}^\delta_{r_1}\big)&=\delta^{-1}\det(\Id-\widehat{\mathcal{P}}\widehat{K}^\delta_{r_1})
=\delta^{-1}\det\big(\Id-\widehat{\mathcal{P}}\widehat{K}-\delta\widehat{\mathcal{P}}\widehat{f}_{r_1}\otimes \widehat{g}_{r_1}\big)\\
	&=\delta^{-1}\det\big(\Id-\delta(\Id-\widehat{\mathcal{P}}\widehat{K})^{-1}\widehat{\mathcal{P}}\widehat{f}_{r_1}\otimes \widehat{g}_{r_1}\big)\cdot\det\big(\Id-\widehat{\mathcal{P}}\widehat{K}\big)\\
	&=\big(\delta^{-1}-\big\langle (\Id-\widehat{\mathcal{P}}\widehat{K})^{-1}\widehat{\mathcal{P}}\widehat{f}_{r_1},\widehat{g}_{r_1}\big\rangle\big)
\cdot\det\big(\Id-\widehat{\mathcal{P}}\widehat{K}\big)\\
	&=\big(\delta^{-1}-\big\langle (\Id-\mathcal{P}K)^{-1}\mathcal{P}f_{r_1},g_{r_1}\big\rangle\big)\cdot\det\big(\Id-\mathcal{P}K\big).
\end{aligned}\end{equation}
Since the second factor is independent of $\delta$, the remaining task is the analytic continuation of the first. Using \eqref{contIntfg}, decompose $f_{r_1}$  as
\begin{equation}\label{eqfDec}
	f_{r_1}(s)=1+\frac{1}{2\pi\mathrm{i}}\int_{\rangle0} \mathrm{d} W\frac{e^{-W^3/3-r_1W^2+sW}}{W}=1+f^*(s).
\end{equation}
Now,
\begin{equation}\label{eq49}\begin{aligned}
		\langle P_{s_1}\mathbf{1},g_{r_1}\rangle&=\int_{s_1}^\infty \mathrm{d} s\,\frac{1}{2\pi\mathrm{i}}\int_{0\langle\delta} \mathrm{d} Z\frac{e^{Z^3/3+r_1Z^2-sZ}}{Z-\delta}\\
&=\frac{1}{2\pi\mathrm{i}}\int_{0\langle\delta} \mathrm{d} Z\frac{e^{Z^3/3+r_1Z^2-s_1Z}}{Z(Z-\delta)}\\&=\frac{1}{\delta}+\frac{1}{2\pi\mathrm{i}}\int_{\langle0,\delta} \mathrm{d} Z\frac{e^{Z^3/3+r_1Z^2-s_1Z}}{Z(Z-\delta)}=\frac{1}{\delta}-\mathcal{R}_\delta.
\end{aligned}\end{equation}
The function $\mathcal{R}_\delta$ is analytic in $\delta\in\R$. Using these two identities as well as $(\Id-\mathcal{P}K)^{-1}=\Id+(\Id-\mathcal{P}K)^{-1}\mathcal{P}K$, we can rearrange the inner product as follows:
\begin{equation}\label{eq51}\begin{aligned}
	\frac{1}{\delta}-&\left\langle (\Id-\mathcal{P}K)^{-1}\mathcal{P}f_{r_1},g_{r_1}\right\rangle=\frac{1}{\delta}-\left\langle(\Id-\mathcal{P}K)^{-1}\mathcal{P}\mathbf{1}+(\Id-\mathcal{P}K)^{-1}\mathcal{P}f^*,g_{r_1}\right\rangle
	\\&=\frac{1}{\delta}-\left\langle\mathcal{P}\mathbf{1}+(\Id-\mathcal{P}K)^{-1}(\mathcal{P}K\mathcal{P}\mathbf{1}+\mathcal{P}f^*),g_{r_1}\right\rangle
\\&=\frac{1}{\delta}-\left\langle P_{s_1}\mathbf{1},g_{r_1}\right\rangle-\left\langle(\mathcal{P}-P_{s_1})\mathbf{1}+(\Id-\mathcal{P}K)^{-1}(\mathcal{P}K\mathcal{P}\mathbf{1}+\mathcal{P}f^*),g_{r_1}\right\rangle\\
&=\mathcal{R}_\delta-\left\langle(\Id-\mathcal{P}K)^{-1}\left(\mathcal{P}f^*+\mathcal{P}KP_{s_1}\mathbf{1}+(\mathcal{P}-P_{s_1})\mathbf{1}\right),g_{r_1}\right\rangle
\end{aligned}\end{equation}
Since $g_{r_1}$ is evidently analytic in $\delta\in\R$, we are left to show convergence of the scalar product.

All involved functions are locally bounded, so to establish convergence it is enough to investigate their asymptotic behaviour. $g_{r_1}$ may grow exponentially at arbitrary high rate, depending on $r_1$ and $\delta$, for both large positive and large negative arguments. We therefore need superexponential bounds on the function:
\begin{equation}
 (\Id-\mathcal{P}K)^{-1}\left(\mathcal{P}f^*+\mathcal{P}KP_{s_1}\mathbf{1}+(\mathcal{P}-P_{s_1})\mathbf{1}\right).
\end{equation}
For this purpose we first need an expansion of the operator $\mathcal{P}$:
\begin{equation}\label{eqPExp}
 \mathcal{P}=\sum_{k=1}^n\bar{P}_{s_1}V_{r_1,r_2}\dots \bar{P}_{s_{k-1}}V_{r_{k-1},r_k}P_{s_k}V_{r_k,r_1}.\\
\end{equation}
Notice that all operators $P_{s_i}$, $\bar{P}_{s_i}$ and $V_{r_i,r_j}$ map superexponentially decaying functions onto superexponentially decaying functions. Moreover $P_{s_i}$ and $\bar{P}_{s_i}$ generate superexponential decay for large negative resp. positive arguments.

The function $f^*$ decays superexponentially for large arguments but may grow exponentially for small ones. Since every part of the sum contains one projection $P_{s_k}$, $\mathcal{P}f^*$ decays superexponentially on both sides.

Examining $(\mathcal{P}-P_{s_1})\mathbf{1}$, notice that the $k=1$ contribution in \eqref{eqPExp} is equal to $P_{s_1}$, which is cancelled out here. All other contributions contain both $\bar{P}_{s_1}$ and $P_{s_k}$, which ensure superexponential decay.

Using the usual asymptotic bound on the Airy function, we see that the operator $K$ maps any function in its domain onto one which is decreasing superexponentially for large arguments. By previous arguments, functions in the image of $\mathcal{P}K$ decay on both sides, in particular $\mathcal{P}KP_{s_1}\mathbf{1}$.

Now, in order to establish the finiteness of the scalar product, decompose the inverse operator as $(\Id-\mathcal{P}K)^{-1}=\Id+\mathcal{P}K(\Id-\mathcal{P}K)^{-1}$. The contribution coming from the identity has just been settled. As inverse of a bounded operator, $(\Id-\mathcal{P}K)^{-1}$ is also bounded. Because of the rapid decay, the functions $\mathcal{P}f^*$, $\mathcal{P}KP_{s_1}\mathbf{1}$ and $(\mathcal{P}-P_{s_1})\mathbf{1}$ are certainly in $L^2(\R)$ and thus mapped onto $L^2(\R)$ by this operator. Finally, the image of an $L^2(\R)$-function under the operator $\mathcal{P}K$ is decaying superexponentially on both sides.

The expression \eqref{eq51} is thus an analytic function in $\delta$ in the domain $\R$. Setting $\delta=0$ returns the value of $G_m(\vec{r},\vec{s})$. Combining these results with \eqref{eqFac} finishes the proposition.
\end{proof}
\chapter{More general initial conditions and their asymptotics}\label{secMixed}
Besides the three basic initial conditions analyzed in the previous two chapters, a further class of natural initial measures is obtained
by first restricting the basic measures to a half-space and then joining pairwise. They all satisfy an asymptotic theorem, where the limit process is given by a crossover Airy process, which interpolates between two of the processes Airy$_2$, Airy$_1$ and Airy$_\textrm{stat}$. All the results can be further generalized to \emph{bounded modifications} of the discussed initial conditions, as well as to asymptotics along space-like paths instead of just at a fixed time.
\section{Half-Periodic initial conditions}\label{secHF}
Let $\vec{\zeta}^\shf\in\R^{\Z_{>0}}$ be the vector defined by $\zeta^\hf_n=n$. The determinantal structure for this case has already been obtained as a byproduct of the full periodic case. We can thus directly state the kernel, in its alternative form as mentioned in Remark~4.3~\cite{FSW13}. Notice again that the direction of space is reversed in our convention.
\begin{prop}[Proposition 4.2 of~\cite{FSW13}]\label{propHfKernel}
Let $\{x_n(t),n\geq1\}$ be the system of one-sided reflected Brownian motions satisfying the initial condition \mbox{$\vec{x}(0)=\vec{\zeta}^\shf$}. Then for any finite subset $S$ of $\Z_{>0}$, it holds
\begin{equation}%\label{eq33}
\Pb\bigg(\bigcap_{n\in S} \{x_n(t)\leq a_n\}\bigg)=\det(\Id-\chi_a \mathcal{K}_\hf\chi_a)_{L^2(S\times \R)},
\end{equation}
where $\chi_a(n,\xi)=\Id_{\xi>a_n}$. The kernel $\mathcal{K}_\hf$ is given by
\begin{equation}\label{eqKtHf}
\mathcal{K}_\hf(n_1,\xi_1;n_2,\xi_2)=-\phi_{n_1,n_2}(\xi_1,\xi_2)\Id_{n_2>n_1}+\mathcal{K}_1(n_1,\xi_1;n_2,\xi_2),
\end{equation}
with
\begin{equation}\label{eqK1def}
\begin{aligned}	\phi_{n_1,n_2}(\xi_1,\xi_2)&=\frac{(\xi_2-\xi_1)^{n_2-n_1-1}}{(n_2-n_1-1)!}\Id_{\xi_1\leq \xi_2}\\
\mathcal{K}_1(n_1,\xi_1;n_2,\xi_2)&=\frac{1}{(2\pi\mathrm{i})^2}\int_{\mathrm{i}\R-1}\mathrm{d} w\,\oint_{\Gamma_0}\mathrm{d} z\,\frac{e^{{t} w^2/2+\xi_1w}}{e^{{t} z^2/2+\xi_2 z}}\frac{(-w)^{n_1}}{(-z)^{n_2}}\frac{(1+z)e^z}{we^w-ze^z},
\end{aligned}
\end{equation}
where $\Gamma_0$ is chosen in such a way that $|ze^z|<|we^w|$ holds always.
\end{prop}
The law of large numbers now depends on the region we examine:
\begin{equation}
 \lim_{t\to\infty}\frac{x_{\alpha^2 t}(t)}{t}=\begin{cases}2\alpha, &\text{ for } 0\leq\alpha\leq1 \\1+\alpha^2,&\text{ for } \alpha>1.\end{cases}
\end{equation}
Notice that the first case is the same scaling as for the packed initial condition and the second case is the scaling for the full periodic case. This analogy carries over to the behaviour of the fluctuations around the macroscopic position. For $0<\alpha<1$ they are of order $t^{1/3}$ and given by the Airy$_2$ process, while for $\alpha>1$ they are of the same order and given by the Airy$_1$ process. These limits are not proven here, instead we focus on the more interesting transition regions. The transition at $\alpha=0$ is simple, as for any finite $n$, $x_n(t)$ is a bounded modification of the system with packed initial condition. By Lemma~\ref{propBoundMod} this bounded modification stays bounded and is therefore irrelevant on the scale $\sqrt{t}$. This implies that $\lim_{t\to\infty}x_n(t)/\sqrt{t}$ behaves as the top line of a $n$-particle Dyson's Brownian motion.

At the transition point $\alpha=1$ we find a new process that interpolates between $\mathcal{A}_2$ and $\mathcal{A}_1$ and is therefore called the Airy$_{2\to1}$ process, $\mathcal{A}_{2\to1}$:

\begin{thm}\label{thmAsympHf}
With $\{x_n(t),n\geq1\}$ being the system of one-sided reflected Brownian motions with initial condition $\vec{x}(0)=\vec{\zeta}^\shf$, define the rescaled process
\begin{equation}\label{eqHfScaledProcess}
	r\mapsto X_t^\hf(r) = t^{-1/3}\big(x_{\lfloor t+2rt^{2/3}\rfloor}(t)-2t-2rt^{2/3} \big).\,
\end{equation}
In the sense of finite-dimensional distributions,
\begin{equation}
	\lim_{t\to\infty}X_t^\hf(r)\stackrel{d}{=}\mathcal{A}_{2\to1}(r).
\end{equation}
\end{thm}

\begin{proof}
With
\begin{equation}\label{eqScalinghfl}\begin{aligned}
	n_i&={t}+2{t}^{2/3}r_i\\
	\xi_i&=2{t}+2{t}^{2/3}r_i+{t}^{1/3}s_i,
\end{aligned}\end{equation}
define a rescaled and conjugated kernel by
\begin{equation}
	\mathcal{K}_\hf^\textrm{resc}(r_1,s_1;r_2,s_2)={t}^{1/3}e^{\xi_1-\xi_2}\mathcal{K}_\hf(n_1,\xi_1;n_2,\xi_2).
\end{equation}
Once the Propositions~\ref{propK1Pointw} and~\ref{propK1Bound} are established, the result follows in the same way as in the proof of Theorem~\ref{thmAsymp}.
\end{proof}

\begin{prop}\label{propK1Pointw}
Consider any $r_1,r_2$ in a bounded set and fixed $L$. Then, the kernel converges as
\begin{equation}
	\lim_{{t}\to\infty}\mathcal{K}_\hf^\textrm{resc}(r_1,s_1;r_2,s_2)=K_{\mathcal{A}_{2\to1}}(r_1,s_1;r_2,s_2)
\end{equation}
uniformly for $(s_1,s_2)\in[-L,L]^2$.
\end{prop}

\begin{proof}

The proof is conceptually similar to the case of packed initial conditions. However, some new issues arise, mainly due to the \emph{double} contour integral. The convergence of $\phi$ is clear from previous cases so we jump directly to the main part of the kernel.
Defining functions as
\begin{equation}\begin{aligned}
	f_3(z)&=-(z^2-1)/2-2(z+1)-\ln(-z)\\
	f_2(z,r)&=-2r(z+1+\ln(-z))\\
	f_1(z,s)&=-s(z+1),
\end{aligned}\end{equation}
we can write $G(z,r,s)={t} f_3(z)+{t}^{2/3}f_2(z,r)+{t}^{1/3}f_1(z,s)$, leading to
\begin{equation}\label{eqK1G}
	\mathcal{K}_1^\textrm{resc}(r_1,s_1;r_2,s_2)=\frac{{t}^{1/3}}{(2\pi\mathrm{i})^2}\int_{\mathrm{i}\R-1}\mathrm{d} w\,\oint_{\Gamma_0}\mathrm{d} z\,e^{G(z,r_2,s_2)-G(w,r_1,s_1)}\frac{(1+z)e^z}{we^w-ze^z}.
\end{equation}

The contour of the variable $w$ is already a steep descent curve for the leading order function $-f_3(w)$:
\begin{equation}
 \frac{\D\Re\left(-f_3(-1+\mathrm{i} u)\right)}{\mathrm{d} u}=\Re\left(\frac{(\mathrm{i} u)^2}{1+\mathrm{i} u}\mathrm{i}\right)=\frac{-u^3}{1+u^2},
\end{equation}
which means that the real part is maximal at $w=-1$ and strictly decreasing when moving away from it, quadratically fast for large $|w|$.

The choice for the contour of $z$ is trickier. Not only do we need a steep descent curve that comes close to the critical point $z=-1$, but we also have to ensure that we are not crossing any poles by deforming it, i.e.\ respect the inequality $|ze^z|<|we^w|$. Let $0<\omega<\sqrt{2}$ be a parameter and $L_0(z)$ the principal branch of the \emph{Lambert W function} defined by the inverse of the function $D_L=\{a+\mathrm{i} b\in\C,a+b\cot(b)>0\text{ and } -\pi<b<\pi\}\to\C$, $z\mapsto ze^z$. Choose
\begin{equation}
 \gamma^\omega(u)=L_0\left(-(1-\omega^2/2)e^{-1+\mathrm{i} u}\right).
\end{equation}

\begin{figure}
\centering
\includegraphics{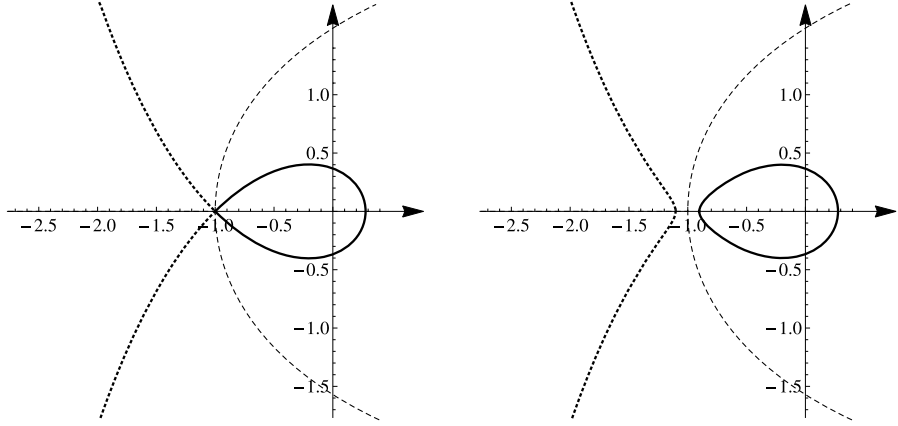}

\caption{The contours $\gamma^\omega$ (solid line) and $\widehat{\gamma}^\omega$ (dotted line)
 for $\omega=0$ (left picture) and
some small positive $\omega$ (right
picture). The dashed lines separate the ranges
of the principal branch $0$ (right) and the branches $1$ (top left)
and $-1$ (bottom left) of the Lambert W function.}%\vspace*{-3pt}
\label{figGamma}
\end{figure}

We have to show that this is actually a simple loop around the origin. For points in the domain $D_L$ of the principal branch of the Lambert function, we have
\begin{equation}
\sgn\Im \bigl((a+\mathrm{i} b)e^{a+\mathrm{i} b} \bigr)=\sgn (a\sin b +b\cos b )=\sgn b.
\end{equation}
So the function $z\mapsto ze^z$ preserves the sign of the imaginary part, and consequently its inverse does the same. $\gamma^\omega(u)$ is thus in the lower half-plane for $0<u<\pi$ and in the upper one for $\pi<u<2\pi$ and meets the real line in the two points $\gamma^\omega(0)<0$ and $\gamma^\omega(\pi)>0$. The latter two inequalities follow from the monotonicity of $z\mapsto ze^z$ on $[-1,\infty)$.

It immediately follows that for any $z\in\gamma^\omega$,
\begin{equation}
 |ze^z|=|1-\omega^2/2|e^{-1}< \sqrt{1+u^2}e^{-1}=|we^w|.
\end{equation}
To see the steep descent property, first notice that $L_0$ satisfies the following differential identity:
\begin{equation}\label{eqDiffId}
 L_0'(z)=\frac{L_0(z)}{z(1+L_0(z))}.
\end{equation}
Combining this formula with $f_3'(z)=-(z+1)^2/z$ and using the shorthand \mbox{$z(u)=-(1-\omega^2/2)e^{-1+\mathrm{i} u}$} leads to
\begin{equation}\label{eqLambDesc}\begin{aligned}
 \frac{\D\Re\left(f_3(\gamma^\omega(u))\right)}{\mathrm{d} u}&=\Re\left(-\frac{(\gamma^\omega(u)+1)^2}{\gamma^\omega(u)}\cdot\frac{\gamma^\omega(u)}{z(u)(1+\gamma^\omega(u))}\cdot z(u)\cdot\mathrm{i}\right)\\
 &=\Re\left(-\mathrm{i}(\gamma^\omega(u)+1)\right)=\Im(\gamma^\omega(u)).
\end{aligned}\end{equation}
So the real part is strictly decreasing along $0<u<\pi$, reaches its minimum at $\gamma^\omega(\pi)$ and then increases along $\pi<u<2\pi$ back to its maximum at $\gamma^\omega(0)=\gamma^\omega(2\pi)$.

By (4.22) in~\cite{CGHJK96} the Lambert W function can be expanded
around the branching point $-e^{-1}$ as
\begin{equation}
L_0(z)=-1+p-\frac{1}{3}p^2+
\frac{11}{72}p^3+\cdots,
\end{equation}
with $p(z)=\sqrt{2(ez+1)}$. This means that
\begin{equation}
 \gamma^\omega(0)=-1+\omega+\Or(\omega^2)
\end{equation}
for small $\omega$.

As both contours are steep descent curves, we can restrict them to a neighbourhood of their critical point while making an error which is exponentially small in $t$. We do this by simply intersecting both curves with the ball $B_{-1}(2\omega)=\{z\in\C,|z+1|<2\omega\}$.

Now we can apply Taylor expansion of the function $G$ as in \eqref{eqTaylf} and estimate the error made by omitting the higher order terms in the same way as in \eqref{eqTaylErr}. With $z=-1+Zt^{-1/3}$ and $w=-1+Wt^{-1/3}$ the term constant in $t$ converges as
\begin{equation}
 \frac{t^{-1/3}(1+z)e^z}{we^w-ze^z}=\frac{t^{-2/3}Ze^{Zt^{-1/3}}}{-1+W^2t^{-2/3}/2+1-Z^2t^{-2/3}/2+\Or(t^{-1})}\to\frac{2Z}{W^2-Z^2}.
\end{equation}
Applying this transformation of variable to the integral results in
\begin{equation}\label{eqhfprefinal}\begin{aligned}
 \mathcal{K}_1^\textrm{resc}(r_1,s_1&;r_2,s_2)+\Or(t^{1/3}e^{-\const t})+\Or(t^{-1/3})\\&=\frac{1}{(2\pi\mathrm{i})^2}\int^{ 2\omega\mathrm{i} t^{1/3}}_{-2\omega \mathrm{i} t^{1/3}}\mathrm{d} W\,\int_{\widetilde{\gamma}^\omega_t}\mathrm{d} Z\,\frac{e^{Z^3/3+r_2Z^2-s_2Z}}{e^{W^3/3+r_1W^2-s_1W}}\frac{2Z}{W^2-Z^2},
\end{aligned}\end{equation}
with $\widetilde{\gamma}^\omega_t=t^{1/3}(\gamma^\omega\cap B_{-1}(2\omega)+1)$. The endpoints of this contour are at $\pm2\omega e^{i\theta}$, where $\theta$ is close to $\pi/4$ for $\omega$ being small. The limit $t\to\infty$ now extends both contours up to infinity. While $t\to\infty$ we have to keep deforming $\widetilde{\gamma}^\omega_t$ in order for it to not vanish to infinity. This is allowed, since as long as $|\arg(Z)|<\pi/4$, we have $\Re(W^2-Z^2)<0$, i.e.\ we are not crossing any poles. We can then change the contours to the generic Airy contours  with the added restriction given in Definition~\ref{defAiry2to1} to arrive at the final result:
\begin{equation}
 \lim_{t\to\infty}\mathcal{K}_1^\textrm{resc}(r_1,s_1;r_2,s_2)=\frac{1}{(2\pi\mathrm{i})^2}\int_{\gamma_W}\hspace{-3pt}\mathrm{d} W\,\int_{\gamma_Z}\hspace{-3pt}\mathrm{d} Z\,\frac{e^{Z^3/3+r_2Z^2-s_2Z}}{e^{W^3/3+r_1W^2-s_1W}}\frac{2Z}{W^2-Z^2}
\end{equation}
\end{proof}

\begin{prop}\label{propK1Bound}
For fixed $r_1,r_2,L$ there exists ${t}_0>0$ such that the estimate
\begin{equation}
	\left|\mathcal{K}_\hf^\textrm{resc}(r_1,s_1;r_2,s_2)\right|\leq\const\cdot e^{-s_1}
\end{equation}
holds for any ${t}>{t}_0$ and $s_1,s_2>0$.
\end{prop}

\begin{proof}
Presumably we could show exponential decay in $s_2$, too, but the analysis would be even more involved, and the single-sided decay is enough for proving the convergence of the Fredholm determinant.

 We start with another representation of the kernel, which can be found in (4.38),~\cite{FSW13} combined with (4.51) therein.
 Defining a new contour
 \begin{equation}
  \widehat{\gamma}^\omega=\{L_{\lfloor u\rfloor}(-(1-\omega^2/2)e^{2\pi\mathrm{i} u-1}),u\in\R\setminus[0,1)\}
 \end{equation}
 this representation is given by
 \begin{equation}\begin{aligned}
   \mathcal{K}_1^\textrm{resc}(r_1,s_1;r_2,s_2)&=\mathcal{K}_{1(a)}^\textrm{resc}(r_1,s_1;r_2,s_2)+\mathcal{K}_{1(b)}^\textrm{resc}(r_1,s_1;r_2,s_2),
 \end{aligned}\end{equation}
 with
 \begin{equation}\begin{aligned}
   \mathcal{K}_{1(a)}^\textrm{resc}(r_1,s_1;r_2,s_2)&=\frac{{t}^{1/3}}{(2\pi\mathrm{i})^2}\int_{\widehat{\gamma}^\omega}\mathrm{d} w\,\oint_{\gamma^0}\mathrm{d} z\,e^{G(z,r_2,s_2)}e^{-G(w,r_1,s_1)}\frac{(1+z)e^z}{we^w-ze^z}\\
   \mathcal{K}_{1(b)}^\textrm{resc}(r_1,s_1;r_2,s_2)&=\frac{t^{1/3}}{2\pi\mathrm{i}}\int_{\widehat{\gamma}^0}\mathrm{d} w\,\frac{e^{tw^2/2+(w+1)\xi_1}(-w)^{n_1}}{e^{t\varphi(w)^2/2+(\varphi(w)+1)\xi_2}(-\varphi(w))^{n_2}},
 \end{aligned}\end{equation}
where $\varphi(w)=L_0(we^w)$. For a visualization of $\widehat{\gamma}^\omega$ see Figure~\ref{figGamma} and proven results about its pathway can be found in Lemma A.1~\cite{FSW13}. Notice that now we have $|ze^z|=e^{-1}>|we^w|$ for any $\omega>0$.

By Corollary 5.2 and Proposition 5.3~\cite{FSW13} $\mathcal{K}_{1(b)}^\textrm{resc}$ is bounded by a constant times $e^{-(s_1+s_2)}$ under the scaling
\begin{equation}\begin{aligned}
 n_i&=-t+2^{5/3}r_i\\
 \xi_i&=2^{5/3}r_i+(2t)^{1/3}.
\end{aligned}\end{equation}
Noticing $\varphi(w)e^{\varphi(w)}=we^w$ we see that the kernel is invariant under simultaneous shifts $n_i\to n_i+k$, $\xi_i\to\xi_i+k$. Using $k=2t$ we recover our scaling up to constant factors, which are irrelevant for our purposes.

We are thus left to show boundedness of $\mathcal{K}_{1(a)}$:
\begin{equation}\label{eqK1Dec}\begin{aligned}
  |\mathcal{K}_{1(a)}^\textrm{resc}&(r_1,s_1;r_2,s_2)|\\
  &\leq\sup_{w\in\widehat{\gamma}^\omega,z\in\gamma^0}\left|\frac{e^z}{we^w-ze^z}\right|
    \oint_{\gamma^0}\mathrm{d} z\,\left|e^{G(z,r_2,s_2)}(1+z)\right|\\
    &\quad\times\frac{{t}^{1/3}}{(2\pi)^2}\int_{\widehat{\gamma}^\omega}\mathrm{d} w\,\left|e^{-G(w,r_1,s_1)}\right|\\
  &\leq e^{-G(w_0,r_1,s_1)}
    \sup_{w\in\widehat{\gamma}^\omega,z\in\gamma^0}\left|\frac{t^{-2/3}e^z}{we^w-ze^z}\right|
    \frac{t^{2/3}}{2\pi}\oint_{\gamma^0}\mathrm{d} z\,\left|e^{G(z,r_2,s_2)}(1+z)\right|\\
   &\quad\times \frac{{t}^{1/3}}{2\pi}\int_{\widehat{\gamma}^\omega}\mathrm{d} w\,\left|e^{-G(w,r_1,s_1)+G(w_0,r_1,s_1)}\right|.
\end{aligned}\end{equation}
The point $w_0$ is defined by
\begin{equation}
 w_0=\widehat{\gamma}^\omega(0)=-1-\omega+\Or(\omega^2),
\end{equation}
and $\omega$ is chosen specifically as $\omega:=\min\{t^{-1/3}\sqrt{s_2},\e\}$ for some small positive $\e$ chosen in the following.

We can estimate the factor $e^{G(w_0,r_1,s_1)}$ as in the proof of Lemma~\ref{lemAlphaBound}, see \eqref{eqG0bound} and the preceding arguments:
\begin{equation}
 e^{-G(w_0,r_1,s_1)}\leq e^{-\frac{1}{2}\omega t^{1/3}s_1}\leq c_L e^{-s_1}.
\end{equation}
Noticing $|ze^z|=e^{-1}$ and $|we^w|=|1-\omega^2/2|e^{-1}$ the remaining prefactor can be estimated as
\begin{equation}
 \sup_{w\in\widehat{\gamma}^\omega,z\in\gamma^0}\left|\frac{t^{-2/3}e^z}{we^w-ze^z}\right|\leq \const \cdot\frac{1}{\omega^2 t^{2/3}}
\end{equation}
by choosing $\e$ small enough. This can be bounded, as for any fixed $\e$, $\omega t^{1/3}$ is large by choosing $t_0$ and $L$ large.

It remains to show that the two integral expressions converge, starting with the $z$ integral. First notice that by $s_2\geq 0$, and $\Re(z)\geq -1$, we have $|e^{G(z,r_2,s_2)}|\leq |e^{G(z,r_2,0)}|$, avoiding the problem of large $s_2$. Now we can apply standard steep descent analysis, i.e.\ restrict the contour to a $\delta$ neighbourhood of the critical point $-1$, Taylor expand the integrand and estimate the error term as in \eqref{eqTaylErr}. The analogue of \eqref{eq3.2.26} becomes
\begin{equation}\begin{aligned}
 \frac{t^{2/3}}{2\pi}\int_{\gamma^0_\delta}|\mathrm{d} z|\,|e^{\tilde{G}(z)}(1+z)|&=\frac{t^{2/3}}{2\pi}\int_{\gamma^0_\delta+1}|\mathrm{d} \omega|\,|e^{t\omega^3/3+t^{2/3}r\omega^2}\omega|\\&\leq\frac{1}{2\pi}\int_{e^{\pi\mathrm{i}/4}\infty}^{-e^{\pi\mathrm{i}/4}\infty}|\mathrm{d} Z|\,|e^{Z^3/3+rZ^2}Z|<\infty
\end{aligned}\end{equation}

Similarly the inequalities $\Re(w-w_0)<0$ and $s_1\geq0$ allow us to simply set $s_1=0$ in the $w$ integral. The differential identity \eqref{eqDiffId} holds for any branch of the Lambert function, so \eqref{eqLambDesc} can be derived analogously for $\widehat{\gamma}^\omega(u)=L_{\lfloor u\rfloor}(-(1-\omega^2/2)e^{2\pi\mathrm{i} u-1}),u\in\R\setminus[0,1)$:
\begin{equation}\begin{aligned}
 \frac{\D\Re\left(-f_3(\gamma(u))\right)}{\mathrm{d} u}=-\Im(\gamma(u)).
\end{aligned}\end{equation}
$\widehat{\gamma}^\omega$ is thus a steep descent curve for the leading order term and the rate of decline along the path is even quadratic in $|w|$, allowing us to restrict the contour to a small neighbourhood of $w_0$. The path $\widehat{\gamma}^\omega$ is close to vertical in a neighbourhood of $w_0$, thus we can proceed as in the proof of Lemma~\ref{lemAlphaBound} from here on.
\end{proof}

\section{Half-Poisson initial conditions}
In this section we will study the initial condition given by a Poisson process on the positive half-line and no particles on the negative one. Let therefore be $\{\textrm{Exp}_n,n\in\Z_{\geq0}\}$ be i.i.d. random variables with exponential distribution with parameter $1$. For a parameter $\lambda>0$ define the initial condition \mbox{$\vec{x}(0)=\vec{\zeta}^\shst(\lambda)$} by
\begin{equation}\label{halfstatModel}
 \begin{aligned}
  \zeta^\hst_0&=\lambda^{-1}\,\textrm{Exp}_0,\\
  \zeta^\hst_n-\zeta^\hst_{n-1}&=
  \lambda^{-1}\,\textrm{Exp}_n, \quad &\text{for } n\geq1.
 \end{aligned}
\end{equation}
As before, a variation of $\zeta_0$ introduces only a bounded modification of the initial condition and is thus irrelevant in the scaling limit as long as it stays bounded almost surely. We chose it in this way to keep the determinantal structure simple.

Notice that if we let $\rho\to0$ in the initial condition $\vec{\zeta}^\textrm{stat}(\lambda,\rho)$ all particles with negative label vanish to $-\infty$ and can consequently be ignored. So $\vec{\zeta}^\shst(\lambda)$ is simply $\vec{\zeta}^\textrm{stat}(\lambda,0)$ shifted by $\zeta^\hst_0$. This means that we can obtain the determinantal structure as a corollary of Proposition~\ref{propKernel}. Strictly speaking, $\vec{\zeta}^\shst(\lambda)$ equals $\vec{\zeta}^\textrm{stat}(\lambda,0)$ \emph{under the modified measure} $\Pb_+$. This means the shift argument, i.e.\ \emph{Step 2} of the proof of Proposition~\ref{propKernel}, is not necessary and we obtain the determinantal formula by simply specifying \eqref{eqDistP+} to $\rho=0$:

\begin{prop}\label{propHalfStatKernel}
Let $\{x_n(t),n\in\Z_{\geq0}\}$ be the system of one-sided reflected Brownian motions with initial condition $\vec{x}(0)=\zeta^\hst(\lambda)$ for any $\lambda>0$. For any finite subset $S$ of $\Z_{\geq0}$, it holds
\begin{equation}%\label{eq33}
\Pb\bigg(\bigcap_{n\in S} \{x_n(t)\leq a_n\}\bigg)=\det(\Id-\chi_a \mathcal{K}_\hst \chi_a)_{L^2(S\times \R)},
\end{equation}
where $\chi_a(n,\xi)=\Id_{\xi>a_n}$. The kernel $\mathcal{K}_\hst$ is given by
\begin{equation}\label{eqKthalfStat}
\begin{aligned}	\mathcal{K}_\hst(n_1,\xi_1;n_2,\xi_2)&=-\phi_{n_1,n_2}(\xi_1,\xi_2)\Id_{n_2>n_1}+\mathcal{K}_0(n_1,\xi_1;n_2,\xi_2)\\&\quad+\lambda\mathpzc{f}(n_1,\xi_1)\mathpzc{g}_0(n_2,\xi_2).
\end{aligned}
\end{equation}
where
\begin{equation}\begin{aligned}
	\phi_{n_1,n_2}(\xi_1,\xi_2)&=\frac{(\xi_2-\xi_1)^{n_2-n_1-1}}{(n_2-n_1-1)!}\Id_{\xi_1\leq \xi_2}, \quad\text{ for } 0\leq n_1<n_2,\\	\mathcal{K}_0(n_1,\xi_1;n_2,\xi_2)&=\frac{1}{(2\pi\mathrm{i})^2}\int_{\mathrm{i}\R-\e}\mathrm{d} w\,\oint_{\Gamma_0}\mathrm{d} z\,\frac{e^{{t} w^2/2+\xi_1w}}{e^{{t} z^2/2+\xi_2 z}}\frac{(-w)^{n_1}}{(-z)^{n_2}}\frac{1}{w-z},\\
	\mathpzc{f}(n_1,\xi_1)&=\frac{1}{2\pi\mathrm{i}}\int_{\mathrm{i}\R-\e}\mathrm{d} w\,\frac{e^{{t} w^2/2+\xi_1w}(-w)^{n_1}}{w+\lambda},\\
	\mathpzc{g}_0(n_2,\xi_2)&=\frac{-1}{2\pi\mathrm{i}}\oint_{\Gamma_{0}}\mathrm{d} z\,e^{-{t} z^2/2-\xi_2 z}(-z)^{-n_2-1},
\end{aligned}\end{equation}
for any fixed $0<\e<\lambda$.
\end{prop}

We assume $\lambda=1$ from now on. It is clear that the initial macroscopic shape in this case is the same as in the half periodic case. We thus have again the law of large numbers,
\begin{equation}
 \lim_{t\to\infty}\frac{x_{\alpha^2 t}(t)}{t}=\begin{cases}2\alpha, &\text{ for } 0\leq\alpha\leq1 \\1+\alpha^2,&\text{ for } \alpha>1.\end{cases}
\end{equation}
The fluctuations in the case $0<\alpha<1$ are of order $t^{1/3}$ and given by the Airy$_2$ process, while for $\alpha>1$ they are now governed by the Airy$_\textrm{stat}$ process. For finite $n$, $x_n(t)$ fluctuates on the scale $\sqrt{t}$ according to Dyson's Brownian motion. At the interesting transition point $\alpha=1$, however, the Airy$_{2\to\textrm{BM}}$ process will appear:

\begin{thm}\label{thmAsympHalfStat}
With $\{x_n(t),n\in\Z_{\geq0}\}$ being the system of one-sided reflected Brownian motions with initial condition $\vec{x}(0)=\zeta^\hst(1)$, define the rescaled process
\begin{equation}%\label{eqScaledProcessOriginal}
	r\mapsto X_t^\hst(r) = t^{-1/3}\big(x_{\lfloor t+2rt^{2/3}\rfloor}(t)-2t-2rt^{2/3} \big).\,
\end{equation}
In the sense of finite-dimensional distributions,
\begin{equation}%\label{eqX0limit}
	\lim_{t\to\infty}X_t^\hst(r)\stackrel{d}{=}\mathcal{A}_{2\to \textrm{BM}}(r)-r^2.
\end{equation}
\end{thm}
\begin{proof}[Proof of Theorem~\ref{thmAsympHalfStat}]
With
\begin{equation}\begin{aligned}
	n_i&={t}+2{t}^{2/3}r_i\\
	\xi_i&=2{t}+2{t}^{2/3}r_i+{t}^{1/3}s_i,
\end{aligned}\end{equation}
define the rescaled kernel
\begin{equation}
	\mathcal{K}_\hst^\textrm{resc}(r_1,s_1;r_2,s_2)={t}^{1/3}e^{\xi_1-\xi_2}\mathcal{K}_\hst(n_1,\xi_1;n_2,\xi_2).
\end{equation}
It decomposes into
\begin{equation}\label{eqKHfResc}\begin{aligned}
	\mathcal{K}_\hst^\textrm{resc}(r_1,s_1;r_2,s_2)&=-\phi_{r_1,r_2}^\textrm{resc}(s_1,s_2)\Id_{r_1<r_2}+\mathcal{K}_0^\textrm{resc}(r_1,s_1;r_2,s_2)\\
	&\quad+\mathpzc{f}^\textrm{resc}(r_1,s_1)\mathpzc{g}_0^\textrm{resc}(r_2,s_2),
\end{aligned}\end{equation}
by
\begin{equation}\begin{aligned}
	\mathpzc{f}^\textrm{resc}(r_1,s_1)&=e^{-{t}/2+\xi_1}\mathpzc{f}(n_1,\xi_1)\\
	\mathpzc{g}_0^\textrm{resc}(r_2,s_2)&=e^{{t}/2-\xi_2}t^{1/3}\mathpzc{g}_0(n_2,\xi_2).
\end{aligned}\end{equation}

We can apply the proof of Theorem~\ref{thmAsymp} once we have shown compact convergence as well as uniform boundedness of the kernel. The former means that for any $r_1,r_2$ in a bounded set and fixed $L$, the kernel converges as
\begin{equation}\label{eqHalfStatConv}
	\lim_{{t}\to\infty}\mathcal{K}_\hst^\textrm{resc}(r_1,s_1;r_2,s_2)=\frac{e^{\frac{2}{3}r_2^3+r_2s_2}}{e^{\frac{2}{3}r_1^3+r_1s_1}} K_{{\cal A}_{2\to\textrm{BM}}}(r_1,s_1+r_1^2,;r_2,s_2+r_2^2)
\end{equation}
uniformly for $(s_1,s_2)\in[-L,L]^2$.

The first two parts of \eqref{eqKHfResc} are simply $\mathcal{K}_\textrm{packed}^\textrm{resc}$, so we can apply Proposition~\ref{propStepPointw} to them. In the proof of Proposition~\ref{propPointw} it is shown that $\mathpzc{f}^\textrm{resc}(r_1,s_1)$ converges compactly to $f_{r_1}(s_1)$. Finally,
\begin{equation}\
 \mathpzc{g}_0^\textrm{resc}(r_2,s_2)=-\beta_t(r+t^{-2/3}/2,s-t^{-1/3}),
\end{equation}
allows applying Lemma~\ref{lemAlphaLimit}. Putting this together, we have that as $t\to\infty$, the kernel $\mathcal{K}_\hst^\textrm{resc}$ converges compactly to
\begin{equation}\begin{aligned}
 \mathcal{K}_{\mathcal{A}_2}&(r_1,s_1;r_2,s_2)\\&+\left(1-e^{-\frac{2}{3}r_1^3-r_1s_1}\int_0^\infty\mathrm{d} x\, \Ai(r_1^2+s_1+x)e^{-r_1x}\right)\Ai(r_2^2+s_2)e^{\frac{2}{3}r_2^3+r_2s_2},
\end{aligned}\end{equation}
which is the right hand side of \eqref{eqHalfStatConv}.

Since we can apply Proposition~\ref{propStepPhiBound} and Proposition~\ref{propStepK0Bound} on the first two parts of the rescaled kernel, the only bound we have to show is the following:

For fixed $r_1,r_2,L$ there exists ${t}_0>0$ such that the estimate
\begin{equation}\label{eqHalfStatBound}
	\left|\mathpzc{f}^\textrm{resc}(r_1,s_1)\mathpzc{g}_0^\textrm{resc}(r_2,s_2)\right|\leq\const\cdot e^{-s_2}
\end{equation}
holds for any ${t}>{t}_0$ and $s_1,s_2>0$.

This can be seen by:
\begin{equation}\begin{aligned}
 \Big|\mathpzc{f}^\textrm{resc}&(r_1,s_1)\mathpzc{g}_0^\textrm{resc}(r_2,s_2)\Big|\\
 &=\left|1-\int_0^\infty \mathrm{d} x\, \alpha_{{t}}(r_1,s_1+x)\right|\left|\beta_t(r_2+t^{-2/3}/2,s_2-t^{-1/3})\right|\\
 &\leq\left(1+\int_0^\infty \mathrm{d} x\, e^{-(s_1+x)}\right)e^{-s_2+t^{-1/3}}\leq\const\cdot e^{-s_2},
\end{aligned}\end{equation}
which finishes the proof.
\end{proof}

\section{Poisson-Periodic initial conditions}\label{secStFl}
The remaining mixed initial condition consists of equally spaced particles on the positive half-axis and a Poisson process on the negative one.  Let therefore be $\{\textrm{Exp}_n,n\in\Z_{\leq1}\}$ be i.i.d. random variables with exponential distribution with parameter $1$. For a parameter $\rho>0$ define the initial condition \mbox{$\vec{x}(0)=\vec{\zeta}^\stf(\rho)$} by
\begin{equation}\label{statflatModel}
 \begin{aligned}
  \zeta^\stf_n&=n, \quad &\text{for } n\geq1\\
  \zeta^\stf_n-\zeta^\hst_{n-1}&=  \rho^{-1}\,\textrm{Exp}_n, \quad &\text{for } n\leq1.
 \end{aligned}
\end{equation}
By the stationarity property, we know that $x_1(t)=1+\widetilde{B}(t)+\rho t$ for a Brownian motion $\widetilde{B}(t)$ that is independent of $B_n(t)$, $n\geq2$. We can therefore restrict our analysis to a half-infinite system, which is in fact the same system as in the half-periodic setting except for the drift of the first particle $x_1(t)$.

\begin{prop}\label{propStatFlatKernel}
Let $\{x_n(t),n\in\Z\}$ be the system of one-sided reflected Brownian motions with initial condition \mbox{$\vec{x}(0)=\vec{\zeta}^\stf(\rho)$}, $\rho>0$. Then for any finite subset $S$ of $\Z_{>0}$, it holds
\begin{equation}%\label{eq33}
\Pb\bigg(\bigcap_{n\in S} \{x_n(t)\leq a_n\}\bigg)=\det(\Id-\chi_a \mathcal{K}_\stf\chi_a)_{L^2(S\times \R)},
\end{equation}
where $\chi_a(n,\xi)=\Id_{\xi>a_n}$. The kernel $\mathcal{K}_\stf$ is given by
\begin{equation}\begin{aligned}\label{eqKtstf}
\mathcal{K}_\stf&(n_1,\xi_1;n_2,\xi_2)\\
&=\mathcal{K}_\hf(n_1,\xi_1;n_2,\xi_2) + \Psi^{n_1}_{{n_1}-1}(\xi_1)\left(\widehat{\Phi}^{n_2}_{(1)}(\xi_2)+\widehat{\Phi}^{n_2}_{(2)}(\xi_2)\right),
\end{aligned}\end{equation}
with $\mathcal{K}_\hf$ as in \eqref{eqKtHf} and
\begin{equation}
\begin{aligned}	\Psi^n_{n-1}(\xi)&=\frac{1}{2\pi\mathrm{i}}\int_{\mathrm{i}\R-\e}\mathrm{d} w\,\,e^{{t} w^2/2+w(\xi-1)}(-w)^{n-1}\\
\widehat{\Phi}^n_{(1)}(\xi)&=\frac{(-1)^{n}}{2\pi\mathrm{i}}\oint_{\Gamma_0}\mathrm{d} z\,\frac{e^{-t z^2/2-z(\xi-1)}}{z^{n}}\frac{\rho(1+z)}{\rho+ze^{\rho+z}},\\
\widehat{\Phi}^n_{(2)}(\xi)&=\rho^{1-n}e^{-t\rho^2/2+\rho(\xi-1)}.
\end{aligned}
\end{equation}
\end{prop}
\begin{proof}[Proof of Proposition~\ref{propStatFlatKernel}]
 Applying Proposition~\ref{PropWarren} with parameters $\mu_1=\rho$, $\mu_k=0$ for $2\leq k\leq N$ gives
\begin{equation}\label{eqSFPstep}
  \Pb\left(\vec{x}({t})\in \D\vec{\xi}\right)=e^{\rho(\xi_1-1)-t\rho^2/2}\det_{1\leq k,l\leq N}[F_{k-l}(\xi_{N+1-l}-\zeta_{N+1-k},{t})],
\end{equation}
where
\begin{equation}
	F_{k}(\xi,{t})=\frac{1}{2\pi\mathrm{i}}\int_{\mathrm{i}\R+1}\mathrm{d} w\,e^{{t} w^2/2+\xi w}w^k.
\end{equation}

Using repeatedly the identity
\begin{equation}
	F_{k}(\xi,{t})=\int^\xi_{-\infty} \mathrm{d} x\,F_{k+1}(x,{t}),
\end{equation}
relabeling $\xi_1^k:=\xi_k$, and introducing new variables $\xi_l^k$ for $2\leq l\leq k\leq N$,
we can write
\begin{equation}\begin{aligned}
	\det_{1\leq k,l\leq N}&\big[F_{k-l}(\xi_1^{N+1-l}-\zeta_{N-1+k},{t})\big]\\&
	=\int_{\mathcal D'} \det_{1\leq k,l\leq N}\big[F_{k-1}(\xi_l^{N}-\zeta_{N-1+k},{t})\big]\prod_{2\leq l\leq k\leq N} \D\xi_l^k,
\end{aligned}\end{equation}
where $\mathcal D' = \{\xi_l^k\in\R,2\leq l\leq k\leq N|x_l^k\leq x_{l-1}^{k-1}\}$. Using the antisymmetry of the determinant and encoding the constraint on the integration variables into indicator functions, we obtain that the measure \eqref{eqSFPstep} is a marginal of
\begin{equation}\label{SFExtM}\begin{aligned}
	\const&\cdot e^{\rho\xi_1^1} \prod_{n=2
	}^{N}\det_{1\leq i,j\leq n}\left[\Id_{\xi_i^{n-1}\leq\xi_j^n}\right]\det_{1\leq k,l\leq N}\big[ F_{k-1}(\xi_l^{N}-\zeta_{N-1+k},{t})\big]\\=
	\const&\cdot \prod_{n=1}^{N}\det_{1\leq i,j\leq n}\big[\phi_n(\xi_i^{n-1},\xi_j^n)\big]\det_{1\leq k,l\leq N}\big[ F_{k-1}(\xi_l^{N}-\zeta_{N-1+k},{t})\big]
\end{aligned}\end{equation}

with
\begin{equation}\begin{aligned}
	\tilde{\phi}_n(x,y)&=\Id_{x\leq y}, \quad \text{for } n\geq2\\
	\tilde{\phi}_1(x,y)&=e^{\rho y},
\end{aligned}\end{equation}
and using the convention that $\xi_n^{n-1}\leq y$ always holds.

The measure \eqref{SFExtM} has the appropriate form for applying Lemma~\ref{lemDetMeasure}. The composition of the $\tilde{\phi}$ functions can be evaluated explicitly as
\begin{equation}\begin{aligned}
	\tilde{\phi}_{0,n}(x,y)&=(\tilde{\phi}_1*\dots*\tilde{\phi}_{n})(x,y)=\rho^{1-n}e^{\rho y}, & &\text{for } n\geq1,\\
	\tilde{\phi}_{m,n}(x,y)&=(\tilde{\phi}_{m+1}*\dots*\tilde{\phi}_{n})(x,y)=\frac{(y-x)^{n-m-1}}{(n-m-1)!}\Id_{x\leq y}, & &\text{for } n>m\geq1.
\end{aligned}\end{equation}
Define
\begin{equation}
	\Psi^n_{n-k}(\xi):=\frac{(-1)^{n-k}}{2\pi\mathrm{i}}\int_{\mathrm{i}\R-\e}\mathrm{d} w\,\,e^{{t} w^2/2+w(\xi-k)}w^{n-k},
\end{equation}
for $n,k\geq1$ and some $\e>0$. In the case $n\geq k$ the integrand has no poles, which implies $\Psi^n_{n-k}(\xi)=(-1)^{n-k}F_{n-k}(\xi-\zeta_k)$, now understood specifically with $\zeta_k=\zeta^\stf_k=k$. Reversing the order of the index $k$ one sees that the second determinant in \eqref{SFExtM} is equal to $\det_{k,l}\big[ \Psi^N_{N-k}(\xi_l^{N})\big]$. The straightforward recursion
\begin{equation}
	(\tilde{\phi}_n*\Psi^n_{n-k})(\xi)=\Psi^{n-1}_{n-1-k}(\xi)
\end{equation}
eventually leads to condition \eqref{Sasdef_psi} being satisfied.
A basis for the space $V_n$ is given by
\begin{equation}
	\{e^{\rho x},x^{n-2},x^{n-3},\dots,x,1\}.
\end{equation}
Choose functions $\Phi^n_{n-k}$ as follows
\begin{equation}\label{eq6.3.11}
		\Phi^n_{n-k}(\xi)=\begin{cases}
\widetilde{\Phi}^n_{n-k}(\xi)& 2\leq k\leq n,\\
\widetilde{\Phi}^n_{n-1}(\xi) +\widehat{\Phi}^n_{(1)}(\xi)+\widehat{\Phi}^n_{(2)}(\xi)& k=1,
\end{cases}
\end{equation}
where
\begin{equation}\begin{aligned}
\widetilde{\Phi}^n_{n-k}(\xi)&=\frac{(-1)^{n-k}}{2\pi\mathrm{i}}\oint_{\Gamma_0}\mathrm{d} z\,\frac{e^{-t z^2/2-z(\xi-k)}}{z^{n-k+1}}(1+z),\\
\widehat{\Phi}^n_{(1)}(\xi)&=\frac{(-1)^{n}}{2\pi\mathrm{i}}\oint_{\Gamma_0}\mathrm{d} z\,\frac{e^{-t z^2/2-z(\xi-1)}}{z^{n}}\frac{\rho(1+z)}{\rho+ze^{\rho+z}},\\
\widehat{\Phi}^n_{(2)}(\xi)&=\rho^{1-n}e^{-t\rho^2/2+\rho(\xi-1)}.
\end{aligned}\end{equation}
The functions $\widetilde{\Phi}^n_{n-k}$ are polynomials of order $n-k$ by elementary residue calculating rules. $\widetilde{\Phi}^n_{n-1}+\widehat{\Phi}^n_{(1)}$ is in fact a polynomial of order $n-2$, since the poles of order $n$ in each integrand cancel each other out exactly. So the functions \eqref{eq6.3.11} indeed generate $V_n$. To show the orthogonality \eqref{Sasortho}, we decompose the scalar product as follows:
\begin{equation}\label{eqSF2.41}
\int_{\R_-} \mathrm{d} \xi\, \Psi^{n}_{n-k}(\xi) \widetilde{\Phi}^n_{n-\ell}(\xi) + \int_{\R_+} \mathrm{d} \xi\, \Psi^{n}_{n-k}(\xi) \widetilde{\Phi}^n_{n-\ell}(\xi).
\end{equation}
Since $n-k\geq0$ we are free to choose the sign of $\e$ as necessary. For the first term, we choose $\e<0$ and the path $\Gamma_0$ close enough to zero, such that always \mbox{$\Re(w-z)>0$}. Then, we can take the integral over $\xi$ inside and obtain
\begin{equation}
\int_{\R_-} \mathrm{d} \xi\, \Psi^{n}_{n-k}(\xi) \widetilde{\Phi}^n_{n-\ell}(\xi) =\frac{(-1)^{k-l}}{(2\pi\mathrm{i})^2}\int_{\mathrm{i}\R-\e}\mathrm{d} w \oint_{\Gamma_0}\mathrm{d} z\, \frac{e^{{t} w^2/2} w^{n-k}e^{-wk}(1+z)}{e^{{t} z^2/2}z^{n-\ell+1}e^{-z\ell}(w-z)}.
\end{equation}
For the second term, we choose $\e>0$ to obtain \mbox{$\Re(w-z)<0$}. Then again, we can take the integral over $\xi$ inside and arrive at the same expression up to a minus sign. The net result of \eqref{eqSF2.41} is a residue at $w=z$, which is given by
\begin{equation}
\frac{(-1)^{k-l}}{2\pi\mathrm{i}}\oint_{\Gamma_0}\mathrm{d} z\, \left(ze^z\right)^{\ell-k}\frac{1+z}{z}=\frac{(-1)^{k-l}}{2\pi\mathrm{i}}\oint_{\Gamma_0}\mathrm{d} Z\, Z^{\ell-k-1}=\delta_{k,\ell},
\end{equation}
where we made the change of variables $Z=ze^z$. In the same way we get
\begin{equation}\label{eq6.3.16}\begin{aligned}
\int_\R \mathrm{d} \xi\, \Psi^{n}_{n-k}(\xi) \widehat{\Phi}^n_{(1)}(\xi)&=\frac{(-1)^{k}}{2\pi\mathrm{i}}\oint_{\Gamma_0}\mathrm{d} z\, \left(ze^z\right)^{1-k}\frac{\rho(1+z)}{(\rho+ze^{\rho+z})z}\\
&=\frac{1}{2\pi\mathrm{i}}\oint_{\Gamma_0}\mathrm{d} Z\, \left(-Z\right)^{-k}\frac{1}{(1+Ze^\rho/\rho)}\\&=\frac{1}{2\pi\mathrm{i}}\oint_{\Gamma_0}\mathrm{d} Z\,\left(-Z\right)^{-k}\sum_{i\geq0}\left(-Z\frac{e^\rho}{\rho}\right)^i=-\left(\rho e^{-\rho}\right)^{1-k}.
\end{aligned}\end{equation}
Regarding the scalar product with $\widehat{\Phi}^n_{(2)}$, choose $\e<\rho$ for the integral over $\R_-$ and $\e>\rho$ for the one over $\R_+$. We are left with the following:
\begin{equation}\begin{aligned}
\int_\R \mathrm{d} \xi\, \Psi^{n}_{n-k}(\xi) \widehat{\Phi}^n_{(2)}(\xi)&=\frac{(-1)^{n-k}\rho^{1-n}}{2\pi\mathrm{i}}\oint_{\Gamma_{-\rho}}\mathrm{d} w\, e^{t(w^2-\rho^2)/2-wk-\rho}w^{n-k}\frac{1}{w+\rho}\\&=\left(\rho e^{-\rho}\right)^{1-k},
\end{aligned}\end{equation}
which cancels out \eqref{eq6.3.16}, eventually proving the orthogonality relation \mbox{$\left\langle\Psi^n_{n-k}(\xi),\Phi^n_{n-l}(\xi)\right\rangle=\delta_{k,l}$}.

Furthermore, both $\phi_n(\xi_{n}^{n-1},x)$ and $\Phi_0^{n}(\xi)$ are constants, so the kernel has a simple form \eqref{SasK}:
\begin{equation}
\mathcal{K}_\stf(n_1,\xi_1;n_2,\xi_2)=-\phi_{n_1,n_2}(\xi_1,\xi_2)\Id_{n_2>n_1} + \sum_{k=1}^{n_2} \Psi_{n_1-k}^{n_1}(\xi_1) \Phi_{n_2-k}^{n_2}(\xi_2).
\end{equation}

Note that we are free to extend the summation over $k$ up to infinity, since the integral expression for $\Phi_{n-k}^{n}(\xi)$ vanishes for $k>n$ anyway. Taking the sum inside the integrals we can write
\begin{equation}\begin{aligned}\label{eqStepFKernSum}
	\sum_{k\geq1}\;&\Psi_{n_1-k}^{n_1}(\xi_1) \widetilde{\Phi}_{n_2-k}^{n_2}(\xi_2)\\&=\frac{1}{(2\pi\mathrm{i})^2}\int_{\mathrm{i}\R-\e}\hspace{-0.4cm}\mathrm{d} w \oint_{\Gamma_0}\mathrm{d} z\,	\frac{e^{{t} w^2/2+\xi_1 w}(-w)^{n_1}}{e^{{t} z^2/2+\xi_2 z}(-z)^{n_2}}\frac{z+1}{z}\sum_{k\geq1}\frac{(ze^z)^k}{(we^w)^k}.
\end{aligned}\end{equation}

By choosing contours such that $|z|<|w|$, we can use the formula for a geometric series, resulting in
\begin{equation}\begin{aligned}
	\eqref{eqStepFKernSum}&=\frac{1}{(2\pi\mathrm{i})^2}\int_{\mathrm{i}\R-\e}\hspace{-0.4cm}\mathrm{d} w \oint_{\Gamma_0}\mathrm{d} z\,	\frac{e^{{t} w^2/2+\xi_1 w}(-w)^{n_1}}{e^{{t} z^2/2+\xi_2 z}(-z)^{n_2}}\frac{(z+1)e^z}{we^w-ze^z}\\
	&=\mathcal{K}_1(n_1,\xi_1;n_2,\xi_2).
\end{aligned}\end{equation}
Note also that we required $n_1>0$, in which case the definition of the function $\phi$ is equal to the one in \eqref{eqK1def}, allowing us to combine these equations to \eqref{eqKtstf}.
\end{proof}

The asymptotic behaviour of the Poisson-periodic initial condition depends substantially on the parameter $\rho$. For $\rho>1$ we have a \emph{shock} moving in positive direction, where the macroscopic particle density changes discontinuously. In the case $0<\rho<1$ there is a region of linearly increasing density between the two plateaus of density $\rho$ and $1$, and at the right edge of this region we have again the Airy$_{2\to1}$ process. We will not focus more on these two cases and refer the reader to~\cite{BFS09} where these results are derived for TASEP.

Instead we will study the case $\rho=1$, that where the last of the crossover Airy processes appears. There is a constant macroscopic density of $1$, which also results in the particles having an average speed of $1$. The law of large numbers therefore looks like this:
\begin{equation}
 \lim_{t\to\infty}\frac{x_{\alpha t}(t)}{t}=1+\alpha.
\end{equation}
The behaviour of the fluctuations is not clear immediately. One thing that follows directly from the Poisson case, is that for $\alpha\leq0$ we have fluctuations of order $t^{1/3}$ described by $\mathcal{A}_\mathrm{stat}$, if we normalize along the characteristic direction, i.e.\ consider the position relative to $x_{(\alpha-1)t}(0)$ (otherwise it is just Gaussian fluctuations on a $t^{1/2}$ scale). It is also to be expected, that for large $\alpha$, the periodic initial condition will dominate, leading to $t^{1/3}$ fluctuations given by $\mathcal{A}_1$. As it turns out, there is again a single transition point, and it is $\alpha=1$. The limit process is $\mathcal{A}_\mathrm{stat}$ for $\alpha<1$, $\mathcal{A}_1$ for $\alpha>1$, and a transition process $\mathcal{A}_\mathrm{BM\to1}$ for $\alpha=1$. We will prove the limit only in the transition case.

\begin{thm}\label{thmAsympStatFlat}
With $\{x_n(t),n\in\Z\}$ being the system of one-sided reflected Brownian motions with initial condition $\vec{x}(0)=\vec{\zeta}^\stf(1)$, define the rescaled process
\begin{equation}%\label{eqScaledProcessOriginal}
	r\mapsto X_t^\stf(r) = t^{-1/3}\big(x_{\lfloor t+2rt^{2/3}\rfloor}(t)-2t-2rt^{2/3} \big).\,
\end{equation}
In the sense of finite-dimensional distributions,
\begin{equation}%\label{eqX0limit}
	\lim_{t\to\infty}X_t^\stf(r)\stackrel{d}{=}\mathcal{A}_\mathrm{BM \to 1}(r)-r^2.
\end{equation}
\end{thm}

\begin{proof}
With
\begin{equation}\begin{aligned}
	n_i&={t}+2{t}^{2/3}r_i\\
	\xi_i&=2{t}+2{t}^{2/3}r_i+{t}^{1/3}s_i,
\end{aligned}\end{equation}
define the rescaled kernel
\begin{equation}
	\mathcal{K}_\stf^\textrm{resc}(r_1,s_1;r_2,s_2)={t}^{1/3}e^{\xi_1-\xi_2}\mathcal{K}_\stf(n_1,\xi_1;n_2,\xi_2).
\end{equation}

 Once the Propositions~\ref{propKsfPointw} and~\ref{propKsfBound} are established, the result follows in the same way as in the proof of Theorem~\ref{thmAsymp}.
\end{proof}

\begin{prop}\label{propKsfPointw}
Consider any $r_1,r_2$ in a bounded set and fixed $L$. Then, the kernel converges as
\begin{equation}
	\lim_{{t}\to\infty}\mathcal{K}_\stf^\textrm{resc}(r_1,s_1;r_2,s_2)=K_{\mathcal{A}_\mathrm{BM \to 1}}(r_1,s_1;r_2,s_2)
\end{equation}
uniformly for $(s_1,s_2)\in[-L,L]^2$.
\end{prop}

\begin{proof}
 We can apply Proposition~\ref{propK1Pointw} to the $\mathcal{K}_\hf^\textrm{resc}$ part of the kernel and are left with studying the additional part of \eqref{eqKtstf} given as a product.

 We attach a factor $e^{-t/2-1}$ to $\Psi$ and its inverse to both $\widehat{\Phi}_{(1)}$ and $\widehat{\Phi}_{(2)}$. Now,
 \begin{equation}\begin{aligned}
  t^{1/3}e^{\xi_1-t/2-1}\Psi^{n_1}_{{n_1}-1}(\xi_1)&=\frac{t^{1/3}}{2\pi\mathrm{i}}\int_{\mathrm{i}\R-\e}\mathrm{d} w\,\,e^{{t} (w^2-1)/2+(w+1)(\xi_1-1)}(-w)^{n_1-1}\\
  &=\alpha_t(r_1-t^{-2/3}/2,s_1)
 \end{aligned}\end{equation}
which converges to $e^{-\frac{2}{3}r_1^3-r_1s_1}\Ai(r_1^2+s_1)$ uniformly for $s_1$, $r_1$ in a compact set.

The function $\widehat{\Phi}_{(2)}$ satisfies
\begin{equation}\label{eqPhi2}
 e^{-\xi_2+t/2+1}\widehat{\Phi}^n_{(2)}(\xi_1)=1.
\end{equation}

We are thus left to prove the limit of $\widehat{\Phi}_{(1)}$. Recognize that
\begin{equation}\begin{aligned}
 e^{-\xi_2+t/2+1}\widehat{\Phi}^n_{(1)}(\xi_1)&=\frac{{t}^{1/3}}{2\pi\mathrm{i}}\oint_{\Gamma_0}\hspace{-2pt}\mathrm{d} z\,e^{-{t}(z^2-1)/2-(\xi_2-1)(z+1)}(-z)^{-n_2}\frac{(1+z)t^{-1/3}}{1+ze^{1+z}}
\end{aligned}\end{equation}
is precisely $\beta_t(r_2,s_2-t^{-1/3})$ up to the fraction appearing in the integrand. We can thus carry over the proof of Lemma~\ref{lemAlphaLimit} almost completely. Up to \eqref{eq3.2.26} the generalization is straightforward, and the analogue of this equation becomes:
\begin{equation}\label{eq6.3.32}\begin{aligned}
  \frac{t^{1/3}}{2\pi\mathrm{i}}\int_{\widetilde{\Gamma}_\delta}\mathrm{d} z\,&e^{\tilde{G}(z)}\frac{(1+z)t^{-1/3}}{1+ze^{1+z}}\\
  &=\frac{t^{1/3}}{2\pi\mathrm{i}}\int_{\widetilde{\Gamma}_\delta+1}\mathrm{d} \omega\,e^{t\omega^3/3+t^{2/3}r\omega^2-t^{1/3}s\omega}\frac{2}{\omega t^{1/3}}(1+\Or(\omega))\\&=\frac{1+\Or(\delta)}{2\pi\mathrm{i}}\int_{e^{\theta\mathrm{i}}\delta t^{1/3},\text{ right of 0}}^{-e^{\theta\mathrm{i}}\delta t^{1/3}}\mathrm{d} Z\,e^{Z^3/3+rZ^2-sZ}\frac{2}{Z},
\end{aligned}\end{equation}
where $\widetilde{\Gamma}$ is a small deformation of $\Gamma$ such that it passes the real line to the right of $-1$.
Letting $t\to\infty$ and $\delta\to0$ we express the $1/Z$ term as an integral and recognize the contour integral representation of the Airy function to arrive at the desired limit
\begin{equation}
 -2e^{\frac{2}{3}r_2^3+r_2s_2}\int_0^\infty\mathrm{d} x\,\Ai(r_2^2+s_2+x)e^{r_2 x}.
\end{equation}
\end{proof}

\begin{prop}\label{propKsfBound}
For fixed $r_1,r_2,L$ there exists ${t}_0>0$ such that the estimate
\begin{equation}
	\left|\mathcal{K}_\stf^\textrm{resc}(r_1,s_1;r_2,s_2)\right|\leq\const\cdot e^{-s_1}
\end{equation}
holds for any ${t}>{t}_0$ and $s_1,s_2>0$.
\end{prop}

\begin{proof}
 The bound on the first part of the kernel is already established by Proposition~\ref{propK1Bound}. We decompose the product as in the proof of the pointwise convergence:
  \begin{equation}\begin{aligned}
  |t^{1/3}e^{\xi_1-t/2-1}\Psi^{n_1}_{{n_1}-1}(\xi_1)|&=|\alpha_t(r_1-t^{-2/3}/2,s_1)|\leq c_L e^{-s_1},
 \end{aligned}\end{equation}
as shown in Lemma~\ref{lemAlphaBound}. Recalling \eqref{eqPhi2}, we are finished if we show that $\widehat{\Phi}_{(1)}$ has a $s_2$-independent upper bound.

Start by
\begin{equation}
 |e^{-\xi_2+t/2+1}\widehat{\Phi}^n_{(1)}(\xi_1)|\leq \frac{{t}^{1/3}}{2\pi}\oint_{\widetilde{\Gamma}}|\mathrm{d} z|\,e^{\Re\,G(z,r_2,s_2)}\left|\frac{(1+z)t^{-1/3}}{1+ze^{1+z}}\right|.
\end{equation}
We can now use $e^{\Re(-t^{1/3}s_2(z+1))}\leq1$ since $s_2\geq0$ and $z$ stays to the right of $-1$. Continuing with the steep descent analysis as in the proof of the convergence we arrive at the absolute value analogue of equation \eqref{eq6.3.32},
\begin{equation}
 |e^{-\xi_2+t/2+1}\widehat{\Phi}^n_{(1)}(\xi_1)|\leq \const \int_{e^{\theta\mathrm{i}}\infty,\text{ right of 0}}^{-e^{\theta\mathrm{i}}\infty}|\mathrm{d} Z|\,e^{\Re(Z^3/3+rZ^2)}\left|\frac{2}{Z}\right|,
\end{equation}
which is finite.

\end{proof}

%\chapter{Generalizations of the results}\label{secGen}

\section{Attractiveness and a more general class of initial data}\label{secAttr}
We can relax the strict assumptions on the initial conditions by recognizing that our model shows \emph{attractiveness}. A stochastic particle system is called attractive, if for two distinct initial configurations evolving under the same noise their order is preserved.

\begin{prop}\label{propBoundMod}
Let us consider two admissible initial conditions, denoted by $\vec{a}\in\R^\Z$, $\vec{b}\in\R^\Z$. Under the same
noise they evolve to $x_m^a(t)$ and $x_m^b(t)$. If there is $M>0$ such
that $|a_m-b_m|\leq M$ for all $m\in\mathbb{Z}$, then also
\begin{equation}
\bigl|x_m^a(t)-x_m^b(t)\bigr|\leq M\qquad
\forall m\in\mathbb{Z}, t>0.
\end{equation}
\end{prop}

The same property holds for the standard coupling of the TASEP, as
explained in Section~2.1 of~\cite{BFS07}.

As an immediate consequence, the limit result of Theorem~\ref{thmAsympFixedTime} holds for bounded modifications of the initial condition $\vec{x}(0)=\vec{\zeta}^\textrm{\,flat}$, since an error of size $M$ vanishes under the $t^{1/3}$ scaling. For example, one could choose a unit cell of length 1 and take an arbitrary initial condition with the only restriction that there are $\ell$ particles in each cell. Then the convergence to the Airy$_1$ process holds.

Furthermore, this Proposition allows us to choose $\zeta_0=0$ in Theorem~\ref{thmAsymp0} instead of considering the true Poisson case, since $\zeta_0$ gives rise to a global shift, that is bounded with probability $1$.

\begin{proof}[Proof of Proposition~\ref{propBoundMod}]
By definition,
%
%e2.20 #&#
\begin{eqnarray}
 x_m^a(t)&=&-\max
_{ k\leq m}\bigl\{Y_{k,m}(t)-a_k\bigr\},
\nonumber
\\[-8pt]
\\[-8pt]
\nonumber
x_m^b(t)&=&-\max_{ k\leq m}\bigl
\{Y_{k,m}(t)-b_k\bigr\}.
\end{eqnarray}
Since the inequality
%
%e2.21 #&#
\begin{equation}
Y_{k,m}(t)-a_k\leq Y_{k,m}(t)-b_k+M
\end{equation}
holds for each $k$, the maximum can be taken on each side, resulting in
%
%e2.22 #&#
\begin{eqnarray}
\max_{ k\leq m}\bigl
\{Y_{k,m}(t)-a_k\bigr\}&\leq&\max_{ k\leq m}
\bigl\{ Y_{k,m}(t)-b_k\bigr\} + M,
\nonumber
\\[-8pt]
\\[-8pt]
\nonumber
x_m^a(t)&\geq& x_m^b(t)-M.
\end{eqnarray}
Correspondingly, one has $x_m^b(t)\geq x_m^a(t)-M$.
\end{proof}

\section{Asymptotics along space-like paths and slow decorrelations}\label{SectTagged}
The rescaled process at fixed time is not the only one in which the Airy limit processes appears. It is also the case for the joint distributions of the positions of a tagged Brownian motion at different times, which means correlations along the $t$ direction. This is the content of Theorem~\ref{thmTagged} below. It is a consequence of a phenomenon shared by many models in the KPZ universality class, which is called \emph{slow decorrelations}~\cite{Fer08,CFP10b}. This means that in the time-like direction the correlation length is not of order $t^{2/3}$ but $t$, as proven in Proposition~\ref{PropSlowDecorrelation}. Thus instead of evaluating the distribution along points with fixed $t$ only, we can shift the points in the time-like direction up to some $t^\nu$ with $\nu<1$, and still keep the same limit result. These statistics on space-like paths are proven in Theorem~\ref{ThmAsymptSpaceTime}.

\begin{thm}\label{thmTagged}
For each one of the initial conditions
\begin{equation}
 \vec{\zeta}^{*}\in\{\vec{\zeta}^\textrm{packed},\vec{\zeta}^\textrm{flat},\vec{\zeta}^\textrm{stat},\vec{\zeta}^\shf,\vec{\zeta}^\shst,\vec{\zeta}^\stf\},
\end{equation}
define a rescaled process
\begin{equation}
\widetilde{X}_t^{*}\in\{\widetilde{X}_t^\textrm{packed},\widetilde{X}_t^\textrm{flat},\widetilde{X}_t^\textrm{stat},\widetilde{X}_t^\hf,\widetilde{X}_t^\hst,\widetilde{X}_t^\stf\}
\end{equation}
by
\begin{equation}
\tau\mapsto \widetilde{X}^{*}_t(\tau) := t^{-1/3}\left(x_{\lfloor t\rfloor}(t+2\tau t^{2/3})-2t- 2\tau t^{2/3}\right),
\end{equation}
where the process $\vec{x}(t)$ on the right hand side is subject to \mbox{$\vec{x}(0)=\vec{\zeta}^{*}$}.

In the large time limit,
\begin{equation}\begin{aligned}
	\lim_{t\to\infty} \widetilde{X}_t^\textrm{packed}(\tau) &= \mathcal{A}_2(\tau),\\
	\lim_{t\to\infty} \widetilde{X}_t^\textrm{flat}(\tau) &= 2^{1/3}\mathcal{A}_1(2^{-2/3}\tau),\\
	\lim_{t\to\infty} \widetilde{X}_t^\textrm{stat}(\tau) &= \mathcal{A}_\mathrm{stat}(\tau),\\
	\lim_{t\to\infty} \widetilde{X}_t^\hf(\tau) &= \mathcal{A}_{2\to1}(\tau),\\
	\lim_{t\to\infty} \widetilde{X}_t^\hst(\tau) &= \mathcal{A}_{2\to\textrm{BM}}(\tau),\\
	\lim_{t\to\infty} \widetilde{X}_t^\stf(\tau) &= \mathcal{A}_\mathrm{BM\to1}(\tau),
\end{aligned}\end{equation}
hold in the sense of finite-dimensional distributions.
\end{thm}
%This theorem is proven in Section~\ref{SectTagged}. It is a special case of the more general statement of Theorem~\ref{ThmAsymptSpaceTime} in Section~\ref{SectTagged}. The result is based from the fixed time result, Theorem~\ref{thmAsympFixedTime}, and a slow decorrelation result, Proposition~\ref{PropSlowDecorrelation}. The latter says that along special space-time directions the decorrelation happens over a macroscopic time span.

Theorem~\ref{thmTagged} is a corollary of the following result on space-like paths:

\begin{thm}\label{ThmAsymptSpaceTime}
Let $\vec{x}(t)$ be the system of one-sided reflected Brownian motions subject to some initial condition $\vec{x}(0)=\vec{\zeta}$. Let us fix a $\nu\in [0,1)$, choose any $\theta_1,\ldots,\theta_m\in [-t^{\nu},t^{\nu}]$, $r_1,\ldots,r_m\in\R$ and define the rescaled random variables
\begin{equation}\label{eq5.1}
\widehat{X}_t(r_k,\theta_k):=t^{-1/3}\left(x_{[t+2r_k  t^{2/3}+\theta_k]}(t+\theta_k)-2t-2\theta_k-2r_k  t^{2/3}\right).
\end{equation}
Let $s_1,\ldots,s_m\in\R$ and $\widehat{X}_t$ satisfy the limit
\begin{equation}\label{eqA*limit}
\lim_{t\to\infty}\Pb\left(\bigcap_{k=1}^m \left\{\widehat{X}_t(r_k,0)\leq s_k\right\}\right) = \Pb\left(\bigcap_{k=1}^m \left\{{\cal A}_*(r_k)\leq s_k\right\}\right),
\end{equation}
for some process ${\cal A}_*(r)$ whose joint distribution function is continuous.

Then it holds
\begin{equation}
\lim_{t\to\infty}\Pb\left(\bigcap_{k=1}^m \left\{\widehat{X}_t(r_k,\theta_k)\leq s_k\right\}\right) = \Pb\left(\bigcap_{k=1}^m \left\{{\cal A}_*(r_k)\leq s_k\right\}\right).
\end{equation}
\end{thm}
Notice that specifying \eqref{eq5.1} to $\theta_k=0$ gives exactly the familiar scaling where all our asymptotic results hold. Thus this theorem is applicable to all six of the treated initial conditions.
\begin{proof}[Proof of Theorem~\ref{thmTagged}]
This follows by taking $\theta_k=2\tau_kt^{2/3}$ and \mbox{$r_k=-\tau_k$}  in Theorem~\ref{ThmAsymptSpaceTime}.
\end{proof}

For the proof of Theorem~\ref{ThmAsymptSpaceTime} we need this slow decorrelation property:
\begin{prop}\label{PropSlowDecorrelation}
Let $\widehat{X}_t(r,\theta)$ be defined as in Theorem~\ref{ThmAsymptSpaceTime} and also satisfy \eqref{eqA*limit}. For a $\nu \in [0,1)$, consider $\theta\in [-t^\nu,t^\nu]$ and some $r\in\R$. Then, for any $\e>0$,
\begin{equation}
\lim_{t\to\infty}\Pb\left( |\widehat{X}_t(r,\theta)-\widehat{X}_t(r,0)|\geq \e\right)=0.
\end{equation}
\end{prop}
\begin{proof}
Without loss of generality we consider $\theta\geq 0$. For $\theta<0$ one just have to denote $\tilde t=t+\theta$ so that $\tilde t-\theta=t$ and the proof remains valid with $t$ replaced by $\tilde t$. Recall that by definition we have
\begin{equation}\label{eqSlow3}
x_m(t)=\max_{k\leq m}\{Y_{k,m}(t)+\zeta_k\}.
\end{equation}
First we need an inequality, namely
\begin{equation}\label{eqCoupling}
\begin{aligned}
	x_{n+\theta}(t+\theta)&=\max_{k\leq n+\theta}\{\zeta_k+Y_{k,n+\theta}(t+\theta)\}\geq\max_{k\leq n}\{\zeta_k+Y_{k,n+\theta}(t+\theta)\}\\
	&=\max_{k\leq n}\{\zeta_k+\sup_{0\leq s_{k}\leq \ldots\leq s_{n+\theta}=t+\theta}\sum_{i=k}^{n+\theta} (B_i(s_{i})-B_i(s_{i-1}))\}\\
	&\geq\max_{k\leq n}\{\zeta_k+\sup_{\begin{subarray}{c}0\leq s_{k}\leq\ldots\leq s_{n+\theta}=t+\theta\\ \textrm{with }s_{n}=t\end{subarray}}\sum_{i=k}^{n+\theta} (B_i(s_{i})-B_i(s_{i-1}))\}\\
	&=x_n(t) + Y_{n+1,n+\theta}(t,t+\theta),
\end{aligned}\end{equation}
with
\begin{equation}
Y_{n+1,n+\theta}(t,t+\theta)=\sup_{t\leq s_{n+1}\leq \ldots\leq s_{n+\theta}= t+\theta}\sum_{i=n+1}^{n+\theta} (B_i(s_{i})-B_i(s_{i-1})).
\end{equation}
Remark that $x_n(t)$ and $Y_{n+1,n+\theta}(t,t+\theta)$ are independent. Specifying this to $n=t+2rt^{2/3}$, the inequality can be rewritten as
\begin{equation}\label{eqIneqRewr}
 \widehat{X}_t(r,\theta)\geq \widehat{X}_t(r,0)+\chi(t),
\end{equation}
where
\begin{equation}
 \chi(t)=\frac{Y_{n+1,n+\theta}(t,t+\theta)-2\theta}{t^{1/3}}\stackrel{d}{=}\frac{Y_{1,\theta}(\theta)-2\theta}{t^{1/3}}.
\end{equation}
From Theorem~\ref{thmAsympStep} we know that $(t/\theta)^{1/3}\chi(t)$ converges to a Tracy-Widom distributed random variable, which means that $\chi(t)$ itself converges to $0$ in distribution. By \eqref{eqA*limit}, we have
\begin{equation}
\widehat{X}_t(r,0)\stackrel{D}{\Longrightarrow}{\cal A}_*(r),
\end{equation}
and also
\begin{equation}
\widehat{X}_t(r,\theta)\stackrel{D}{\Longrightarrow}{\cal A}_*(r),
\end{equation}
which follows from $\widehat{X}_t(r,\theta)=\widehat{X}_{t+\theta}(r+\Or(t^{\nu-1)}),0)$.

Since both sides of \eqref{eqIneqRewr} converge in distribution to ${\cal A}_*(r)$, by Lemma~4.1 of~\cite{BC09} (reported below) we have  $\widehat{X}_t(r,\theta)-\widehat{X}_t(r,0)-\chi(t)\to 0$ in probability as $t\to\infty$. As $\chi(t)\to 0$ in probability, too, the proof is finished.
\end{proof}

\begin{lem}[Lemma~4.1 of~\cite{BC09}]\label{BAC_lemma}
Consider two sequences of random variables $\{X_n\}$ and $\{\tilde{X}_n\}$ such that for each $n$,
$X_n$ and $\tilde{X}_n$ are defined on the same probability space $\Omega_n$. If $X_n\geq\tilde{X}_n$ and $X_n\Rightarrow D$
as well as $\tilde{X}_n\Rightarrow D$ then $X_n-\tilde{X}_n$ converges to zero in probability.
Conversely if $\tilde{X}_n\Rightarrow D$ and $X_n-\tilde{X}_n$ converges to zero in probability then $X_n\Rightarrow D$ as well.
\end{lem}

Finally we come to the proof of Theorem~\ref{ThmAsymptSpaceTime}.
\begin{proof}[Proof of Theorem~\ref{ThmAsymptSpaceTime}]
Let us define the random variables
\begin{equation}
\Xi_k:=\widehat{X}_t(r_k,\theta_k)-\widehat{X}_t(r_k,0).
\end{equation}
such that
\begin{equation}\label{eq5.111}
\Pb\bigg(\bigcap_{k=1}^m\{\widehat{X}_t(r_k,\theta_k)\leq s_k\}\bigg) = \Pb\bigg(\bigcap_{k=1}^m\{\widehat{X}_t(r_k,0)+\Xi_k\leq s_k\}\bigg).
\end{equation}
The slow decorrelation result (Proposition~\ref{PropSlowDecorrelation}) implies $\Xi_k\to 0$ in probability as $t\to\infty$. Introducing $\e>0$ we can use inclusion-exclusion to decompose
\begin{equation}\begin{aligned}
	\eqref{eq5.111}&=\Pb\bigg(\bigcap_{k=1}^m\{\widehat{X}_t(r_k,0)+\Xi_k\leq s_k\}\cap\{|\Xi_k|\leq \e\}\bigg)+\sum_j \Pb\left(R_j\right).
\end{aligned}\end{equation}
The sum on the right hand side is finite and each $R_j$ satisfies \mbox{$R_j\subset\{|\Xi_k|> \e\}$} for at least one $k$, implying $\lim_{t\to\infty}\Pb(R_j)=0$.
Using \eqref{eqA*limit} leads to
\begin{equation}\begin{aligned}
	\limsup_{t\to\infty}\Pb\bigg(\bigcap_{k=1}^m\{\widehat{X}_t(r_k,\theta_k)\leq s_k\}\bigg) &\leq \Pb\bigg(\bigcap_{k=1}^m\{{\cal A}_*(r_k)\leq s_k+\e\}\bigg),\\
	\liminf_{t\to\infty}\Pb\bigg(\bigcap_{k=1}^m\{\widehat{X}_t(r_k,\theta_k)\leq s_k\}\bigg) &\geq \Pb\bigg(\bigcap_{k=1}^m\{{\cal A}_*(r_k)\leq s_k-\e\}\bigg).
\end{aligned}\end{equation}
\newpage
Finally, the continuity assumption allows us to take the limit $\e\to 0$ and obtain
\begin{equation}
\lim_{t\to\infty}\Pb\bigg(\bigcap_{k=1}^m\{\widehat{X}_t(r_k,\theta_k)\leq s_k\}\bigg) = \Pb\bigg(\bigcap_{k=1}^m\{{\cal A}_*(r_k)\leq s_k\}\bigg).
\end{equation}
\end{proof}

%\appendix
\backmatter
%\chapter*{Appendix  [Preliminary]}\hspace{-19.5pt}
%\chapter{Airy function identities  [Preliminary]}
%\begin{equation}
% \int_\R \mathrm{d} x\, \Ai(x+s)e^{r x}=e^{\frac{1}{3}r^3-rs}
%\end{equation}
%\addcontentsline{toc}{section}{References}
\bibliographystyle{plain}
\bibliography{Biblio}

\begin{thebibliography}{100}

\bibitem{ANvM10b}
M.~Adler, E.~Nordenstam, and P.~van Moerbeke.
\newblock {Consecutive minors for Dyson's Brownian motions}.
\newblock {\em Stochastic Processes and their Applications}, 124:2023 -- 2051,
  2014.

\bibitem{ANvM10}
M.~Adler, E.~Nordenstam, and P.~van Moerbeke.
\newblock {The Dyson Brownian minor process}.
\newblock {\em Annales de l'Institut Fourier}, 64:971--1009, 2014.

\bibitem{AvM03}
M.~Adler and P.~van Moerbeke.
\newblock {PDE}'s for the joint distribution of the {D}yson, {A}iry and {S}ine
  processes.
\newblock {\em Ann. Probab.}, 33:1326--1361, 2005.

\bibitem{ACQ10}
G.~Amir, I.~Corwin, and J.~Quastel.
\newblock {Probability distribution of the free energy of the continuum
  directed random polymer in 1+1 dimensions}.
\newblock {\em Comm. Pure Appl. Math.}, 64:466--537, 2011.

\bibitem{AO76}
R.F. Anderson and S.~Orey.
\newblock Small random perturbation of dynamical systems with reflecting
  boundary.
\newblock {\em Nagoya Math. J.}, 60:189--216, 1976.

\bibitem{BBdF08}
J.~Baik, R.~Buckingham, and J.~DiFranco.
\newblock {Asymptotics of Tracy-Widom distributions and the total integral of a
  Painleve II function}.
\newblock {\em Comm. Math. Phys.}, 280:463--497, 2008.

\bibitem{BFP09}
J.~Baik, P.L. Ferrari, and S.~P{\'e}ch{\'e}.
\newblock {Limit process of stationary TASEP near the characteristic line}.
\newblock {\em Comm. Pure Appl. Math.}, 63:1017--1070, 2010.

\bibitem{BR00}
J.~Baik and E.M. Rains.
\newblock Limiting distributions for a polynuclear growth model with external
  sources.
\newblock {\em J. Stat. Phys.}, 100:523--542, 2000.

\bibitem{BW10}
J.~{Baik} and D.~{Wang}.
\newblock {On the largest eigenvalue of a Hermitian random matrix model with
  spiked external source I. Rank one case}.
\newblock {\em Int. Math. Res. Notices}, 2011:5164--5240, 2011.

\bibitem{BJ00}
T.C. Banwell and A.~Jayakumar.
\newblock Exact analytical solution for current flow through diode with series
  resistance.
\newblock {\em Electronics Letters}, 36:291--292, 2000.

\bibitem{BaSt95}
A.L. Barabasi and H.E. Stanley.
\newblock {\em {Fractal Concepts in Surface Growth}}.
\newblock Cambridge University Press, 1995.

\bibitem{BPLPCS00}
D.A. Barry, J.-Y. Parlange, L.~Li, H.~Prommer, C.J. Cunningham, and
  F.~Stagnitti.
\newblock Analytical approximations for real values of the {L}ambert
  {W}-function.
\newblock {\em Mathematics and Computers in Simulation}, 53:95--103, 2000.

\bibitem{Bar01}
Y.~Baryshnikov.
\newblock {GUEs and queues}.
\newblock {\em Probab. Theory Relat. Fields}, 119:256--274, 2001.

\bibitem{Batch}
M.~T. {Batchelor}.
\newblock {The Bethe ansatz after 75 years}.
\newblock {\em Physics Today}, 60:36, 2007.

\bibitem{BC09}
G.~{Ben Arous} and I.~Corwin.
\newblock {Current fluctuations for TASEP: a proof of the Pr\"ahofer-Spohn
  conjecture}.
\newblock {\em Ann. Probab.}, 39:104--138, 2011.

\bibitem{Bet31}
H.~Bethe.
\newblock {Zur Theorie der Metalle. I. Eigenwerte und Eigenfunktionen der
  linearen Atomkette}.
\newblock {\em Zeitschrift f{\"u}r Physik}, 71:205--226, 1931.

\bibitem{BFP08}
F.~Bornemann, P.L. Ferrari, and M.~Pr{\"a}hofer.
\newblock {The Airy$_1$ process is not the limit of the largest eigenvalue in
  GOE matrix diffusion}.
\newblock {\em J. Stat. Phys.}, 133:405--415, 2008.

\bibitem{BC11}
A.~Borodin and I.~Corwin.
\newblock Macdonald processes.
\newblock {\em Probab. Theory Relat. Fields (online first)}, 2013.

\bibitem{BCFV14}
A.~{Borodin}, I.~{Corwin}, P.L. {Ferrari}, and B.~{Vet{\H o}}.
\newblock {Height fluctuations for the stationary KPZ equation}.
\newblock {\em Mathematical Physics, Analysis and Geometry}, 18:1--95, 2015.

\bibitem{BCR13}
A.~Borodin, I.~Corwin, and D.~Remenik.
\newblock {Multiplicative functionals on ensembles of non-intersecting paths}.
\newblock {\em Ann. Inst. H. Poincaré Probab. Statist.}, 51:28--58, 2015.

\bibitem{BF07}
A.~Borodin and P.L. Ferrari.
\newblock {Large time asymptotics of growth models on space-like paths I:
  PushASEP}.
\newblock {\em Electron. J. Probab.}, 13:1380--1418, 2008.

\bibitem{BFPS06}
A.~Borodin, P.L. Ferrari, M.~Pr{\"a}hofer, and T.~Sasamoto.
\newblock {Fluctuation Properties of the TASEP with Periodic Initial
  Configuration}.
\newblock {\em J. Stat. Phys.}, 129:1055--1080, 2007.

\bibitem{BFS07}
A.~Borodin, P.L. Ferrari, and T.~Sasamoto.
\newblock {Transition between Airy$_1$ and Airy$_2$ processes and TASEP
  fluctuations}.
\newblock {\em Comm. Pure Appl. Math.}, 61:1603--1629, 2008.

\bibitem{BFS09}
A.~Borodin, P.L. Ferrari, and T.~Sasamoto.
\newblock {Two speed TASEP}.
\newblock {\em J. Stat. Phys.}, 137:936--977, 2009.

\bibitem{BoGo12}
A.~Borodin and V.~Gorin.
\newblock {Lectures on integrable probability}.
\newblock {\em arXiv:1212.3351}, 2012.

\bibitem{BoPe14}
A.~Borodin and L.~Petrov.
\newblock {Integrable probability: from representation theory to Macdonald
  processes}.
\newblock {\em Probab. Surv.}, 11:1--58, 2014.

\bibitem{BS96}
A.N. Borodin and P.~Salminen.
\newblock {\em Handbook of {B}rownian Motion - Facts and Formulae}.
\newblock Birkh{\"a}user.

\bibitem{CDR10}
P.~Calabrese, P.~Le Doussal, and A.~Rosso.
\newblock {Free-energy distribution of the directed polymer at high
  temperature}.
\newblock {\em EPL}, 90:20002, 2010.

\bibitem{CY92}
C.-C. Chang and H.-T. Yau.
\newblock {Fluctuations of one dimensional Ginzburg-Landau models in
  nonequilibrium}.
\newblock {\em Comm. Math. Phys.}, 145:209--234, 1992.

\bibitem{Che02}
Y.~Chen and K.L. Moore.
\newblock Analytical stability bound for delayed second-order systems with
  repeating poles using {L}ambert function {W}.
\newblock {\em Automatica}, 38:891--895, 2002.

\bibitem{CGHJ93}
R.M. Corless, G.H. Gonnet, D.E.G. Hare, and D.J. Jeffrey.
\newblock Lambert's {W} {F}unction in {M}aple.
\newblock {\em The {M}aple {T}echnical {N}ewsletter}, 9:12--22, 1993.

\bibitem{CGHJK96}
R.M. Corless, G.H. Gonnet, D.E.G. Hare, D.J. Jeffrey, and D.E. Knuth.
\newblock {On the Lambert W function}.
\newblock {\em Advances in Computational Mathematics}, 5:329--359, 1996.

\bibitem{CJK97}
R.M. Corless, D.J. Jeffrey, and D.E. Knuth.
\newblock A sequence of series for the {L}ambert {W} function.
\newblock In {\em Proceedings of the 1997 International Symposium on Symbolic
  and Algebraic Computation}, pages 197--204, 1997.

\bibitem{CJV00}
R.M. Corless, D.J. Jeffrey, and S.R. Valluri.
\newblock Some applications of the {L}ambert {W} function to physics.
\newblock {\em Canadian Journal of Physics}, 78:823--831, 2000.

\bibitem{CorwinRev}
I.~Corwin.
\newblock {The Kardar-Parisi-Zhang equation and universality class}.
\newblock {\em Random Matrices Theory Appl.}, 1:1130001, 2012.

\bibitem{CFP10b}
I.~Corwin, P.L. Ferrari, and S.~P{\'e}ch{\'e}.
\newblock {Universality of slow decorrelation in KPZ models}.
\newblock {\em Ann. Inst. H. Poincar\'e Probab. Statist.}, 48:134--150, 2012.

\bibitem{CH11}
I.~{Corwin} and A.~{Hammond}.
\newblock {Brownian Gibbs property for Airy line ensembles}.
\newblock {\em Invent. Math.}, 195:441--508, 2014.

\bibitem{CLW14}
I.~{Corwin}, Z.~{Liu}, and D.~{Wang}.
\newblock {Fluctuations of TASEP and LPP with general initial data}.
\newblock {\em arXiv:1412.5087, to appear in Ann. Appl. Probab.}, 2014.

\bibitem{CQR11}
I.~Corwin, J.~Quastel, and D.~Remenik.
\newblock {Continuum statistics of the Airy$_2$ process}.
\newblock {\em Comm. Math. Phys.}, 317:347--362, 2013.

\bibitem{CQR15}
I.~Corwin, J.~Quastel, and D.~Remenik.
\newblock Renormalization fixed point of the {KPZ} universality class.
\newblock {\em Journal of Statistical Physics}, 160:815--834, 2015.

\bibitem{Dot10}
V.~Dotsenko.
\newblock {Replica Bethe ansatz derivation of the Tracy-Widom distribution of
  the free energy fluctuations in one-dimensional directed polymers}.
\newblock {\em J. Stat. Mech.}, P07010, 2010.

\bibitem{Dys62}
F.J. Dyson.
\newblock A {B}rownian-motion model for the eigenvalues of a random matrix.
\newblock {\em J. Math. Phys.}, 3:1191--1198, 1962.

\bibitem{EM97}
B.~Eynard and M.L. Mehta.
\newblock Matrices coupled in a chain. {I}. {E}igenvalue correlations.
\newblock {\em J. Phys. A}, 31:4449--4456, 1998.

\bibitem{Fer08}
P.L. Ferrari.
\newblock {Slow decorrelations in KPZ growth}.
\newblock {\em J. Stat. Mech.}, P07022, 2008.

\bibitem{FF10}
P.L. Ferrari and R.~Frings.
\newblock {On the partial connection between random matrices and interacting
  particle systems}.
\newblock {\em J. Stat. Phys.}, 141:613--637, 2010.

\bibitem{FF13}
P.L. {Ferrari} and R.~{Frings}.
\newblock {Perturbed GUE minor process and Warren's process with drifts}.
\newblock {\em J. Stat. Phys.}, 154:356--377, 2014.

\bibitem{FS03}
P.L. Ferrari and H.~Spohn.
\newblock Step fluctations for a faceted crystal.
\newblock {\em J. Stat. Phys.}, 113:1--46, 2003.

\bibitem{FS05b}
P.L. Ferrari and H.~Spohn.
\newblock {A determinantal formula for the GOE Tracy-Widom distribution}.
\newblock {\em J. Phys. A}, 38:L557--L561, 2005.

\bibitem{FS05a}
P.L. Ferrari and H.~Spohn.
\newblock Scaling limit for the space-time covariance of the stationary totally
  asymmetric simple exclusion process.
\newblock {\em Comm. Math. Phys.}, 265:1--44, 2006.

\bibitem{FS11}
P.L. Ferrari and H.~Spohn.
\newblock Random growth models.
\newblock In Baik Akemann and Di~Francesco, editors, {\em The Oxford Handbook
  of Random Matrix Theory}. 2011.

\bibitem{FSW13}
P.L. Ferrari, H.~Spohn, and T.~Weiss.
\newblock Scaling limit for {B}rownian motions with one-sided collisions.
\newblock {\em Ann. Appl. Probab.}, 25:1349--1382, 2015.

\bibitem{FNH99}
P.J. Forrester, T.~Nagao, and G.~Honner.
\newblock Correlations for the orthogonal-unitary and symplectic-unitary
  transitions at the hard and soft edges.
\newblock {\em Nucl. Phys. B}, 553:601--643, 1999.

\bibitem{GlWh91}
P.W. Glynn and W.~Whitt.
\newblock Departures from many queues in series.
\newblock {\em Ann. Appl. Probab.}, 1:546--572, 1991.

\bibitem{GTW00}
J.~Gravner, C.A. Tracy, and H.~Widom.
\newblock Limit theorems for height fluctuations in a class of discrete space
  and time growth models.
\newblock {\em J. Stat. Phys.}, 102:1085--1132, 2001.

\bibitem{GP15}
M.~{Gubinelli} and N.~{Perkowski}.
\newblock {KPZ reloaded}.
\newblock {\em arXiv:1508.03877}, 2015.

\bibitem{GPV88}
M.Z. Guo, G.C. Papanicolaou, and S.R.S. Varadhan.
\newblock {Nonlinear diffusion limit for a system with nearest neighbor
  interactions}.
\newblock {\em Commun. Math. Phys.}, 118:31--59, 1988.

\bibitem{GS92}
L-H. Gwa and H.~Spohn.
\newblock Bethe solution for the dynamical-scaling exponent of the noisy
  {B}urgers equation.
\newblock {\em Phys. Rev. A}, 46:844--854, 1992.

\bibitem{Hai11}
M.~Hairer.
\newblock {Solving the KPZ equation}.
\newblock {\em Ann. Math.}, 178:559--664, 2013.

\bibitem{HaKa15}
T.~Halpin-Healy and K.~Takeuchi.
\newblock {A KPZ cocktail-shaken, not stirred: Toasting 30 years of kinetically
  roughened surfaces}.
\newblock {\em J. Stat. Phys.}, 160:794--814, 2015.

\bibitem{HH95}
T.~Halpin-Healy and Y.-C. Zhang.
\newblock {Kinetic roughening phenomena, stochastic growth, directed polymers
  and all that}.
\newblock {\em Physics Reports}, 254:215--414, 1995.

\bibitem{Har65}
T.E. Harris.
\newblock Diffusion with ``collisions'' between particles.
\newblock {\em J. Appl. Probab.}, 2:323--338, 1965.

\bibitem{SI04}
T.~Imamura and T.~Sasamoto.
\newblock Fluctuations of the one-dimensional polynuclear growth model with
  external sources.
\newblock {\em Nucl. Phys. B}, 699:503--544, 2004.

\bibitem{IS13}
T.~Imamura and T.~Sasamoto.
\newblock {Stationary correlations for the 1D KPZ equation}.
\newblock {\em J. Stat. Phys.}, 150:908--939, 2013.

\bibitem{JK04}
A.~Jain and A.~Kapoor.
\newblock Exact analytical solutions of the parameters of real solar cells
  using {L}ambert {W}-function.
\newblock {\em Solar Energy Materials and Solar Cells}, 81:269--277, 2004.

\bibitem{Jo00b}
K.~Johansson.
\newblock Shape fluctuations and random matrices.
\newblock {\em Comm. Math. Phys.}, 209:437--476, 2000.

\bibitem{Jo03b}
K.~Johansson.
\newblock Discrete polynuclear growth and determinantal processes.
\newblock {\em Comm. Math. Phys.}, 242:277--329, 2003.

\bibitem{Jo05}
K.~Johansson.
\newblock Random matrices and determinantal processes.
\newblock In {\em Mathematical Statistical Physics, Session LXXXIII: Lecture
  Notes of the Les Houches Summer School 2005}, pages 1--56. 2006.

\bibitem{Jo15}
K.~Johansson.
\newblock {Two time distribution in Brownian directed percolation}.
\newblock {\em Comm. Math. Phys.}, online first:1--52, 2016.

\bibitem{JN06}
K.~Johansson and E.~Nordenstam.
\newblock {Eigenvalues of GUE minors}.
\newblock {\em Electron. J. Probab.}, 11:1342--1371, 2006.

\bibitem{KPS12}
I.~Karatzas, S.~Pal, and M.~Shkolnikov.
\newblock {Systems of Brownian particles with asymmetric collisions}.
\newblock {\em Ann. Inst. H. Poincaré Probab. Statist.}, 52:323--354, 2016.

\bibitem{Kar87}
M.~{Kardar}.
\newblock {Replica Bethe ansatz studies of two-dimensional interfaces with
  quenched random impurities}.
\newblock {\em Nuclear Physics B}, 290:582--602, 1987.

\bibitem{KPZ86}
M.~Kardar, G.~Parisi, and Y.Z. Zhang.
\newblock Dynamic scaling of growing interfaces.
\newblock {\em Phys. Rev. Lett.}, 56:889--892, 1986.

\bibitem{Kr97}
J.~Krug.
\newblock {Origins of scale invariance in growth processes}.
\newblock {\em Adv. Phys.}, 46:139--282, 1997.

\bibitem{LL63}
E.H. Lieb and W.~Liniger.
\newblock Exact analysis of an interacting {B}ose gas. {I}. the general
  solution and the ground state.
\newblock {\em Phys. Rev.}, 130:1605--1616, 1963.

\bibitem{Me98}
P.~Meakin.
\newblock {\em {Fractals, Scaling and Growth Far From Equilibrium}}.
\newblock Cambridge University Press, 1998.

\bibitem{OCY01}
N.~O'Connell and M.~Yor.
\newblock {Brownian analogues of Burke's theorem}.
\newblock {\em Stoch. Proc. Appl.}, 96:285--304, 2001.

\bibitem{PS02}
M.~Pr{\"a}hofer and H.~Spohn.
\newblock Scale invariance of the {PNG} droplet and the {A}iry process.
\newblock {\em J. Stat. Phys.}, 108:1071--1106, 2002.

\bibitem{PS02b}
M.~Pr{\"a}hofer and H.~Spohn.
\newblock Exact scaling function for one-dimensional stationary {KPZ} growth.
\newblock {\em J. Stat. Phys.}, 115:255--279, 2004.

\bibitem{QRev}
J.~Quastel.
\newblock {Introduction to KPZ}.
\newblock {\em Current Developments in Mathematics}, Vol. 2011:125--194, 2011.

\bibitem{QR12}
J.~Quastel and D.~Remenik.
\newblock {Local behavior and hitting probabilities of the Airy$_1$ process}.
\newblock {\em Prob. Theory Relat. Fields}, 157:605--634, 2013.

\bibitem{QR13}
J.~{Quastel} and D.~{Remenik}.
\newblock {Airy processes and variational problems}.
\newblock {\em Springer Proceedings in Mathematics \& Statistics}, 69:121--171,
  2014.

\bibitem{QuSp15}
J.~Quastel and H.~Spohn.
\newblock {The one-dimensional KPZ equation and its universality class}.
\newblock {\em J. Stat. Phys.}, 160:965--984, 2015.

\bibitem{Sas05}
T.~Sasamoto.
\newblock Spatial correlations of the {1D KPZ} surface on a flat substrate.
\newblock {\em J. Phys. A}, 38:L549--L556, 2005.

\bibitem{SS10b}
T.~Sasamoto and H.~Spohn.
\newblock {Exact height distributions for the KPZ equation with narrow wedge
  initial condition}.
\newblock {\em Nucl. Phys. B}, 834:523--542, 2010.

\bibitem{SaSp11a}
T.~Sasamoto and H.~Spohn.
\newblock {The 1+1-dimensional Kardar-Parisi-Zhang equation and its
  universality class}.
\newblock {\em J. Stat. Mech.}, P01031, 2011.

\bibitem{SS15}
T.~Sasamoto and H.~Spohn.
\newblock Point-interacting {B}rownian motions in the {KPZ} universality class.
\newblock {\em Electron. J. Probab.}, 20(87), 2015.

\bibitem{SW98}
T.~Sasamoto and M.~Wadati.
\newblock {Determinantal form solution for the derivative nonlinear
  Schr{\"o}dinger type model}.
\newblock {\em J. Phys. Soc. Jpn}, 67:784--790, 1998.

\bibitem{Sch97}
G.M. Sch{\"u}tz.
\newblock Exact solution of the master equation for the asymmetric exclusion
  process.
\newblock {\em J. Stat. Phys.}, 88:427--445, 1997.

\bibitem{Sep97}
T.~Sepp{\"a}l{\"a}inen.
\newblock A scaling limit for queues in series.
\newblock {\em Ann. Appl. Probab.}, 7:855--872, 1997.

\bibitem{SV10}
T.~Sepp{\"a}l{\"a}inen and B.~Valk{\'o}.
\newblock Bounds for scaling exponents for a 1+1 dimensional directed polymer
  in a brownian environment.
\newblock {\em ALEA}, VII:451--476, 2010.

\bibitem{Sim00}
B.~Simon.
\newblock {\em Trace Ideals and Their Applications}.
\newblock American Mathematical Society, second edition, 2000.

\bibitem{Sim11}
J.~Simon, W.S. Bakr, R.~Ma, M.E. Tai, P.M. Preiss, and M.~Greiner.
\newblock Quantum simulation of antiferromagnetic spin chains in an optical
  lattice.
\newblock {\em Nature}, 472:307--312, 2011.

\bibitem{Sko61}
A.V. Skorokhod.
\newblock Stochastic equations for diffusions in a bounded region.
\newblock {\em Theory Probab. Appl.}, 6:264--274, 1961.

\bibitem{Sp06}
H.~Spohn.
\newblock {Exact solutions for KPZ-type growth processes, random matrices, and
  equilibrium shapes of crystals}.
\newblock {\em Physica A}, 369:71--99, 2006.

\bibitem{Sp16a}
H.~Spohn.
\newblock {The Kardar-Parisi-Zhang equation - a statistical physics
  perspective}.
\newblock In {\em Stochastic Processes and Random Matrices}, \'{E}cole
  d'\'{E}t\'{e} Physique, Les Houches, 2015. Oxford University Press.
\newblock arXiv:1601.00499.

\bibitem{Sut78}
B.~Sutherland.
\newblock A brief history of the quantum soliton with new results on the
  quantization of the toda lattice.
\newblock {\em Rocky Mountain J. Math.}, 8:413--430, 1978.

\bibitem{Sut04}
B.~Sutherland.
\newblock {\em {Beautiful Models. 70 Years of Exactly Solved Quantum Many-Body
  Problems.}}
\newblock World Scientific Publishing, 2004.

\bibitem{Ta15}
K.~Takeuchi.
\newblock {Experimental approaches to universal out-of-equilibrium scaling
  laws: turbulent liquid crystal and other developments}.
\newblock {\em J. Stat. Mech.}, P01006, 2014.

\bibitem{Tod67}
M.~{Toda}.
\newblock {Vibration of a chain with nonlinear interaction}.
\newblock {\em Journal of the Physical Society of Japan}, 22:431, 1967.

\bibitem{TW94}
C.A. Tracy and H.~Widom.
\newblock {Level-spacing distributions and the Airy kernel}.
\newblock {\em Comm. Math. Phys.}, 159:151--174, 1994.

\bibitem{War07}
J.~Warren.
\newblock {Dyson's Brownian motions, intertwining and interlacing}.
\newblock {\em Electron. J. Probab.}, 12:573--590, 2007.

\bibitem{Yang67}
C.~N. Yang.
\newblock Some exact results for the many-body problem in one dimension with
  repulsive delta-function interaction.
\newblock {\em Phys. Rev. Lett.}, 19:1312--1315, 1967.

\end{thebibliography}

\end{document}